\documentclass[11pt, oneside,letterpaper]{article}
\usepackage[left=1in, bottom=1in, right=1in, top=1in]{geometry}
\usepackage{setspace}

\usepackage{graphicx}
\usepackage{lmodern}				
\usepackage{amssymb,amsmath,amsthm}
\usepackage{changepage}
\usepackage{color}
\usepackage{pdflscape}
\usepackage{rotating}
\usepackage[toc,page,titletoc]{appendix}
    \usepackage{natbib}
\usepackage[colorlinks=true,
            linkcolor=black,
            urlcolor=blue,
            citecolor=black]{hyperref}
\usepackage{grffile}
\usepackage{lineno}
\numberwithin{equation}{section}
\numberwithin{equation}{section}
	
\DeclareMathOperator*{\argmin}{arg\,min}	
\DeclareMathOperator*{\Mat}{Mat}

\DeclareMathOperator*{\diag}{diag}

\usepackage{wrapfig}
\usepackage{enumerate}

\usepackage{booktabs}
\usepackage{textgreek}
\usepackage{threeparttable}
\usepackage{float}
\usepackage{pdfpages}
\usepackage{thmtools}

\newtheorem{theorem}{Theorem}[]

\newtheorem{lemma}{Lemma}[section]
\newtheorem{assumption}{Assumption}[]
\newtheorem*{assumptionm}{Assumption M}
\newtheorem*{assumptiont}{Assumption T}
\newtheorem*{dassumption}{Assumption DA}
\newtheorem*{kappaassumption}{Assumption K}
\newtheorem{msassumption}{Assumption GMS}[]
\newtheorem*{msassumption6}{Assumption GMS 6}
\newtheorem*{msassumption7}{Assumption GMS 7}

\declaretheoremstyle[
  bodyfont=\normalfont\itshape,
  headformat=\NAME\NUMBER
]{nospacetheorem}
\declaretheorem[style=nospacetheorem,name=Assumption LA]{lassumption}

\declaretheoremstyle[
  bodyfont=\normalfont\itshape,
  headformat=\NAME\NUMBER
]{nospacetheorem}

\usepackage[utf8]{inputenc}
\usepackage[english]{babel}

\newtheorem{definition}{Definition}
\newtheorem{remark}{Remark}[]

\usepackage[inline]{enumitem}
\usepackage{enumerate}
\title{\textbf{Inference for Moment Inequalities: A Constrained Moment Selection Procedure}}
\author{Rami V. Tabri\footnote{School of Economics, The University of Sydney, Sydney, New South Wales 2006, Australia, Tel: +61 2 9351 3092, Fax: +61 2 9351 4341, Email: rami.tabri@sydney.edu.au.}\, and Christopher D. Walker\footnote{Corresponding author. Department of Economics, Harvard University, Cambridge, MA 02138, United States of America, Email: cwalker@g.harvard.edu}}

\date{}
\begin{document}
\maketitle
\begin{abstract}
Inference in models where the parameter is defined by moment inequalities is of interest in many areas of economics. This paper develops a new method for improving the performance of generalized moment selection (GMS) testing procedures in finite-samples. The method modifies GMS tests by tilting the empirical distribution in its moment selection step by an amount that maximizes the empirical likelihood subject to the restrictions of the null hypothesis. We characterize sets of population distributions on which a modified GMS test is (i) asymptotically equivalent to its non-modified version to first-order, and (ii) superior to its non-modified version according to local power when the sample size is large enough. An important feature of the proposed modification is that it remains computationally feasible even when the number of moment inequalities is large. We report simulation results that show the modified tests control size well, and have markedly improved local power over their non-modified counterparts.
\end{abstract}
\newpage
\section{Introduction}
Statistical inference in models defined by moment inequalities is a frequently encountered topic in econometrics. Examples of applications include games of entry with multiple equilibria (e.g.,~\citealp{ciliberto2009market}), single/multiple agent optimization problems (e.g.,~\citealp{pakes2015moment}), censored and missing data (e.g.,~ \citealp{manski2002inference,imbens2004confidence}), model selection tests (e.g.,~\citealp{Shi-KL}, and~\citealp{Shi-Hsu}), event-study designs (e.g.,~\citealp{rambachanhonest}), stochastic dominance comparisons (e.g. \citealp{whang_2019}) and New-Keynesian DSGE models (e.g., \citealp{moon2009estimation}). This paper considers inference for a finite-dimensional parameter defined by a finite number of unconditional moment inequalities.

\par We suppose that there exists a true value of the parameter $\theta_0\in\Theta\subseteq\mathbb{R}^{d}$ that satisfies the moment inequality restrictions
\begin{align}\label{eq - moment inequality}
E_{F_{0}}\big(g_{j}(W_i,\theta_{0})\big) \geq 0\quad\text{for}\quad j=1,...,J,
\end{align}
where $\{g_{j}(\cdot,\theta): j=1,...,J\}$ are known real-valued functions, $\{W_i:i\leq n\}$ are independent and identically distributed (i.i.d.) with unknown distribution $F_0,$ and $W_i\in\mathbb{R}^{\text{dim}(W_i)}.$  Under these moment conditions,  the set $\Theta_{I}(F_0)\equiv\left\{\theta\in\Theta: E_{F_{0}}\big(g_{j}(W_i,\theta)\big) \geq 0\;\forall j=1,...,J\right\}$
denotes the so-called \emph{identified set} while any $\theta\in\Theta_{I}(F_0)$ is termed an \emph{identifiable parameter}. Thus, the true value of the parameter might not be uniquely identified by $F_0$ and the economic model.

\par We are interested in confidence sets for $\theta_0$ constructed by test inversion. The test is based on a statistic $T_n,$ for testing individual hypotheses for each $\theta$ that have the form
\begin{align}\label{eq - Test Problem}
H_0:\,\theta\in\Theta_{I}(F_0)\quad\text{versus}\quad H_1:\,\theta\notin\Theta_{I}(F_0).
\end{align}
Inference in this model is challenging because the pointwise limiting null distribution of conventional test statistics are discontinuous in the parameter -- the dependence on the parameter is through the index set of moment inequalities~(\ref{eq - moment inequality}) that are binding. In particular, a moment inequality enters the pointwise asymptotic null distribution of the test statistic $T_n$ whenever it holds as an equality. Tests of~(\ref{eq - Test Problem}) that have good properties incorporate information about which moments $E_{F_{0}}\big(g_{j}(W_i,\theta)\big)$ are ``positive'', in order to exclude them from the computation of a critical value. Tests of this sort are known as two-step procedures in the literature, examples of which include~\cite{andrews2010inference},~\cite{canay2010inference},~\cite{andrews2012inference}, and~\cite{romano2014practical}. The first step of those testing procedures use the data to determine whether the moment inequalities~(\ref{eq - moment inequality}) are close to or far from being equalities. The second step uses the outcome of the first step to yield information about which moment inequalities are ``positive'' when constructing tests of~(\ref{eq - Test Problem}).

\par The literature on two-step tests of~(\ref{eq - Test Problem}) is vast, and almost all of these tests use the sample-analogue estimator of the moments $E_{F_{0}}\big(g_{j}(W_i,\theta)\big)$ in the first step to determine the slackness of the moment inequalities. This feature ignores the information present in the restrictions~(\ref{eq - moment inequality}), because the sample-analogue estimator does not exploit the fact that the moments satisfy these restrictions under the null hypothesis in~(\ref{eq - Test Problem}). Thus, we conjecture that implementing this information in such tests can improve their accuracy in finite-samples under the null and alternative hypotheses. This paper provides such a modification for the broad class of \emph{generalized moment selection} (GMS) testing procedures put forward by \cite{andrews2010inference}, and finds that our conjecture is in the right direction.

\par We propose a modification of GMS testing procedures that implements the information present in~(\ref{eq - moment inequality}) using the method of empirical likelihood (\citealp{owen2001empirical}). The modification is to replace the sample-analogue estimator of the moments in the first step of the GMS procedure with its constrained empirical likelihood counterpart, where the constraints are the moment inequalities~(\ref{eq - moment inequality}). We label this modification \emph{constrained moment selection} (CMS). For a given test statistic and moment selection function, the CMS and GMS tests only differ in terms of which moments they select for the computation of the critical value in tests of~(\ref{eq - Test Problem}). The motivation for our proposal is that the detection of the ``positive'' moment inequalities in the first step would be more accurate because we are using additional information that is available to us, which the sample-analogue estimator of the moments ignores. Consequently, the CMS procedure alters the GMS critical value for testing~(\ref{eq - Test Problem}) in a data-dependent way that incorporates the information contained in~(\ref{eq - moment inequality}) through a reduction of the parameter space for $F_0.$ For this reason, we expect CMS tests of~(\ref{eq - Test Problem}) to be more accurate than their GMS counterparts in finite-samples.

\par This paper characterises the parameter space for $(\theta_0,F_0)$ over which the CMS and GMS testing procedures are asymptotically equivalent, to first-order, under the null, local alternatives, and distant alternatives. We focus, though, on the GMS class of testing procedures in which the moment selection function is given by the moment selection $t$-test. This focus is without loss of generality, as the results extend naturally, with appropriate modifications, to the more general setup in~\cite{andrews2010inference} using their assumptions. This means that for a given test statistic, CMS tests inherit all of the asymptotic properties of GMS tests. Specifically, under the null, CMS confidence sets are asymptotically valid with uniformity over the parameter space, not asymptotically conservative, and not asymptotically similar. Furthermore, CMS tests of~(\ref{eq - Test Problem}) have greater asymptotic local power than tests based on subsampling or fixed asymptotic critical values, and are consistent against distant alternatives. The parameter space imposes only three conditions in addition to the conditions that define the parameter space~\cite{andrews2010inference} introduce. These conditions are part of Assumption~GEL in~\cite{andrews2009validity}: (i) a uniform bound on the variances of the moment functions, (ii) a lower bound on the determinant of their correlation matrix, and (iii) a regularity condition on an estimator of the degree of slackness of the moments arising from the dual formulation of the constrained empirical likelihood problem. Collectively, the conditions that define our parameter space enables the use of results from~\cite{andrews2009validity} on constrained empirical likelihood estimation in our proofs of the aforementioned asymptotic results.

\par While GMS and CMS are asymptotically equivalent procedures, we characterise local alternatives under which the power of CMS tests dominate their GMS counterparts for sufficiently large, but finite, samples. These are directions in the alternative that have some non-violated moment inequalities (SNVIs) and a non-negative correlational structure. That is, configurations where some of the moments $E_{F_{0}}\big(g_{j}(W_i,\theta_{0})\big)$ under the alternative hypothesis are ``positive", and the covariance matrix of $\{g_{j}(W_i,\theta): j=1,...,J\}$ has non-negative entries only. The non-negative correlational structure arises in empirical applications; see, for example,~\cite{Lok-Tabri-inpress} who point to that structure for moment inequalities characterising stochastic dominance comparisons. It is quite difficult to determine the extent of this difference in local powers analytically. However, using a Monte Carlo simulation experimental design based on~\cite{andrews2012inference}, who focus on finite-sample comparisons of the maximum null rejection probability (MNRP), we show using the modified method of moments (MMM) statistic that along such local alternatives the differences in MNRP-corrected powers of CMS and GMS tests can be approximately 36 percentage points when $J=4$ and $n=250,$ which is strikingly large. See Section~\ref{Setion - MC Sims} for more details.

\par The two-step tests in this literature that exploit the information~(\ref{eq - moment inequality}) are the procedures put forward by~\cite{andrews2009validity} and~\cite{canay2010inference}. They implement this information using (generalised) empirical likelihood.~\cite{andrews2009validity} and~\cite{canay2010inference} develop subsampling and bootstrap tests of~(\ref{eq - Test Problem}), respectively, using empirical-likelihood-type test statistics. Both tests have correct asymptotic size in a uniform sense and are shown not to be asymptotically conservative. However,~\cite{canay2010inference}'s test has higher asymptotic power because it is a GMS procedure. More generally,~\cite{andrews2010inference} show the asymptotic power of GMS tests dominate that of subsampling and plug-in asymptotic tests. A disadvantage of~\cite{canay2010inference}'s procedure is that it may be more computationally burdensome than other GMS tests. Thus, our modification of GMS tests can improve finite-sample performance without incurring a high computational cost.

\par \cite{andrews2012inference} proposed a refinement of GMS termed \emph{refined moment selection} (RMS) and discussed the reasons why such an approach is preferable. However, the RMS procedure is quite computationally expensive when $J>10.$ By contrast, the CMS procedure remains computationally feasible when $J$ is large. The reason is that the constrained empirical likelihood optimization problem it is based upon has a strictly concave objective function, convex feasible set, and the choice variables enter linearly into the constraints. As a consequence, there is a unique global solution to this optimization problem and its implementation involves an of-the-shelf programming routine. More recently,~\cite{romano2014practical} proposed a two-step testing procedure for moment inequalities that is similar in spirit to the RMS procedure and remains computationally feasible when $J$ is large. An important distinction between the CMS testing procedure and these tests is that, like GMS tests, neither of them exploits the information present in the moment inequality constraints~(\ref{eq - moment inequality}), because they employ the sample-analogue estimator of the moments in their first step.

\par We examine the finite-sample performance of CMS tests using the MMM and adjusted quasi-likelihood-ratio (AQLR) test statistics in Monte Carlo simulations based on the experimental design in~\cite{andrews2012inference}. The experiment compares the performance of CMS to its GMS, RMS, RSW counterparts in terms of MNRP and MNRP-corrected local power. The inclusion of the RMS and RSW procedures in the simulation experiment is to benchmark the performance of CMS. Overall, the simulation results showcase the value of implementing the information~(\ref{eq - moment inequality}) in the CMS procedure in terms of finite-sample size and power properties, and corroborate its theoretical superior performance over GMS. The simulation results also show the performance of CMS and RMS tests based on the AQLR statistic are comparable. This finding is encouraging as the RMS test has desirable asymptotic properties but can be computationally expensive when $J$ is large, while the CMS procedure isn't costly to compute at all.

\par The idea of exploiting information on parameters defined by constraints for improving performance in statistical problems, through constrained estimation, is one of the most natural ideas in statistics. The literature on constrained estimation via tilting the empirical distribution overlaps with this paper, where the problem is that the constraints/information are not adequately reflected by the empirical distribution (e.g.,~\citealp{Hall-Presnell}). Tilting the empirical distribution allows one to incorporate information selectively into a statistical procedure without changing the procedure itself.~\cite{Lok-Tabri-inpress}
apply this idea to modifying two-step bootstrap tests for restricted stochastic dominance orderings using empirical likelihood and semi-infinite programming. The parameter of interest in their setup is infinite-dimensional and there is a continuum of moment inequality restrictions, which are defined by moment functions that have a particular form. The form of the moment functions in their setup yields a correlational structure that facilitates the analysis of such moment inequalities. Contrastingly, in the setup of this paper, $J$ is finite and the form of the moment functions $\{g_{j}(\cdot,\theta): j=1,...,J\}$ is arbitrary. The implementation of empirical likelihood in their setup has a data-driven number of inequalities that increases with the sample size, which can be as large as 500 in moderate sample sizes. The ability of empirical likelihood to straightforwardly execute with a large number of moment inequality restrictions transfers to the CMS procedure for models with large $J.$ This computational feasibility of CMS is an important feature of our approach. Similar to~\cite{Lok-Tabri-inpress}, this paper is also part of the econometrics literature on shape restrictions (e.g.,~\citealp{Chetverikov-Santos-Shaikh}, and the references therein), as the inequalities~(\ref{eq - moment inequality}) can be thought of as finite-dimensional analogues of shape restrictions on nonparametric functions.

\par We organize the paper as follows. Section~\ref{Section - Steup and Background} introduces the statistical framework, as well as the GMS and CMS procedures. Section~\ref{Section - CMS} introduces the main results of the paper. Section~\ref{Setion - MC Sims} reports the results of Monte Carlo simulations, and Section~\ref{Section - Conclusion} concludes.

\par For notational simplicity, throughout the paper we write partitioned column vectors as $h=(h_1,h_2).$ rather than $h=(h_{1}',h_{2}')'$. Let $\mathbb{R}_{+}=\{x\in\mathbb{R}:x\geq0\}$, $\mathbb{R}_{+,\infty}=\mathbb{R}_{+} \cup\{+\infty\}$,
$\mathbb{R}_{[+\infty]}=\mathbb{R} \cup\{+\infty\}$,$\mathbb{R}_{[\pm\infty]}=\mathbb{R} \cup\{\pm\infty\}$, ``$:=$" denote the definitional identity, and $\overline{A}$ denote the closure of a set $A$.

\section{Setup}\label{Section - Steup and Background}

\subsection{Moment Inequality Model and Test statistic}
The object of interest is a parameter $\theta_{0} \in \Theta \subseteq \mathbb{R}^{d}$, $d<+\infty$, defined by a finite number of known moment functions $g_{j}: \mathcal{W} \times \Theta \rightarrow \mathbb{R}$ that satisfy the following unconditional moment inequality restrictions:
\begin{align}\label{Moment inequality restrictions S3}
E_{F_{0}}\big(g_{j}(W,\theta_{0})\big) \geq 0 \ \forall \ j\in\mathcal{J},
\end{align}
where $F_{0}$ denotes the true distribution of the observed data $W$ and $\mathcal{J}:=\{1,...,J\}$ with $J<\infty$. In general, the identified set, $\Theta_{I}(F_0)=\{\theta \in \Theta: \ E_{F_0}(g_{j}(W,\theta)) \geq 0 \ \forall \ j \in \mathcal{J}\}$, is not a singleton meaning that the parameter is partially identified.

\par The moment inequality model is given by the following definition.
\begin{definition}\label{def1}[\textbf{Moment Inequality Model}] Let $\mathcal{F}$ be the set of parameters $(\theta, F)$ that satisfy:
\begin{enumerate}
\item $\theta \in \Theta\subseteq \mathbb{R}^{d}$. \label{A1}
\item $\{W_{i}: i \geq 1\}$ are i.i.d. under $F$. \label{A2}
\item $E_{F}\big(g_{j}(W_{i},\theta)\big) \geq 0 \ \text{for} \ j\in\mathcal{J}$. \label{A3}
\item $\sigma^{2}_{F,j}(\theta) := Var_{F}\big(g_{j}(W_{i},\theta)\big) \in [\varepsilon_{*}, M_{*}]$ for some $M_{*}>\varepsilon_{*}>0$. \label{A4}
\item $\Omega(\theta,F) \in \varPsi_{2}$, where $\Omega(\theta,F)$ is the $J\times J$ correlation matrix of $\{g_{j}(W_{i},\theta),j=1,\ldots,J\},$ and $\varPsi_{2}$ is the space of correlation matrices whose determinant is greater than $\varepsilon>0$. \label{A5}
\item $\exists \delta$ and $M >0:$ $E_{F}|g_{j}(W_{i},\theta)/\sigma_{F,j}(\theta)|^{2+\delta} \leq M \ \forall \ j\in\mathcal{J}$. \label{A6}
\end{enumerate}
\end{definition}
\noindent All of the conditions in this definition, except for Conditions~\ref{A4} and~\ref{A5}, are those presented in (2.2) of~\citet{andrews2010inference}. Condition~\ref{A4} is a strengthening of Condition (v) in~\citet{andrews2010inference} so that the variances of the moment functions are uniformly bounded. Condition~\ref{A5} specifies the nonsingularity of the matrix $\Omega(\theta,F).$  These conditions are relatively unrestrictive and are part of Assumption GEL in~\citet{andrews2009validity}. Furthermore, they arise frequently in papers that consider empirical likelihood inference for moment inequalities (e.g.,~\citealp{canay2010inference}, and~\citealp{Lok-Tabri-inpress}).

\par For a given value of the parameter, $\theta=\theta_0,$ we invert tests of the hypothesis $H_0:\,\theta_0\in\Theta_{I}(F_0)$ to construct confidence sets of the form $CS_{n}=\{\theta \in \Theta: T_{n}(\theta) \leq c_{1-\alpha}(\theta)\},$ where $T_{n}(\theta)$ denotes a test statistic and $c_{1-\alpha}(\theta)$ is a critical value for tests with nominal level $\alpha \in (0, 1/2)$. We say $CS_{n}$ is a uniformly valid confidence set for $\theta$ if
\begin{align}\label{Uniform Validity}
\liminf_{n\rightarrow+\infty}\inf_{(\theta,F)\in\mathcal{F}}P_{F}(T_{n}(\theta) \leq c_{1-\alpha}(\theta))\geq 1-\alpha,
\end{align}
where $P_{F}(\cdot)$ is the probability measure induced by repeated sampling from $F$. Uniformity is essential in order for asymptotic size to be a good approximation to the finite-sample size of confidence sets, because the test statistic exhibits a discontinuity in its asymptotic distribution (as a function of the distribution generating the data), but not in its finite-sample distribution. Discontinuities of this type can create asymptotic size problems that are analogous to those that arise with parameters that are near a boundary (e.g.,~\citealp{andrews2009validity}).

\par A test statistic is a function $S: \mathbb{R}_{[+\infty]}^{J}\times \mathcal{V}_{J\times J} \rightarrow \mathbb{R}$ given by $T_{n}(\theta_0):=S\left(n^{\frac{1}{2}}\hat{g}_{n}(\theta_0), \hat{\Sigma}_{n}(\theta_0)\right),$ where $\mathcal{V}_{J\times J}$ is the set of invertible $J\times J$ variance matrices,
\begin{align*}
\hat{g}_{n}(\theta_0) &  = \left[\frac{1}{n}\sum_{i=1}^{n}g_{1}(W_{i},\theta_0),...,\frac{1}{n}\sum_{i=1}^{n}g_{J}(W_{i},\theta_0)\right]^{\top},\;g(W_{i},\theta_0)=\big[g_{1}(W_{i},\theta_0),...,g_{J}(W_{i},\theta_0)\big]^{\top},\\
\text{and} & \quad \hat{\Sigma}_{n}(\theta_0)= n^{-1}\sum_{i=1}^{n}(g(W_{i},\theta_0)-\hat{g}_{n}(\theta))(g(W_{i},\theta_0)-\hat{g}_{n}(\theta_0))^{\top}.
\end{align*}
 Two examples are the modified method of moments (MMM) and adjusted quasi-likelihood-ratio (AQLR) statistics. In the context of the moment inequality model $\mathcal{F}$ given by Definition~\ref{def1}, these test statistics are defined as
\begin{align}
S_{1}\big(n^{\frac{1}{2}}\hat{g}_{n}(\theta_0), \hat{\Sigma}_{n}(\theta_0)\big) &  = n\sum_{j=1}^{J}\Big(\min\big\{0, \hat{g}_{n,j}(\theta_0)/\hat{\sigma}_{n,j}(\theta_0)\big\}\Big)^{2}\;\text{and} \label{MMM} \\
S_{2A}\big(n^{\frac{1}{2}}\hat{g}_{n}(\theta_0), \hat{\Sigma}_{n}(\theta_0)\big) & = n\inf_{t \in \mathbb{R}_{+,\infty}^{J}}(\hat{g}_{n}(\theta_0)-t)^{\intercal}\big(\tilde{\Sigma}_{n}(\theta_0)\big)^{-1}(\hat{g}_{n}(\theta_0)-t),\label{QLR}
\end{align}
respectively, where $\tilde{\Sigma}_{n}(\theta_0) =  \hat{\Sigma}_{n}(\theta_0)+\max\{0, 0.012-|\hat{\Omega}_{n}(\theta_0)|\}\hat{D}_{n}(\theta_0),\; \hat{\Omega}_{n}(\theta_{0})= \hat{D}_{n}^{-\frac{1}{2}}(\theta_{0})\hat{\Sigma}_{n}(\theta_{0})\hat{D}_{n}^{-\frac{1}{2}}(\theta_{0})$ and $\hat{D}_{n}(\theta_{0})=\diag \hat{\Sigma}_{n}(\theta_{0})$, where $\diag \hat{\Sigma}_{n}(\theta_{0})$ is a diagonal matrix with dimensions equal to those of $\hat{\Sigma}_{n}(\theta_{0})$ whose diagonal elements equal those of $\hat{\Sigma}_{n}(\theta_{0})$.

\subsection{GMS and CMS Procedures}
\par The point of departure for establishing that~(\ref{Uniform Validity}) holds for the GMS procedure is to consider the asymptotic distribution of $T_n(\theta_0)$ under a suitable sequence of null distributions. For any sequence $\{F_n: n\geq1\}$ in the model of the null hypothesis, the test statistic satisfies
\begin{align}\label{eq - Asy GMS}
T_n(\theta_0)\stackrel{d}{\longrightarrow}S\left(\Omega_0^{1/2}Z^{*}+h_1,\Omega_0\right)\quad\text{where}\quad Z^{*}\sim N(0_J,I_J),
\end{align}
where $h_1\in\mathbb{R}^{J}_{+,\infty},$ and $\Omega_0$ is a $J\times J$ correlation matrix.\footnote{Specifically, this large-sample result~(\ref{eq - Asy GMS}) follows from the form of the test statistic, the Central Limit Theorem, and the convergence in probability of the sample correlation matrix.} The vector $h_1=(h_{1,1},...,h_{1,J})'$ has elements given by $\lim_{n\rightarrow+\infty}n^{1/2}(E_{F_n}(g_{j}(W_{i},\theta_0))/ \sigma_{F,j}(\theta_0)$ and measures the degree of slackness of the moment inequalities. The crux of this asymptotic construction is that the limiting distribution in~(\ref{eq - Asy GMS}) now depends continuously on the degree of slackness of the moment inequalities via the parameter $h_1,$ which reflects the finite-sample situation.

\par The asymptotic implementation of the GMS critical value is the $1-\alpha$ quantile of a data-dependent version of the asymptotic null distribution in~(\ref{eq - Asy GMS}). It replaces $\Omega_0$ by a consistent estimator and replaces $h_1$  with a function $\varphi : \mathbb{R}^{J} \times \varPsi_{2} \rightarrow \mathbb{R}^{J}_{+,\infty}$, which measures the slackness of moment inequalities through $\hat{\xi}_{n}(\theta_{0})=\kappa_{n}^{-1}n^{\frac{1}{2}}\hat{D}_{n}^{-\frac{1}{2}}(\theta_{0})\hat{g}_{n}(\theta_{0}),$ where $\{\kappa_{n}: n \geq 1\}$ is a divergent sequence of scalars \citep{andrews2010inference}. The GMS critical value, $\hat{c}_{n}(\theta_{0},1-\alpha),$ is the $1-\alpha$ quantile of
\begin{align}\label{eq - L_n GMS}
L_n(\theta_0,Z^{*})=S\left(\hat{\Omega}_{n}^{\frac{1}{2}}(\theta_{0})Z^{*}+\varphi(\hat{\xi}_{n}(\theta_{0}), \hat{\Omega}_{n}(\theta_{0})), \hat{\Omega}_{n}(\theta_{0})\right),
\end{align}
where $Z^{*}\sim N(0_{J}, I_{J})$ and is independent of $\{W_i: i\geq 1\}.$ That is,
\begin{align}\label{GMScv}
\hat{c}_{n}(\theta_{0},1-\alpha) := \inf\bigg\{x \in \mathbb{R}: P\Big(L_n(\theta,Z^{*})\leq x\Big) \geq 1-\alpha \bigg\}.
\end{align}
where $P\Big(L_n(\theta_0,Z^{*})\leq x\Big)$ denotes the conditional CDF at $x$ of $L_n(\theta_0,Z^{*}),$ conditional upon $(\hat{\xi}_{n}(\theta_{0}), \hat{\Omega}_{n}(\theta_{0})).$ In practice,
the calculation of $\hat{c}_{n}(\theta_{0},1-\alpha)$ is by simulating $L_n(\theta_0,Z^{*})$ using $R$ i.i.d. draws from $Z^{*}\sim N(0_{J}, I_{J})$ and computing the $1-\alpha$ quantile of the empirical CDF from $\{L_n(\theta_0,Z_r^{*}): r =1,...,R\}.$

\par Alternatively, one may compute the GMS critical value using the bootstrap. We briefly describe this approach. Let $\{W_{i}^{*}: i \leq n\}$ be a bootstrap sample drawn from the empirical distribution of the data $\{W_{i}: i \leq n\}$, and define $\hat{g}_{n}^{*}(\theta_{0}) = n^{-1}\sum_{i=1}^{n}g(W_{i}^{*},\theta_{0})$, $\hat{\Sigma}_{n}^{*}(\theta_{0})=n^{-1}\sum_{i=1}^{n}(g(W_{i}^{*},\theta_{0})-\hat{g}_{n}^{*}(\theta_{0}))(g(W_{i}^{*},\theta_{0})-\hat{g}_{n}^{*}(\theta_{0}))^{\intercal},$ $\hat{D}_{n}^{*}(\theta_{0}) = \text{diag} \hat{\Sigma}_{n}^{*}(\theta_{0})$, and $\hat{\Omega}_{n}^{*}(\theta_{0}) = (\hat{D}_{n}^{*}(\theta_{0}))^{-\frac{1}{2}}\hat{\Sigma}_{n}^{*}(\theta_{0})(\hat{D}_{n}^{*}(\theta_{0}))^{-\frac{1}{2}}$. The bootstrap implementation of the GMS procedure replaces $L_{n}(\theta_{0},Z^{*})$ in~(\ref{eq - L_n GMS}) with $$L_{n}(\theta_{0},\{W_{i}^{*}: i \leq n\})=S\left(G_{n}^{*}(\theta_{0})+\varphi(\hat{\xi}_{n}(\theta_{0}), \hat{\Omega}_{n}(\theta_{0})), \hat{\Omega}_{n}^{*}(\theta_{0})\right),$$ where $G_{n}^{*}(\theta_{0}) = n^{\frac{1}{2}}(\hat{D}_{n}^{*}(\theta_{0}))^{-\frac{1}{2}}(\hat{g}_{n}^{*}(\theta_{0})-\hat{g}_{n}(\theta_{0})),$ and defines a critical value analogous to (\ref{GMScv}). In practice, this critical value is the empirical $1-\alpha$ quantile of the bootstrap statistics $\{L_{n}(\theta_{0},\{W_{i,r}^{*}: i \leq n\}):r=1,...,R\}$, where $\{\{W_{i,r}^{*}:i\leq n\}:r=1,...,R\}$ are bootstrap samples drawn from the empirical distribution of the data $\{W_{i}: i \leq n\}$. The asymptotic results of this paper hold for the bootstrap provided that $G_{n}^{*}(\theta_{n,h}) \overset{d}{\rightarrow} \Omega_{0}^{\frac{1}{2}}Z^{*}$, where the convergence is conditional on $\{W_i:i\leq n\}$ for almost every sample path, for all sequences $\{(\theta_{n,h},F_{n,h}): n \geq 1\}$ in $\mathcal{F}$.

\par There are numerous choices for $\varphi$ and $\{\kappa_{n}: n \geq 1\}$. \cite{chernozhukov2007estimation} and~\cite{andrews2010inference} recommend using $\kappa_{n} = (\ln n)^{\frac{1}{2}}$. Another option is to set $\kappa_{n} = (2\ln \ln n)^{\frac{1}{2}}$, which is used in~\cite{canay2010inference}.
Our main results set $\varphi=\varphi^{(1)}$, where
\begin{align}\label{t-test MS function}
\varphi^{(1)}_{j}(\xi, \Omega)=
\begin{cases}
0 &\text{if $\xi_{j} \leq 1$} \\
+\infty &\text{if $\xi_{j} > 1$}
\end{cases}
\end{align}
for each $j \in \mathcal{J}$, and is referred to as the `moment selection $t$-test' because it resembles a $t$-test with deterministic critical value $\kappa_{n}$. The decision reflects the recommendations of~\cite{andrews2012inference}, and is essentially without loss of generality because our results extend to any choice of $\varphi$ that satisfies the assumptions of~\cite{andrews2010inference}. Appendix~\ref{OtherGMS} discusses how to generalize our results to other suitable choices of $\varphi$.

\par The advantage of the GMS procedure is that it asymptotically detects the ``positive'' moments  $E_{F_0}(g_{j}(W_{i},\theta_{0}))$ and excludes them from the computation of the critical value, so as to mimic the discontinuity in the asymptotic null distribution of $T_n(\theta).$ This ability of GMS tests to detect such moments is the source of its improvements over the subsampling and plug-in procedures under the null and alternative hypotheses.

\par Although GMS tests are computationally simple and have desirable asymptotic properties, their performance in finite-samples depends crucially on how well they detect the ``positive'' moments, so as to omit them from the computation of the critical value. Their use of the sample-analogue estimator of the moments for detecting the positive moments does not implement the information embedded in~(\ref{Moment inequality restrictions S3}) and implementing this information appropriately can improve the detection accuracy of ``positive" moments in finite-samples.

\par For a given moment selection function $\varphi,$ the CMS procedure implements the information present in~(\ref{Moment inequality restrictions S3}) through a surgical modification of the GMS procedure. The modification is to replace $\hat{g}_{n}(\theta_0)$ with its constrained empirical likelihood counterpart, where the constraints impose the inequality restrictions~(\ref{Moment inequality restrictions S3}). Specifically, CMS replaces $\hat{\xi}_{n}(\theta_{0})$ with $\acute{\xi}_{n}(\theta_{0}) =\kappa_{n}^{-1}n^{\frac{1}{2}}\hat{D}_{n}^{-\frac{1}{2}}(\theta_{0})\acute{g}_{n}(\theta_{0})$, where $\acute{g}_{n}(\theta_{0}) = \sum_{i=1}^{n}\acute{p}_{i}g(W_{i},\theta_{0})$ and the probabilities $\acute{p}_{1},...,\acute{p}_{n}$ solve
\begin{align}\label{eq - EL problem}
\max_{p_{1},...,p_{n}}\Bigg\{\sum_{i=1}^{n}\ln(p_{i}): \sum_{i=1}^{n}p_{i}g_{j}(W_{i},\theta_{0}) \geq 0 \ \forall \ j\in\mathcal{J}, \ \sum_{i=1}^{n}p_{i}=1, \ p_{i} \geq 0 \ \forall \ i\Bigg\},
\end{align}
and then computes a critical value as described in~(\ref{GMScv}), but replaces $\varphi(\hat{\xi}_{n}(\theta_{0}), \hat{\Omega}_{n}(\theta_{0}))$ with $\varphi(\acute{\xi}_{n}(\theta_{0}), \hat{\Omega}_{n}(\theta_{0}))$ in~(\ref{eq - L_n GMS}). The CMS modification of GMS can easily be applied to all choices of $\varphi$ and $\{\kappa_{n}: n \geq 1\}$ presented in \citet{andrews2010inference} because it only replaces $\hat{g}_{n}(\theta_{0})$ with $\acute{g}_{n}(\theta_{0}).$ The estimator $\acute{g}_{n,j}(\theta_{0})$ of $E_{F_{0}}\big(g_{j}(W,\theta_{0})\big)>0$ is more accurate than $\hat{g}_{n,j}(\theta_{0})$ because the optimization problem~(\ref{eq - EL problem}), which gives rise to $\acute{g}_{n}(\theta_{0}),$ imposes a correct constraint $E_{F_{0}}\big(g_{j}(W,\theta_{0})\big)\geq0,$ while $\hat{g}_{n}(\theta_{0})$ ignores such information. Thus, when $E_{F_{0}}\big(g_{j}(W,\theta_{0})\big)>0$ (under the null or alternative), the moment selection function based on $\acute{g}_{n,j}$ detects this configuration more reliably than $\varphi_{j}(\hat{\xi}_{n}, \hat{\Omega}_{n}(\theta_{0})),$ and therefore, takes it into account by delivering a critical value that is suitable for the case where this moment inequality is omitted. This feature of CMS leads to it having better finite-sample properties than GMS under the null and alternative hypotheses.

\par The CMS procedure is not computationally expensive because the empirical likelihood optimization problem~(\ref{eq - EL problem}) has a strictly concave objective function and a convex feasible set that is characterised by affine functions of the choice variables \citep{owen2001empirical}. This means that the optimization problem~(\ref{eq - EL problem}) has a unique global solution, and it can be computed numerically using standard optimization routines in software such as Matlab, R, or GAUSS. This computational simplicity of the optimization problem~(\ref{eq - EL problem}) is an important feature of the CMS procedure.

\begin{remark}
\normalfont
One can `fully constrain' the CMS procedure by using restricted estimators of the correlation matrix. In this case, we evaluate $\varphi(\acute{\xi}_{n}^{FC}(\theta_{0}), \acute{\Omega}_{n}(\theta_{0}))$, where
\begin{align*}
\acute{\xi}_{n}^{FC}(\theta_{0}) & =\kappa_{n}^{-1}n^{\frac{1}{2}}\acute{D}_{n}^{-\frac{1}{2}}(\theta_{0})\acute{g}_{n}(\theta_{0}),\quad\acute{\Omega}_{n}(\theta_{0}) = \acute{D}_{n}^{-\frac{1}{2}}(\theta_{0})\acute{\Sigma}_{n}(\theta_{0})\acute{D}_{n}^{-\frac{1}{2}}(\theta_{0}), \\
\acute{\Sigma}_{n}(\theta_{0}) & = \sum_{i=1}^{n}\acute{p}_{i}(g(W_{i},\theta_{0})-\acute{g}_{n}(\theta_{0}))(g(W_{i},\theta_{0})-\acute{g}_{n}(\theta_{0}))^{\top},\;\text{and}\;
\acute{D}_{n}^{-\frac{1}{2}}(\theta_{0})= \diag \acute{\Sigma}_{n}(\theta_{0}).
\end{align*}
In our simulations not presented in this paper, we find limited practical difference between the $\varphi(\acute{\xi}_{n}^{FC}(\theta_{0}), \acute{\Omega}_{n}(\theta_{0}))$ and $\varphi(\acute{\xi}_{n}(\theta_{0}), \hat{\Omega}_{n}(\theta_{0}))$. Consequently, the rest of the paper focuses on $\varphi(\acute{\xi}_{n}(\theta_{0}), \hat{\Omega}_{n}(\theta_{0}))$ because it is simpler to show that there are power advantages over GMS.
\end{remark}

\section{Main Results}\label{Section - CMS}

\par We start by introducing the assumptions that beget the main results of this paper. They are conditions on the test statistic $S$, the moment selection function $\varphi$, and the parameter space $\mathcal{F}$. The assumptions on $S$ we consider are from~\cite{andrews2010inference}, and are stated as Assumptions 1-7 in Appendix~\ref{Section - Test Stat Assump} for ease of exposition. Recall that we set $\varphi=\varphi^{(1)}$ in~(\ref{t-test MS function}), and the main results we present are based on this choice of moment selection function. It should be noted that this choice of $\varphi$ is without loss of generality as one can employ assumptions identical to those in~\cite{andrews2010inference} on $\varphi$ to deduce the same conclusions, because the CMS procedure does not alter the moment selection function in the GMS procedure. See Appendix~\ref{OtherGMS} for the details on other choices of $\varphi.$

\par The first assumption concerns the sequence $\{\kappa_n: n\geq1\}.$
\begin{kappaassumption}
\begin{enumerate*}
\item $\kappa_{n} \rightarrow +\infty$ as $n\rightarrow +\infty$.\label{Assump - Kappa 1}
\item $\kappa_{n}^{-1}n^{\frac{1}{2}} \rightarrow +\infty$ as $n\rightarrow +\infty$. \label{Assump - Kappa 2}
\end{enumerate*}
\end{kappaassumption}
\noindent The conditions in this assumption are not restrictive -- the aforementioned examples of $\{\kappa_{n}: n \geq 1\}$ satisfy them. The `optimal' choice of $\{\kappa_{n}: n \geq 1\}$ is an important question, but the goal of our paper is more modest: to demonstrate how incorporating statistical information can improve finite-sample inference for moment inequalities in a computationally simple way and, for this purpose, our analysis conditions on an arbitrary choice of $\{\kappa_{n}: n \geq 1\}$. For our Monte Carlo experiment (Section \ref{Setion - MC Sims}), we set $\kappa_{n}=(\ln n)^{\frac{1}{2}}$ which is the recommended choice in \cite{chernozhukov2007estimation} and \cite{andrews2010inference}.

\par The next assumption we present is the first part in Part (d) of Assumption GEL in~\cite{andrews2009validity}. It is helpful in establishing that $\acute{g}_{n}(\theta_0)$ is a uniformly consistent estimator of the moments under the null hypothesis $H_0:\theta_0\in\Theta_{I}(F_0).$ To introduce this assumption, for each $t\in \mathbb{R}^{J},$ define $g_{i}(t,\theta) = g(W_{i},\theta) - t.$ The vector $t$ is a nuisance parameter that captures the slackness of the moment inequalities. Using the dual formulation of the empirical likelihood problem~(\ref{eq - EL problem}), the amount of slackness is captured by $\acute{t}_n=\argmin_{t \in \mathbb{R}_{+}^{J}}\sup_{\lambda \in \acute{\Lambda}_n(t,\theta)}n^{-1}\sum_{i=1}^{n}\ln\Big(1-\lambda^{\top}g_{i}(t,\theta)\Big),$ where $\acute{\Lambda}_{n}(t,\theta)=\{\lambda\in\mathbb{R}^{J}:\lambda^{\top}g_{i}(t,\theta)\in Q\,\forall i=1,\ldots,n\},$ $Q$ is an open interval of $\mathbb{R}$ containing $0.$ This reformulation of the empirical likelihood problem~(\ref{eq - EL problem}) is feasible because the linear constraint qualification applies to it. The part of Assumption GEL we include in our setup is a regularity condition concerning the uniform asymptotic behavior of $\acute{t},$ and is stated in terms of the following reparametrization of $\mathcal{F}.$
\begin{definition}\label{def Gamma} Let $\Gamma$ be defined as the set of all $\gamma=(\gamma_1,\gamma_2,\gamma_3)$ such that for some $(\theta, F)\in\mathcal{F}$ where
\begin{enumerate}
\item $\mathcal{F}$ is defined in Definition~\ref{def1}.
\item $\gamma_1=(E_{F}(g_{1}(W_{i},\theta))/ \sigma_{F,1}(\theta),\ldots,E_{F}(g_{J}(W_{i},\theta))/\sigma_{F,J}(\theta)).$
\item $\gamma_2= \big(\theta, \text{vech}_{*}(\Omega(\theta,F))\big),$ where $\text{vech}_{*}(\Omega(\theta,F))$ is the vector of lower off-diagonal elements of $\Omega(\theta,F).$
\item $\gamma_{3} = F.$
\end{enumerate}
\end{definition}
\noindent~\cite{andrews2010inference} indicate that there is a one-to-one mapping from $\gamma$ to $(\theta,F);$ see Appendix~A of their paper for the details.  Denote by $\{\gamma_{n,h}: n \geq 1\}\subset \Gamma$ a sequence of parameters in $\Gamma$ such that $n^{1/2}\gamma_{n,h,1} \rightarrow h_{1}\in \mathbb{R}^{J}_{+,\infty}$ and $\gamma_{n,h,2} \rightarrow h_{2} \in \mathbb{R}_{[\pm\infty]}^{q}$ as $n\rightarrow\infty,$ where $q =\dim(\Theta)+ \dim(\text{vech}_{*}(\Omega(\theta,F))).$ The part of Assumption GEL that we include in our setup is given by the following assumption.
\begin{assumptiont}\label{AT}
For all subsequences $\{w_n\}$ of $\{n\}$  and all sequence $\{\gamma_{w_{n},h}: n \geq 1\}\subset \Gamma$ and corresponding $\{(\theta_{w_{n},h}, F_{w_{n},h}): n \geq 1\}\subset\mathcal{F},$
$$\acute{t}_{w_n}=\argmin_{t \in \mathbb{R}_{+}^{J}}\sup_{\lambda \in \acute{\Lambda}_{w_n}(t,\theta_{w_{n},h})}\frac{1}{n}\sum_{i=1}^{n}\ln\Big(1-\lambda^{\top}g_{i}(t,\theta_{w_{n},h})\Big)$$
exists and satisfies $\sup_{n\geq 1}||\acute{t}_{w_n}||_{\ell^{2}_{J}} \leq K$ with probability approaching 1 as $n\rightarrow+\infty$ for some constant $K<+\infty,$ where $||\cdot||_{\ell^{2}_{J}}$ is the usual Euclidean norm on $\mathbb{R}^{J}$.
\end{assumptiont}

\par We also include an assumption from~\cite{andrews2010inference} for the case in which $\text{Int}(\Theta_{I}(F_0)) \neq \emptyset$ for some data-generating process in the model. It is required to show that when there are no binding moment inequalities, the maximum asymptotic coverage probability is equal to $1.$
\begin{assumptionm}
There exists $(\theta, F) \in \mathcal{F}$ that satisfies $E_{F}\big(g_{j}(W_{i},\theta)\big)>0$ for all $j\in \mathcal{J}$.
\end{assumptionm}

\subsection{Asymptotic Size Results}\label{Subsection Asymptotic Size results}
We now present the first main result of the paper. It mirrors Theorem~1 of~\cite{andrews2010inference} which concerns the asymptotic size of GMS confidence sets. Denote by $\acute{c}_{n}(\theta_0,1-\alpha)$ the CMS critical value under the nominal level $1-\alpha$ for testing the null hypothesis $H_0:\theta_0\in\Theta_{I}(F_0).$
\begin{theorem}\label{asymptoticsize}
Suppose $S$ satisfies Assumptions~\ref{Test1} -~\ref{T2}, $\varphi=\varphi^{(1)}$ in~(\ref{t-test MS function}), the sequence $\{\kappa_n:n\geq1\}$ satisfies Part 1 of Assumption~K, and $\alpha \in (0,1/2)$. Furthermore, let $\mathcal{F}_{+}=\{(\theta, F) \in \mathcal{F}: \text{Assumption T holds}\},$ and $\mathcal{F}_{++}=\{(\theta, F) \in \mathcal{F}: \text{Assumptions T and M hold}\}.$ Then, the nominal level $(1-\alpha)$ CMS confidence set based on the statistic $T_n(\theta)$ satisfies the following statements:
\begin{enumerate}
\item $\liminf_{n\rightarrow+\infty}\inf_{(\theta, F) \in \mathcal{F}_+} P_{F}\Big(T_{n}(\theta) \leq \acute{c}_{n}(\theta,1-\alpha)\Big)\geq1-\alpha.$ \label{result1}
\item $\liminf_{n\rightarrow+\infty}\inf_{(\theta, F) \in \mathcal{F}_+} P_{F}\Big(T_{n}(\theta) \leq \acute{c}_{n}(\theta,1-\alpha)\Big)=1-\alpha$, if in addition $S$ and $\{\kappa_n:n\geq1\}$ satisfy Assumption~\ref{Assump test stat 7} and Part 2 of Assumption~K, respectively. \label{result2}
\item $\limsup_{n\rightarrow+\infty}\sup_{(\theta, F) \in \mathcal{F}_{++}}P_{F}\Big(T_{n}(\theta) \leq \acute{c}_{n}(\theta, 1-\alpha)\Big)=1.$  \label{result3}
\end{enumerate}
\begin{proof}
See Appendix~\ref{Appendix - proof Thm 1}.
\end{proof}
\end{theorem}
The first result of Theorem~\ref{asymptoticsize} establishes the uniform validity of CMS confidence sets over the parameter space $\mathcal{F}_{+},$ and the second result of this theorem shows that they are not asymptotically conservative. The third result shows the maximum coverage probability of CMS confidence sets is equal to 1 over the parameter space $\mathcal{F}_{++}.$ The parameter space $\mathcal{F}_+$ is a subset of the one used in Theorem~1 of~\cite{andrews2010inference}, because it imposes Assumption~T and Condition 4 in Definition~\ref{def1} in addition to the conditions they set for their parameter space.

% I applied Jon's comment here.
\par The proof of Theorem~\ref{asymptoticsize} establishes that CMS and GMS procedures are asymptotically equivalent with uniformity over the parameter space $\mathcal{F}_+.$ The essence of this result is that for every sequence $\{\gamma_{n,h}: n \geq 1\}\subset \Gamma$ such that Assumption~T holds, along the corresponding sequence $\{(\theta_{w_{n},h}, F_{w_{n}}): n \geq 1\}\subset\mathcal{F}_+$ we have $\acute{\xi}_{w_{n}}(\theta_{w_{n},h})=\hat{\xi}_{w_{n}}(\theta_{w_{n},h})+o_{p}(1).$ This asymptotic equivalence is a consequence of
$\acute{g}_{w_{n}}(\theta_{w_{n},h})-\hat{g}_{w_{n}}(\theta_{w_{n},h})=O_{p}(w_{n}^{-\frac{1}{2}})$ (see Lemma~\ref{L2CS} in Appendix~\ref{Section - Technical Lemmas for Thm 1}) and Assumption K on $\kappa_{w_{n}}.$
In particular, these arguments are used after re-writing the expression of $\acute{\xi}_{w_{n}}(\theta_{w_{n},h})$
in terms of $\hat{\xi}_{w_{n}}(\theta_{w_{n},h}),$ as such
\begin{equation}
\begin{aligned}
\acute{\xi}_{w_{n}}(\theta_{w_{n},h}) & =\kappa_{w_{n}}^{-1}w_{n}^{\frac{1}{2}}\hat{D}_{w_{n}}^{-\frac{1}{2}}(\theta_{w_{n},h})\acute{g}_{w_{n}}(\theta_{w_{n},h})\\
& =\kappa_{w_{n}}^{-1}w_{n}^{\frac{1}{2}}\hat{D}_{w_{n}}^{-\frac{1}{2}}(\theta_{w_{n},h})\big(\acute{g}_{w_{n}}(\theta_{w_{n},h})-\hat{g}_{w_{n}}(\theta_{w_{n},h})\big) + \hat{\xi}_{w_{n}}(\theta_{w_{n},h}),
\end{aligned}
\label{eq - Asy equiv null seq}
\end{equation}
to obtain the asymptotic equivalence.
\subsection{Limiting Local Power Function of CMS Tests}\label{Subsection CMS Local Power Funcion}
This section employs the setup in Section 8 of~\cite{andrews2010inference} to show the limiting local power function of the CMS tests coincide with their GMS counterparts when the null parameter space is $\mathcal{F}_{+}=\{(\theta, F) \in \mathcal{F}: \text{Assumption T holds}\}.$ For sequences of parameters $\{(\theta_{n,*}, F_{n}): n \geq 1\}$, consider the testing problem
\begin{align}\label{LAhyp}
H_{0}: E_{F_{n}}\big(g_{j}(W_{i},\theta_{n,*})\big) \geq 0 \ \forall \ j \in \mathcal{J} \ \text{vs.} \ H_{1}: \text{$H_{0}$ is false}
\end{align}
where $\theta_{n,*}=\theta_{n}+\eta n^{-\frac{1}{2}}(1+o(1))$ for all $n \geq 1$, $(\theta_{n},F_{n}) \in \mathcal{F}_+$ for all $n \geq 1$, and $\eta \in \mathbb{R}^{d}$ where $d=\dim(\theta_{n})<+\infty$. The idea is to study the behavior of the testing procedure along sequences of parameters $\{(\theta_{n,*}, F_{n}): n \geq 1)\}$ that differ locally from a point in the true parameter space $\mathcal{F}_+$ by $O(n^{-\frac{1}{2}})$. The local power function is defined as $P_{F_{n}}\big(T_{n}(\theta_{n,*}) > \acute{c}_{n}(\theta_{n,*},1-\alpha)\big)$, where $P_{F_{n}}(\cdot)$ is the probability measure induced by random sampling from $F_{n}$ for all $n \geq 1$. The objective is to derive an expression for the limiting local power function, $\lim_{n\rightarrow+\infty}P_{F_{n}}\big(T_{n}(\theta_{n,*})> \acute{c}_{n}(\theta_{n,*},1-\alpha)\big),$ and compare it to its GMS counterpart.

\par To this end, we introduce technical assumptions for deriving the limiting local power function for CMS tests. These assumptions are from Section~8 of~\cite{andrews2010inference}.
\begin{lassumption}\label{LAssumption1}
The true parameters $\{(\theta_{n}, F_{n}): n \geq 1\}$ satisfy:
\begin{enumerate}
\item $\theta_{n}=\theta_{n,*}-\eta n^{-\frac{1}{2}}(1+o(1))$ for some $\eta \in \mathbb{R}^{d}$, $\theta_{n,*} \rightarrow \theta_{0}$ and $F_{n}\rightarrow F_{0}$ as $n\rightarrow+\infty$, where $(\theta_{0},F_{0})\in\mathcal{F}_+$.
\item For each $j\in \mathcal{J}$, there exists $h_{1,j} \in \mathbb{R}_{+,\infty}$ such that $n^{\frac{1}{2}}E_{F_{n}}\big(g_{j}(W_{i},\theta_{n})\big)/\sigma_{F_{n},j}(\theta_{n}) \rightarrow h_{1,j}$ as $n\rightarrow+\infty$.
\item $\sup\big\{E_{F_{n}}|g_{j}(W_{i},\theta_{n,*})/\sigma_{F_{n},j}(\theta_{n,*})|^{2+\delta}: n \geq 1\big\}<+\infty$ for all $j \in \mathcal{J}$ for some $\delta>0.$
\end{enumerate}
\end{lassumption}
\noindent The first two parts of Assumptions~LA1 show that the sequence of true parameters, $\{\theta_{n}: n \geq 1\}$, is $n^{-\frac{1}{2}}$-local to the sequence of parameters under the null hypothesis, $\{\theta_{n,*}: n\geq 1\}$ and provides the limit of the sequence of normalised moment functions when evaluated at the sequence of true parameters $\{\theta_{n}: n\geq 1\}$. The third part of this assumption is a uniform integrability condition that permits the use of stochastic limit theorems for triangular arrays of row-wise IID random variables. The second assumption is as follows.
\begin{lassumption}\label{LAssumption2}
$\Pi(\theta, F) := (\partial/\partial \theta^{\top})[D^{-\frac{1}{2}}(\theta, F)E_{F}(g(W_{i}, \theta))] \in \mathbb{R}^{J\times \dim(\Theta)}$ exists and is a continuous function in a neighbourhood of $(\theta_{0}, F_{0})$.
\end{lassumption}
\noindent Both Assumptions LA\ref{LAssumption1} and LA\ref{LAssumption2} are important for proving the large sample properties of CMS tests under $n^{-\frac{1}{2}}$-local alternatives. Namely, they allow one to mean value expand the normalised moment functions under $H_{0}$ around $\theta=\theta_{n}$ and show that
\begin{align*}
\lim_{n\rightarrow \infty}n^{\frac{1}{2}}D^{-\frac{1}{2}}(\theta_{n,*}, F)E_{F}(g(W_{i}, \theta_{n,*}))=h_{1}+\Pi(\theta_{0},F_{0})\eta
\end{align*}
which can then be used to show that $T_{n}(\theta_{n,*}) \overset{d}{\rightarrow} J_{h_{1}, \eta},$ where $J_{h_{1}, \eta}$ is the distribution function of $S\big(\Omega_{0}^{\frac{1}{2}}Z^{*}+h_{1}+\Pi(\theta_{0},F_{0})\eta, \Omega_{0}\big)$ and $Z^{*}\sim N(0_{J}, I_{J})$ \citep{andrews2010inference}.
\begin{lassumption}\label{LAssumption3}
$\lim_{n\rightarrow+\infty}\kappa_{n}^{-1}n^{\frac{1}{2}}D^{-\frac{1}{2}}(\theta_{n},F_{n})E_{F_{n}}(g(W_{i},\theta_{n}))=\pi_{1} \in \mathbb{R}^{J}_{+,\infty}$.
\end{lassumption}
\noindent The last assumption involves the set
$C(\varphi)=\{\tilde{\pi}_{1}\in \mathbb{R}^{J}_{[+\infty]}: \forall j \in \mathcal{J},\;\text{either}\;\tilde{\pi}_{1,j}=+\infty\;\text{or}\;\varphi_{j}(\xi, \Omega) \rightarrow \varphi_{j}(\tilde{\pi}_{1}, \Omega_{0})\;\text{as}\; (\xi, \Omega) \rightarrow (\tilde{\pi}_{1}, \Omega_{0})\}.$ Loosely, $C(\varphi)$ is the set of all vectors in $\mathbb{R}^{J}_{[+\infty]}$ for which $\varphi$ is continuous at $(\tilde{\pi}_{1}, \Omega_{0})$. With $\varphi=\varphi^{(1)},$ this set is $C(\varphi^{(1)})=\{\tilde{\pi}_{1}\in \mathbb{R}^{J}_{[+\infty]}:\tilde{\pi}_{1,j}\neq 1,\forall j\in\mathcal{J}\}.$
\begin{lassumption}\label{LAssumption4}
\begin{enumerate*}
\item $\pi_{1} \in C(\varphi^{(1)})$ and \item $P_{F}\Big(S\big(\Omega_{0}^{\frac{1}{2}}Z^{*}+\varphi^{(1)}(\pi_{1}, \Omega_{0}), \Omega_{0}\big) \leq x\Big)$ is continuous and strictly increasing at $x=c_{\pi_{1}}(\varphi^{(1)} , 1-\alpha)$, the $1-\alpha$ quantile of the distribution function of $S\big(\Omega_{0}^{\frac{1}{2}}Z^{*}+\varphi^{(1)}(\pi_{1}, \Omega_{0}), \Omega_{0}\big)$.
\end{enumerate*}
\end{lassumption}
\noindent Assumptions~LA3 and LA4 are imposed so that we can use Theorem 2(a) of~\cite{andrews2010inference} to obtain the form of the GMS limiting local power function.

\par Next, we present the second main result of this paper. This result states that GMS and CMS tests are asymptotically equivalent, to first-order, under $n^{-1/2}$-local alternatives.
\begin{theorem}\label{LAthm}
Suppose $S$ satisfies Assumptions 1-5, $\varphi=\varphi^{(1)}$ in~(\ref{t-test MS function}), the sequence $\{\kappa_n:n\geq1\}$ satisfies Assumption~K, and that Assumptions~LA1 - LA4, hold. Then $\lim_{n\rightarrow+\infty}P_{F_{n}}\big(T_{n}(\theta_{n,*})>\acute{c}_{n}(\theta_{n,*}, 1-\alpha)\big)=1-J_{h_{1},\eta}\big(c_{\pi_{1}}(\varphi, 1-\alpha)\big).$
\end{theorem}
\begin{proof}
See Appendix~\ref{Appendix - proof Thm 2}.
\end{proof}
\noindent The intuition behind Theorem~\ref{LAthm} is essentially the same as Theorem~\ref{asymptoticsize}. For a given sequence $\{(\theta_{n,*}, F_{n}): \ n \geq 1\}$ of $n^{-1/2}$-local alternatives, we show
$\acute{\xi}_{n}(\theta_{n,*})=\hat{\xi}_{n}(\theta_{n,*})+o_{p}(1)$. This asymptotic equivalence is a consequence of applying
$\acute{g}_{n}(\theta_{n,*})-\hat{g}_{n}(\theta_{n,*})=O_{p}(n^{-\frac{1}{2}})$ (see Lemma~\ref{LP4} in Appendix~\ref{Section Local Pwr Thms 2 3}) and Assumption~K to a decomposition of $\acute{\xi}_{n}(\theta_{n,*})$ identical to~(\ref{eq - Asy equiv null seq}). Therefore, the pairs $(\acute{\xi}_{n}(\theta_{n,*}), \hat{\Omega}_{n}(\theta_{n,*}))$ and $(\hat{\xi}_{n}(\theta_{n,*}), \hat{\Omega}_{n}(\theta_{n,*}))$ are asymptotically equivalent along sequences of $n^{-1/2}$-local alternatives. As this is the only point of difference between CMS and GMS, Theorem~\ref{LAthm} follows from Theorem~2(a) of~\cite{andrews2010inference}. An important corollary to Theorem~\ref{LAthm} is that CMS inherits the first-order improvements that GMS exhibits over subsampling and plug-in asymptotic critical values (see \citealp{andrews2010inference}).

\subsection{Local Power Comparison Between CMS and GMS Tests}\label{Subsection CMS and GMS Local Power comparison}
While Theorem~\ref{LAthm} establishes the equality of the limiting local power functions of CMS and GMS tests under $n^{-1/2}$-local alternatives, this section presents results that characterize sequences of local alternatives under which the power of CMS tests dominate their GMS counterparts for sufficiently large, but finite, samples. First, we must establish when it is meaningful to compare tests along sequences of $n^{-1/2}$-local alternatives. Under the conditions of part 2 of Theorem~\ref{asymptoticsize},
for every $r>0,$ there exists $N_{r,0},N_{r,1}\in\mathbb{Z}_{+}$ (depending on $r$) such that
\begin{align}
\left|\sup_{m\geq n}\sup_{(\theta, F) \in \mathcal{F}_{+}}P_{F}(T_{m}(\theta) > \acute{c}_{m}(\theta, 1-\alpha))-\alpha\right| & \leq \frac{r}{2}\quad\forall n\geq N_{r,0}\quad\text{and}\;\label{eq - LP-size CMS}\\
\left|\sup_{m\geq n}\sup_{(\theta, F) \in \mathcal{F}_{+}}P_{F}(T_{m}(\theta) > \hat{c}_{m}(\theta, 1-\alpha))-\alpha\right| & \leq \frac{r}{2}\quad\forall n\geq N_{r,1},\label{eq - LP-size GMS}
\end{align}
by the definition of limit superior (with respect to $n$). Then by the triangular inequality,
\begin{align*}
\left|\sup_{m\geq n}\sup_{(\theta, F) \in \mathcal{F}_{+}}P_{F}(T_{m}(\theta) > \acute{c}_{m}(\theta, 1-\alpha))-\sup_{m\geq n}\sup_{(\theta, F) \in \mathcal{F}_{+}}P_{F}(T_{m}(\theta) > \hat{c}_{m}(\theta, 1-\alpha))\right|\leq r
\end{align*}
for all $n\geq N_r=\max\{N_{r,0},N_{r,1}\},$ holds. In words, given an error tolerance $r,$ the tails of the sequences of exact sizes of CMS and GMS tests are within $r$ of $\alpha$ and of each other, when $n\geq N_r.$ Thus, given $r$ (e.g., 0.0001), it is meaningful to compare the rejection probabilities along sequences of local alternatives when $n\geq N_r.$

\par Let $\mathcal{H}$ denote the set of all sequences $\{(\theta_{n,*},F_{n}): n \geq1\}$ that satisfy Assumption~LA1 and LA2. The family we consider for the comparisons is defined as
\begin{align}\label{eq - local alternatives family}
\mathcal{M}=\left\{\{(\theta_{n,*},F_{n}): n \geq1\}\in\mathcal{H}:\text{$\Omega(\theta_{n,*},F_{n})$ has nonnegative off-diagonal elements, $\forall n$}\right\}.
\end{align}
 For $\{(\theta_{n,*},F_{n}): n \geq 1\} \in \mathcal{H}$, let $\hat{\Upsilon}_{n}(\theta_{n,*}) := \{j\in\mathcal{J}:\varphi_j^{(1)}(\hat{\xi}_{n}(\theta_{n,*}),\hat{\Omega}_{n}(\theta_{n,*}))=0\}$ and $\acute{\Upsilon}_{n}(\theta_{n,*}) := \{j\in\mathcal{J}:\varphi_j^{(1)}(\acute{\xi}_{n}(\theta_{n,*}),\hat{\Omega}_{n}(\theta_{n,*}))=0\}$ for each $n \geq 1$. We have the following result.
\begin{theorem}\label{LPorder}
Let $\mathcal{M}$ be as in~(\ref{eq - local alternatives family}). Suppose that $S$ satisfies Part 1 of Assumption~1, $\varphi=\varphi^{(1)}$ in~(\ref{t-test MS function}), and the sequence $\{\kappa_n:n\geq1\}$ satisfies Assumption~K. For every $\{(\theta_{n,*}, F_{n}): n \geq 1\} \in \mathcal{M},$ there exists $N(\theta_{n,*}, F_{n})  \in \mathbb{Z}_{+}$ such that
\begin{align}\label{LPorder - 1}
P_{F_{n}}\Big(T_{n}(\theta_{n})>\hat{c}_{n}(\theta_{n},1-\alpha)\Big) \leq P_{F_{n}}\Big(T_{n}(\theta_{n}) > \acute{c}_{n}(\theta_{n},1-\alpha)\Big)\quad \forall n \geq N(\theta_{n,*}, F_{n}).
\end{align}
If in addition $S$ satisfies part 1 of Assumption 2 and Part 2 of Assumption 5, and the event $$\bigg\{\acute{\Upsilon}_{n}(\theta_{n,*})\subsetneq \hat{\Upsilon}_{n}(\theta_{n,*})\bigg\} \bigcap \bigg\{\hat{c}_{n}(\theta_{n,*},1-\alpha)>0\bigg\} \bigcap \bigg\{\acute{c}_{n}(\theta_{n,*},1-\alpha)<T_{n}(\theta_{n,*}) \leq \hat{c}_{n}(\theta_{n,*},1-\alpha)\bigg\}$$
has positive probability for each $n \geq N(\theta_{n,*},F_{n})$, then the weak inequalities in~(\ref{LPorder - 1}) are strict.
\end{theorem}
\begin{proof}
See Appendix~\ref{Appendix - proof Thm 3 and Prop 1}.
\end{proof}
\noindent Theorem~\ref{LPorder} states the rejection probabilities of CMS tests are no less than their GMS counterparts in large enough, but finite, sample sizes, under local alternatives in $\mathcal{M}.$ It also provides a sufficient condition for the ordering to hold strictly. Thus, for each sequence of local alternatives in $\mathcal{M}$ and small $r>0,$ the local power of a CMS test is larger than its GMS counterpart when $n \geq \max\{N(\theta_{n,*}, F_{n}),N_r\},$ where $N_r=\max\{N_{r,0},N_{r,1}\}$ and $N_{r,0}$ and $N_{r,1}$ defined in~(\ref{eq - LP-size CMS}) and~(\ref{eq - LP-size GMS}), respectively.

\par The key message from Theorem~\ref{LPorder} is that a comparison of GMS and CMS tests based on first-order asymptotics can be misleading, as it does not reflect the finite-sample situation for certain local alternatives. The result of Theorem~\ref{LPorder} is similar to Corollary~6.1 \cite{Lok-Tabri-inpress}; however, it is important to note that their result is specific to moment inequalities arising from restricted stochastic dominance orderings. Consequently, Theorem~\ref{LPorder} provides a nontrivial extension of their result to the moment inequality model with finitely many inequalities and arbitrary moment functions, when the off-diagonal elements of $\Omega(\theta_{n,*},F_{n})$ are non-negative for each $n.$

\par At the heart of this result is the marriage of the non-negative correlational structure on $\Omega(\theta_{n,*},F_{n})$ and constrained empirical likelihood estimation. This marriage begets $\hat{g}_{n,j}(\theta_{n,*}) \leq \acute{g}_{n,j}(\theta_{n,*})$ with probability approaching $1,$ for all sequences in $\mathcal{M}$ (see Lemma~\ref{LPC2}). This ordering of the estimators implies that $\varphi^{(1)}_{j}(\acute{\xi}(\theta_{n,*}),\hat{\Omega}_{n}(\theta_{n,*})) \geq \varphi^{(1)}_{j}(\hat{\xi}_{n}(\theta_{n,*}),\hat{\Omega}_{n}(\theta_{n,*}))$ holds with probability approaching 1, for all sequences in $\mathcal{M}$ (see Lemma~\ref{GMSorder}). It is this ordering of the moment selection functions under such sequences that gives rise to the result of Theorem~\ref{LPorder}.

\par While Theorem~\ref{LPorder} indicates that the local powers of the GMS and CMS tests can be ordered under a class of local alternatives $\mathcal{M}$ for large enough $n,$ it does not specify the extent of the discrepancy in the local powers. It is quite difficult to determine the extent of this discrepancy analytically. However, Section~\ref{Setion - MC Sims} presents Monte Carlo evidence that the discrepancy that Theorem \ref{result3} implies can be very large for local alternative sequences which have some non-violated inequalities (SNVIs). That is, sequences $\{(\theta_{n,*},F_{n}): n \geq1\}$ in $\mathcal{M}$ where there exists $j\in\mathcal{J}$ such that $E_{F_{n}}\big(g_{j}(W_i,\theta_{n,*})\big)>0\;\forall n$ and $\lim_{n\rightarrow+\infty}E_{F_{n}}\big(g_{j}(W_i,\theta_{n,*})\big)=0$.

\subsection{Power Against Distant Alternatives}\label{Subsection - Distant Alternatives}
\par This section shows CMS tests are consistent against distant alternatives. Distant alternatives include fixed alternatives and alternatives that differ from the null by greater than $O(n^{-\frac{1}{2}})$. The next assumption is useful for deducing this result, and it is the same one introduced by~\cite{andrews2010inference} in Section~9 of their paper.
\begin{dassumption}
Let $g_{n,j}^{*}=E_{F_{n}}(g_{j}(W_{i},\theta_{n,*}))/\sigma_{F_{n},j}(\theta_{n,*})$ for each $j \in \mathcal{J},$ and $\upsilon_{n}=\max_{j \in \mathcal{J}}\{-g_{n,j}^{*}\}$.
\begin{enumerate*}
\item $n^{\frac{1}{2}}\upsilon_{n}\rightarrow+\infty$ as $n\rightarrow+\infty$
\item $\Omega(\theta_{n,*},F_{n})\rightarrow \Omega_{1}$, $\Omega_{1} \in \varPsi_{2}$.
\end{enumerate*}
\end{dassumption}
\noindent The key part of this assumption is the first part, which indicates that there exists $j \in \mathcal{J}$ such that $g_{n,j}^{*}<0$ and that the violation of the non-negativity constraint is not $O(n^{-\frac{1}{2}})$. This condition differs from the setup with $n^{-\frac{1}{2}}$-local alternatives, where the sequences of alternatives $\{(\theta_{n,*},F_{n}): n \geq 1\}$ are within a $n^{-\frac{1}{2}}$-neighbourhood of $\mathcal{F}_+$.
\par We have the following result.
\begin{theorem}\label{distant}
Suppose $S$ satisfies Assumptions 1,3,4 and 6, $\varphi=\varphi^{(1)}$ in~(\ref{t-test MS function}), and the sequence $\{\kappa_n:n\geq1\}$ satisfies Assumption~K. Then
$\lim_{n\rightarrow+\infty}P_{F_{n}}\Big(T_{n}(\theta_{n,*})>\acute{c}_{n}(\theta_{n,*},1-\alpha)\Big)=1.$
\end{theorem}
\begin{proof}
See Appendix~\ref{Appendix - proof Thm 4}.
\end{proof}

\section{Simulation Results}\label{Setion - MC Sims}
\par This section studies the finite-sample performance of the CMS procedure and compares it to the GMS procedure using a simulation experiment based on the designs in~\cite{andrews2012inference}. The study uses the test statistics $S_1$ in~(\ref{MMM}) and $S_{2A}$ in~(\ref{QLR}), the recommended moment selection function $\varphi=\varphi^{(1)}$ in~(\ref{t-test MS function}), and the recommended localisation parameter $\kappa_n=(\ln n)^{1/2}.$ The nominal level is set to $\alpha=0.05,$ and we considered sample sizes $n=50,100$, and $250$. We also report simulation results for (i) the RSW procedure that use $S_1$ and $S_{2A}$, and (ii) the recommended RMS testing procedure, which combines $S=S_{2A}$, $\varphi=\varphi^{(1)}$, and $\kappa$-auto (a data-driven choice of $\kappa_n$), as additional benchmarks in studying the finite-sample performance of CMS; see Appendices~\ref{RSWdetails} and~\ref{RMSdetails}, respectively, for further details on these testing procedures. Only bootstrap versions of the tests were implemented, with 10000 bootstrap samples per Monte Carlo replication.  The computations were implemented using \textsf{R}.

\par For a given $\theta,$ the null hypothesis is $H_0:\theta\in\Theta(F_0).$ The experimental design in~\cite{andrews2012inference,andrews2012supplement} is a general formulation of that testing problem that does not require the specification of a particular form for the moment functions $\{g_{j}(\cdot,\theta): j=1,...,J\}.$ They note that the finite-sample properties of tests of $H_0$ depend on the moment functions only through (i) the vector $\mu=[E_{F_{0}}\big(g_{1}(W_i,\theta)\big),\ldots,E_{F_{0}}\big(g_{J}(W_i,\theta)\big)],$ (ii) the correlation matrix $\Omega=\text{Corr}\left(g_{1}(W_i,\theta),\ldots,g_{J}(W_i,\theta)\right),$ and (iii) the distribution of the mean zero, variance $I_J$ random vector $Z^{\dagger}=[Z_1^{\dagger},\ldots,Z_J^{\dagger}]$ where
$$
Z_j^{\dagger}=\text{Var}_{F_0}^{-1/2}\left(g_{j}(W_i,\theta)\right)\left(g_{j}(W_i,\theta)-E_{F_{0}}\big(g_{j}(W_i,\theta)\big)\right),\quad j=1,\ldots, J.
$$
We consider the case $Z^{\dagger}\sim N(0_J,I_J)$ and three correlation matrices, $\Omega_{\text{Neg}},$ $\Omega_{\text{Zero}},$ and $\Omega_{\text{Pos}},$ which exhibit negative, zero, and positive correlations.

\par The assertion of the null hypothesis in this general formulation is $H_0:\mu_j\geq0\;\forall j=1,\ldots,J$. For comparisons under the null hypothesis, we follow~\cite{andrews2012inference} by comparing the tests' maximum null rejection probabilities (MNRPs). The MNRPs are computed over the mean vectors $\mu$ in the null parameter space given the correlation matrix $\Omega\in\{\Omega_{\text{Neg}},\Omega_{\text{Zero}},\Omega_{\text{Pos}}\}$ and under the assumption of normally distributed moment inequalities. Based on simulation evidence, they conjecture that the MNRPs occur for mean vectors $\mu$ whose elements are $0$'s and $+\infty$'s. Thus, given a nominal level $\alpha,$ they compute MNRP results over the set of mean vectors $\mu$ which have that form. The results we report are for $J=2,4$ and $10.$ The matrix $\Omega_{\text{Zero}}$ equals the $J$-dimensional identity matrix. The matrices $\Omega_{\text{Neg}}$ and $\Omega_{\text{Pos}}$ are Toeplitz matrices with correlations given by the following: for $J=2:$ $\rho=-.9$ for $\Omega_{\text{Neg}}$ and $\rho=.5$ for $\Omega_{\text{Pos}};$ for $J=4:$ $\rho=(-.9,.7,-.5)$ for $\Omega_{\text{Neg}}$ and $\rho=(.9,.7,.5)$ for $\Omega_{\text{Pos}};$ for $J=10:$ $\rho=(-.9,.8,-.7,.6,-.5,.4,-.3,.2,-.1)$ for $\Omega_{\text{Neg}}$ and $\rho=(.9,.8,.7,.6,.5,\ldots,.5)$ for $\Omega_{\text{Pos}}$. As in~\cite{andrews2012inference}, the simulation study treats the correlation matrices as unknown in the implementation of all of the tests.

\par For power comparisons, we also follow~\cite{andrews2012inference,andrews2012supplement}. They compare the power of different tests by comparing their empirical power for a chosen set of alternative parameter vectors
$\mu\in\mathbb{R}^{J}$ for a given correlation matrix $\Omega\in\{\Omega_{\text{Neg}},\Omega_{\text{Zero}},\Omega_{\text{Pos}}\}$. The sets of $\mu$ vectors in the alternative are similar to the ones described in~\cite{andrews2012inference,andrews2012supplement}. We adjust those sets so as to compare the local power properties of the testing procedures. The adjustment is as follows. For each $J\in\{2,4,10\}$ the set of $\mu$ vectors are given by $\mathcal{M}_{J,n}(\Omega) = \left\{\mu/\sqrt{n}: \mu\in\mathcal{M}_J(\Omega)\right\},$ where the set $\mathcal{M}_J(\Omega)$ of $\mu$ vectors is described in Section~7.1 of~\cite{andrews2012supplement}. The $\mu$ vectors in $\mathcal{M}_{J,n}(\Omega)$ are scaled versions of those in $\mathcal{M}_J(\Omega),$ where the scaling is by $n^{-1/2}$ to create the $n^{-1/2}$-local alternatives. There are $7,24,$ and $40$ elements in $\mathcal{M}_{J,n}(\Omega)$ for $J=2,4$ and $10,$ respectively. We omit their description for brevity.

\par As the MNRPs of the tests can differ in finite-samples, the simulation results on power comparisons are based on a MNRP correction that is similar to the one employed by~\cite{andrews2012inference}. For each test statistic $S$, the MNRP correction of the CMS, GMS and RMS procedures is to add a constant based on the true matrix $\Omega$ to their corresponding critical values, so that their resulting MNRPs match that of the RSW testing procedure with nominal level $\alpha=0.05;$ see Section~\ref{scdetails} in the appendix for the details. The simulation studies in~\cite{andrews2012inference} and~\cite{romano2014practical} compare tests under the alternative using average MNRP-corrected power, where the average is computed over alternative $\mu$ vectors in $\mathcal{M}_J(\Omega).$ We report simulation results graphically using boxplots of the MNRP-corrected \emph{local} powers over sets of $\mu$ vectors $\mathcal{M}_{J,n}(\Omega)$ for the 54 different combinations of $(J,\Omega,S,n)$ for each of the CMS, GMS and RSW procedures, and 9 different combination of $(J,\Omega,S_{2A},n)$ for the recommended RMS test. Additionally, we report average MNRP-corrected local powers of the different tests across the aforementioned configurations using the symbol $\oplus$ in these plots.

\par While the average MNRP-corrected power is a useful criterion for comparing tests across $\mu$ vectors in a given set $\mathcal{M}_{J,n}(\Omega)$, it does not convey the whole picture of the tests' performance over elements in $\mathcal{M}_{J,n}(\Omega).$ Reporting boxplots, as we do, reveals the variation in powers of the tests across elements in $\mathcal{M}_{J,n}(\Omega);$ thus, presenting a broader and more extensive approach to comparing the tests under the alternative. These plots are especially useful for detecting differences in the performances of tests when the averages of their MNRP-corrected powers are close, but exhibit different distributional variations in MNRP-corrected power across $\mu$ vectors in $\mathcal{M}_{J,n}(\Omega).$

\subsection{Maximum Null Rejection Probabilities}
\par As in~\cite{andrews2010inference},~\cite{andrews2012inference}, and~\cite{romano2014practical}, empirical MNRPs are simulated as the maximum rejection probability over all $\mu$ vectors whose components are $0$ and $+\infty,$ with at least one component equal to zero. Table~\ref{Table - MNRPs} reports the MNRPs for tests. Each experiment used 10000 Monte Carlo replications when $J \in \{2,4\}$ and 2500 when $J=10$.
\begin{table}[pt]
  \caption{Empirical Maximum Null Rejection Probabilities}\label{Table - MNRPs}

\centering
\resizebox{16.5cm}{!}{
  \begin{tabular}{lccccccccccccc}
  \toprule
  & & \multicolumn{3}{r}{$J=2$} & & \multicolumn{3}{r}{$J=4$} & & \multicolumn{3}{r}{$J=10$} \smallskip\smallskip \\ \cline{4-6} \cline{8-10} \cline{12-14}
$n$ & Procedure & Statistic & $\Omega_{\text{Neg}}$ & $\Omega_{\text{Zero}}$ & $\Omega_{\text{Pos}}$ & & $\Omega_{\text{Neg}}$ & $\Omega_{\text{Zero}}$ & $\Omega_{\text{Pos}}$ & & $\Omega_{\text{Neg}}$ & $\Omega_{\text{Zero}}$ & $\Omega_{\text{Pos}}$ \smallskip\smallskip  \\
 \midrule

50   &  GMS  & $S_{1}$ &         0.06         &    0.052             &     0.052      & &   0.065               &       0.054             &         0.053      & & 0.078  & 0.054 & 0.053   \smallskip\smallskip \\
   &    & $S_{2A}$ &    0.07               &   0.052                 &  0.052         & &     0.077                 &      0.054              &      0.058         & & 0.083 & 0.054 & 0.067  \smallskip\smallskip \\
     &  CMS   &  $S_{1}$ &     0.053              &      0.052            &      0.052  & &        0.052             &       0.054          & 0.053   & & 0.062 & 0.054 & 0.053   \smallskip\smallskip \\
     &    &  $S_{2A}$ &        0.053            &       0.052          &   0.053     & &      0.054               &            0.054          &  0.058  & & 0.061 &  0.054 & 0.067 \smallskip\smallskip \\
     &  RSW   & $S_{1}$  &    0.047    & 0.047 & 0.047 & &0.047 &0.047 & 0.047 & &0.054 &0.047 & 0.047
      \smallskip\smallskip \\
           &    & $S_{2A}$  &  0.047	& 0.046 &	0.047 & &	0.047 &	0.047 &	0.047 & &	0.054	& 0.047 & 0.047	
      \smallskip\smallskip \\
           &  RMS   & $S_{2A}$ & 0.053 &	0.052 &	0.056 & &	0.049 &	0.050 &	0.048 & & 0.048 &0.050 & 0.047
      \smallskip\smallskip \\
\midrule

100   &  GMS  & $S_{1}$ &  0.056	& 0.052 &	0.051 & & 	0.059 &	0.052 &	0.051 & & 	0.070	& 0.052 &	0.058  \smallskip\smallskip \\
   &    & $S_{2A}$ &   0.063	&0.052 &	0.052& &	0.067 &	0.050	& 0.055 & &	0.072	& 0.052	&  0.062 \smallskip\smallskip \\
     &  CMS   &  $S_{1}$ &  0.051&	0.052 &	0.051 && 	0.052& 	0.052 &	0.051 & &	0.061	 &0.052	 &0.058   \smallskip\smallskip \\
     &    &  $S_{2A}$ & 0.051&	0.053	&0.052	& &0.053&	0.050&	0.055 & &	0.062	&0.052	& 0.062 \smallskip\smallskip \\
     &  RSW   & $S_{1}$  & 0.047	&0.048&	0.046	&&0.046&	0.045	&0.046& &	0.055&	0.048&	0.050
      \smallskip\smallskip \\
           &    & $S_{2A}$  &  0.047	&0.048&	0.046 & &	0.046	&0.045&	0.046 & &0.056	&	0.048	& 0.051
      \smallskip\smallskip \\
           &  RMS   & $S_{2A}$ &   0.051	& 0.052	 & 0.053 &	& 0.048	 & 0.046	&0.046	& & 0.052 &	0.042	& 0.043
      \smallskip\smallskip \\
\midrule
250   &  GMS  & $S_{1}$ &  0.049 &	0.049	 & 0.051& &	0.053	&0.052 &	0.052 & &	0.053	& 0.058	 & 0.057  \smallskip\smallskip \\
   &    & $S_{2A}$ &       	0.052 &	0.049	&0.051 & &	0.057&	0.052&	0.054& &	0.057 &0.055	& 0.059 \smallskip\smallskip \\
     &  CMS   &  $S_{1}$ &    0.049 &	0.049	& 0.051 & &	0.051& 	0.052	& 0.052 & &	0.051 &	0.058 &	0.057  \smallskip\smallskip \\
     &    &  $S_{2A}$ &  0.049	 & 0.049	 &0.051	 & &0.052&	0.052&	0.054 & &	0.052 & 0.055	& 0.059 \smallskip\smallskip \\
     &  RSW   & $S_{1}$  &     0.044 &	0.043 &	0.046 & &	0.046 &	0.047 &	0.048 & &	0.046	& 0.049	& 0.052     \smallskip\smallskip \\
           &    & $S_{2A}$  &    	0.044	 & 0.043 &	0.046 & &	0.046 &	0.047 &	0.047	 & & 0.046 & 	0.049	 & 0.052
      \smallskip\smallskip \\
           &  RMS   & $S_{2A}$ &  0.049 &	0.049 & 	0.051	 &  & 0.046 &	0.048 & 	0.048	& & 0.046 &	0.046 &	0.045
      \smallskip\smallskip \\
\bottomrule

    \end{tabular}
    }
\end{table}

\par Overall, the procedures achieve a satisfactory performance for all cases considered. The RMS and RSW tests perform the best, as their MNRPs are closest to the 5\% nominal level across all of the cases considered. For the RSW procedure: the MNRPs fall within the ranges [.043,.056] and [.043,.055] when using $S_{2A}$ and $S_1$ test statistics, respectively. For the RMS test: the MNRPs fall into the range [.042,.053]. The CMS tests over-reject the null slightly: the MNRPs fall within the ranges [.049,.067] and [.049,.061] when using $S_{2A}$ and $S_1$ test statistics, respectively. The tables also show CMS tests have better MNRPs than their GMS versions as the latter tend to over-reject more: the GMS MNRPs fall within the ranges [.049,.083] and [.049,.078] for $S_{2A}$ and $S_1$, respectively. The largest MNRPs arise in the configurations where $\Omega=\Omega_{\text{Neg}},$ and these MNRPs increase with larger $J$, for the CMS, GMS and RSW tests. However, the MNRPs of all of these tests do get closer to the 5\% nominal level with larger sample sizes, across all configurations, and for CMS tests, this numerical result is a consequence of Theorem~\ref{asymptoticsize}.

\par While we don't have a theoretical result on improved size control of CMS tests over their GMS versions, Table~\ref{Table - MNRPs} provides simulation-based evidence of such an improvement. Hence, these results point to the potential benefit of implementing the information~(\ref{eq - moment inequality}), as we do, in two step testing procedures, under the null. The next section presents simulation results on MNRP-corrected power of these tests, under local alternatives, and illustrates the result of Theorem~\ref{LPorder}.

\subsection{Local Power}
\par Figures~\ref{Figure - Power QLR} and~\ref{Figure - Power MMM} below report boxplots of the MNRP-corrected powers of the tests under $S_{2A}$ and $S_1$, respectively. The results can be summarised as follows. For each test statistic, the MNRP-corrected power values of the tests are generally distributed in a similar way in configurations where $\Omega=\Omega_{\text{Neg}},$ and all of the tests have comparable average powers in those configurations. By contrast, in configurations where $\Omega=\Omega_{\text{Zero}}$, for each test statistic, the boxplots show the RSW tests' MNRP-corrected power values tend to be (i) more dispersed (as shown by the lengths of their boxes), (ii) have a wider overall range, and (iii) have lower average power in comparison to the remaining tests, which all behave similarly as can be seen by their boxplots. For example, the average power of the RSW test when $S=S_{2A}$, $J=10$, and $n=250$ is approximately equal to 0.57, while the averages of the remaining procedures in that scenario are all approximately equal to 0.66, which is a large difference.

\par More noticeable differences in the tests' performance arise in configurations where $\Omega=\Omega_{\text{Pos}}.$ For each test statistic, there is evidence for the following ranking in terms of average MNRP-corrected power, uniformly in $J$ and $n$: CMS in first place, GMS in second place, and RSW in third place, with the RMS test tied in first place with the CMS-$S_{2A}$ test. The boxplots also show:
\begin{itemize}
\item The MNRP-corrected power values for the RMS and CMS-$S_{2A}$ tests are generally distributed in a similar way for each $n$ and $J$, except when $J=2$ the CMS-$S_{2A}$ power values are slightly more dispersed (as shown by the length of the boxes) than their RMS counterparts for each $n$.
\item For each $S$, the MNRP-corrected power values of CMS tests are markedly less dispersed and have smaller overall ranges than their GMS and RSW counterparts.

\item The difference among the CMS, GMS and RSW tests in these configurations with $S=S_1$ can be strikingly large in terms of average power; for example, with $J=4$ and $n=250,$ the average powers of CMS, GMS and RSW tests are approximately equal to 0.75, 0.65, and 0.60, respectively. By contrast, with $S=S_{2A}$ the difference among these tests is less pronounced, which is on account of using a more effective test statistic. For example, in the aforementioned configuration, the average powers of CMS, GMS and RSW tests are approximately equal to 0.76, 0.74, and 0.73, respectively. However, this less pronounced difference in average powers does not mean that these procedures behave similarly, as evidenced by the radically different boxplots of the tests' power values.
\end{itemize}

\begin{landscape}
 \begin{figure}
\centering
\includegraphics[height=16cm, width=18cm]{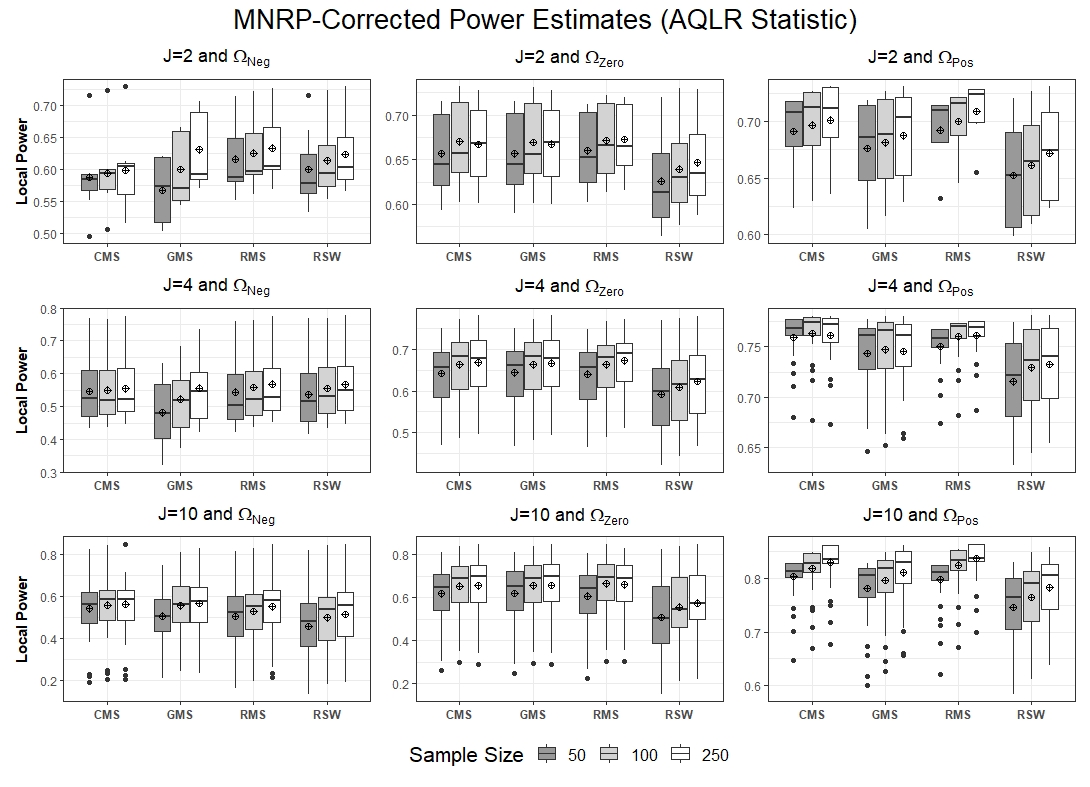}
\caption{Boxplots of MNRP-corrected powers of $S_{2A}$-based tests. For each configuration, the symbol $\oplus$ marks the location of the average MNRP-corrected power of a test. }\label{Figure - Power QLR}
\end{figure}
\end{landscape}

\begin{landscape}
 \begin{figure}
\centering
\includegraphics[height=16cm, width=18cm]{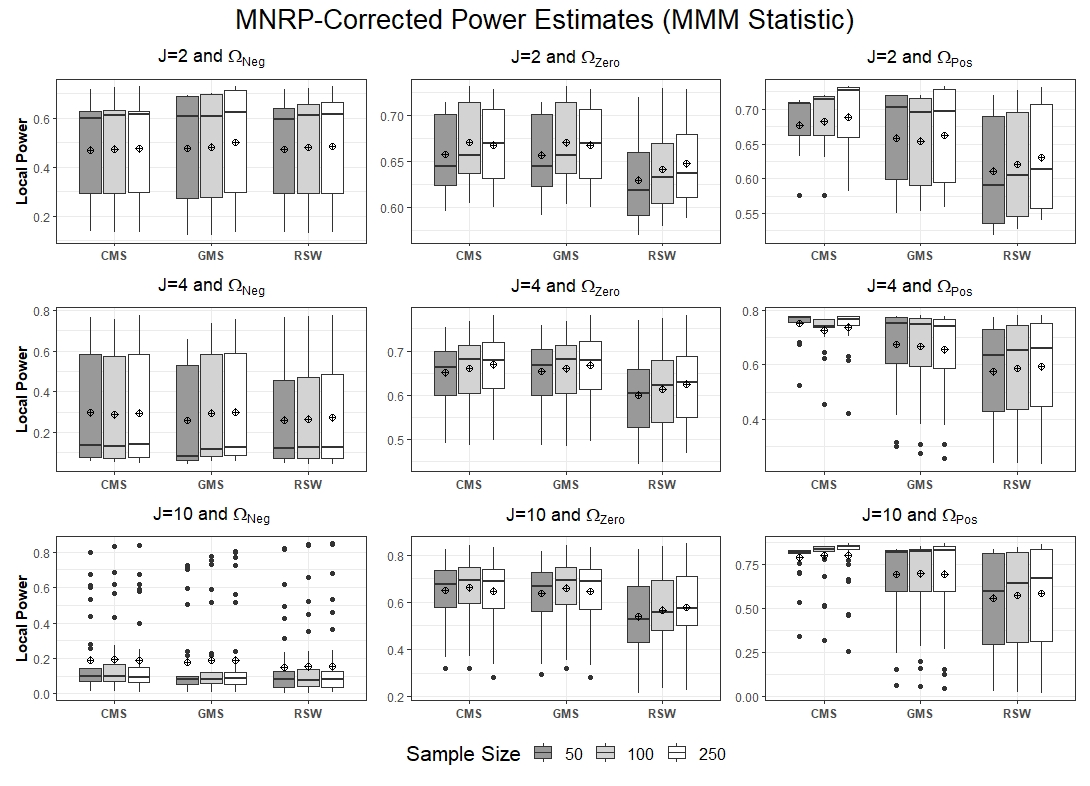}
\caption{Boxplots of MNRP-corrected powers of $S_{1}$-based tests. For each configuration, the symbol $\oplus$ marks the location of the average MNRP-corrected power of a test. }\label{Figure - Power MMM}
\end{figure}
\end{landscape}

\par The result of Theorem~\ref{LAthm} implies that the average power of CMS and GMS tests should get closer together with larger sample sizes. The simulations reflect this implication across all configurations, but indicate that it happens slowly when $\Omega=\Omega_{\text{Pos}}.$ Consequently, there is simulation-based evidence that shows the implementation of the information~(\ref{Moment inequality restrictions S3}), as we do with CMS, may not improve the local power of GMS tests for configurations in which $\Omega\neq\Omega_{\text{Pos}}.$ The reason is that the boxplots for MNRP-corrected power values of CMS and GMS tests are generally quite similar in those configurations. By contrast, the result of Theorem~\ref{LPorder} points to such an improvement in local power for configurations in which $\Omega=\Omega_{\text{Pos}}$, and this result is reflected in the simulations as described above.

\par Table~\ref{Table - LP Pos Omega} reports the average MNRP-corrected powers of the tests when $\Omega=\Omega_{\text{Pos}}$ and we use these results to further contextualize the local power improvement associated with CMS tests over GMS and RSW tests. We benchmark our analysis to RMS because simulation evidence in~\cite{andrews2012inference} suggests that it is superior in terms of asymptotic average power and is therefore the recommended test. The CMS-$S_{2A}$ and RMS tests are neck and neck as their average powers are essentially identical and achieve the highest average powers in all of those scenarios, with the CMS-$S_1$ test having slightly lower average powers than those tests. For a given $S$, the RSW tests are the worst performing, as they achieve the lowest average powers in each corresponding scenario, and the difference between them and the RMS test can be quite large. For example, when $J=10$ and $n=250$, the difference between RSW-$S_1$ and RMS is 0.252, and with RSW-$S_{2A}$ it is 0.055 which is a much smaller on account of using a more effective test statistic. The CMS-$S_1$ test dominates the GMS-$S_{1}$ test in each of those scenarios, where the difference can be as large as 10 percentage points -- see the scenarios with $J=10$. Consequently, the importance of incorporating the statistical information from the constraints, as we do with CMS, picks up the difference in average powers between the RMS and GMS test when $S=S_{2A}$ and most of the difference when $S=S_1,$ in each of the those scenarios.

\begin{table}[pt]
  \caption{MNRP-Corrected Average Powers: $\Omega = \Omega_{pos}$}\label{Table - LP Pos Omega}
     
\centering
\resizebox{13cm}{!}{
  \begin{tabular}{lccccccccc}
  \toprule
$J$ & $n$ & & GMS-$S_{1}$ & GMS-$S_{2A}$ & CMS-$S_{1}$ & CMS-$S_{2A}$ & RSW-$S_{1}$ & RSW-$S_{2A}$ & RMS \smallskip\smallskip  \\
 \midrule
& $50$ & & 0.658 &	0.676	&0.677	&0.691	&0.611	&0.652	&0.692  \smallskip\smallskip \\
2 & $100$ & & 0.654 &	0.681&	0.682	&0.697	&0.620&	0.661&	0.700 \smallskip\smallskip \\
 & $250$ & &0.662&	0.687	&0.689	&0.701&	0.630	&0.672&	0.709\smallskip\smallskip \\
 \midrule
 & $50$ & & 0.675	&0.743	&0.752&	0.759	&0.575	&0.715&	0.750\smallskip\smallskip \\
4 & $100$ & & 0.667	&0.747	&0.728&0.763&	0.587&	0.729&	0.760 \smallskip\smallskip \\
 & $250$ & & 0.656	&0.746&	0.738&	0.761&	0.593	&0.733&	0.761\smallskip\smallskip \\
 \midrule
 & $50$ & & 0.692	&0.782	&0.787&	0.803&	0.557	&0.746	&0.797 \smallskip\smallskip \\
10 & $100$ & & 0.696&	0.796&	0.799	&0.819&	0.573	&0.765&	0.825\smallskip\smallskip \\
 & $250$ & & 0.695	&0.810	&0.802	&0.830	&0.585	&0.782&	0.837\smallskip\smallskip \\
 \bottomrule

    \end{tabular}
    }
\end{table}

\par While the focus above has been on average power, for individual $\mu$ vectors the power differences can be massive with $\Omega=\Omega_{\text{Pos}}$. Consider, for example, the element $\mu/\sqrt{n}\in\mathcal{M}_{4,n}(\Omega_{\text{Pos}})$ with $\mu=(-2.4705,1,1,1)^{\intercal}$. This mean vector is an example of an SNVI local alternative. Table \ref{Table - Illustration Output} reports the MNRP-corrected power estimates for the tests under this local alternative for $n=50,100,250$. The estimates indicate that:
\begin{table}[pt]
  \caption{MNRP-Corrected Power: $\mu/\sqrt{n}\in\mathcal{M}_{4,n}(\Omega_{\text{Pos}})$ with $\mu=(-2.4705,1,1,1)^{\intercal}$.} \label{Table - Illustration Output}
  \smallskip   %  \begin{threeparttable}
\centering
\resizebox{13cm}{!}{
  \begin{tabular}{lccccccccc}
  \toprule
$J$ & $n$ & & GMS-$S_{1}$ & GMS-$S_{2A}$ & CMS-$S_{1}$ & CMS-$S_{2A}$ & RSW-$S_{1}$ & RSW-$S_{2A}$ & RMS \smallskip\smallskip  \\
 \midrule
  & $50$ &   & 0.3        &	0.668	     & 0.684	   & 0.733	      & 0.237	    & 0.633	       & 0.734  \smallskip\smallskip \\
4 & $100$ &  & 0.275      &	0.663        & 0.625	   & 0.726	      & 0.236       & 0.645        & 0.74 \smallskip\smallskip \\
  & $250$ &  & 0.256      &	0.665	     & 0.616	   & 0.718        &	0.234	    & 0.654        & 0.744\smallskip\smallskip \\
\bottomrule

    \end{tabular}
    }
\end{table}
 \begin{figure}[H]
\centering
\includegraphics[height=6.5cm, width=15cm]{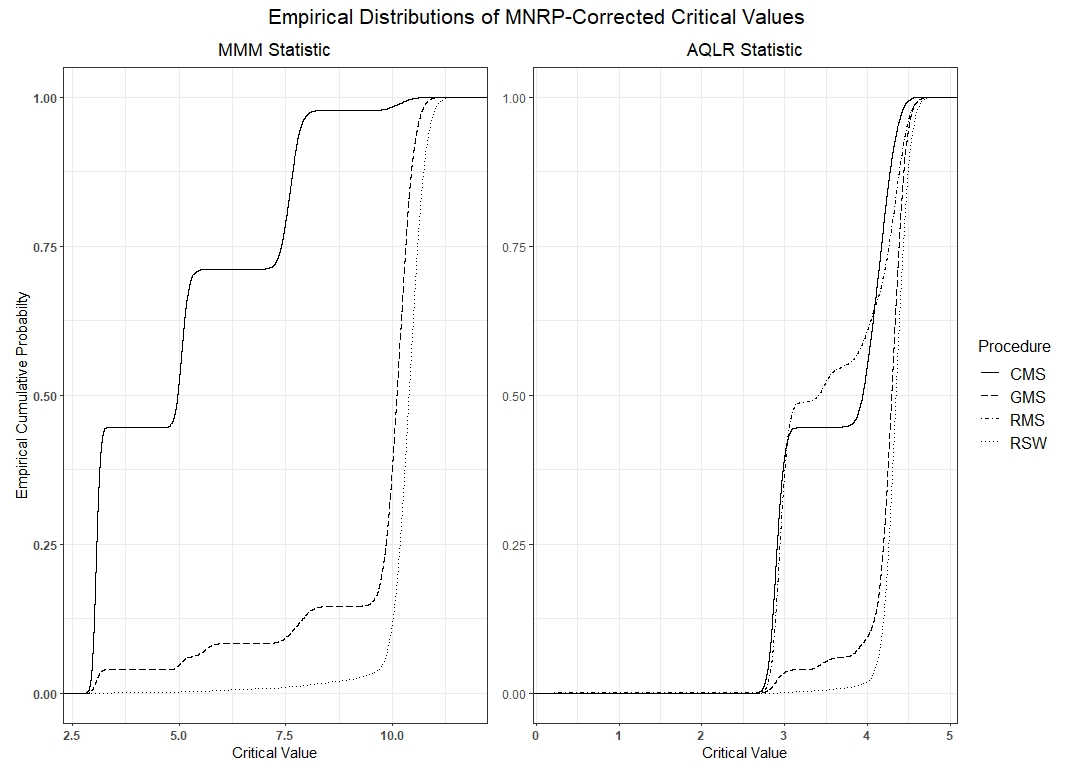}
\caption{ECDFs of MNRP-corrected CMS (solid line), GMS (dashed line), RMS (dash-dot line), and RSW (dotted line) critical values using the $S_1$ (MMM) and $S_{2A}$ (AQLR) test statistics.}\label{Figure - ECDFs_new}
\end{figure}

\begin{itemize}
\item There can be \textit{extremely} large power improvements associated with CMS relative to RSW and GMS when $S=S_{1}$. Indeed, the improvement in power of CMS over GMS is approximately 36 percentage points and 40 percentages points over RSW.
\item The improvements persist with $S=S_{2A}$, but are not as large. The AQLR statistic results in CMS experiencing a six percentage point improvement over GMS and eight percentages point improvement over RSW. In absolute terms, all procedures experience higher local power with $S=S_{2A}$.
\item The MNRP-corrected powers of CMS-$S_{2A}$ are comparable to their RMS counterparts.
\end{itemize}

 To gain a deeper insight into the behavior of the tests under this local alternative, Figure \ref{Figure - ECDFs_new} reports the empirical distribution functions (ECDFs) of the MNRP-corrected critical values for $n=250.$ The focus on this sample size is without loss of generality as similar graphs of the critical values' ECDFs arise in all of the other values of $n$ we considered. For either test statistic, the ECDFs in Figure~\ref{Figure - ECDFs_new} show strong evidence of a first-order stochastic dominance ranking among the critical values of the CMS, GMS, and RSW, tests. Specifically, for both types of test statistics, there is evidence for the ordering $\acute{c}_{n} \leq \hat{c}_{n} \leq \check{c}_{n}$, where $\check{c}_{n}$ denotes the RSW critical value. By contrast, the ECDF of the recommended RMS test crosses that of CMS with $S=S_{2A},$ which means that there isn't evidence of a clear ordering of their critical values. Overall, the differences between the ECDFs is quite striking and indicates that there is a big difference in the behavior of the tests even in moderately large sample sizes. The stochastic ordering of the CMS and GMS critical values is a reflection of Theorem \ref{result3} and provides evidence for local power improvements under SNVI local alternatives which have positively correlated moment functions. Finally, we discuss the behavior of the RSW procedure. The RSW procedure rejects on the event $\{M_{n}(\beta) \nsubseteq \mathbb{R}_{+}^{J}\}\bigcap \{S>\check{c}_{n}\}$, where $M_{n}(\beta)$ is a lower confidence rectangle that is used to detect ``positive'' moments in the first step of their two-step procedure (see Appendix \ref{RSWdetails}). Across the two test statistics, our simulations indicate that (i) the event $\{M_{n}(\beta) \nsubseteq \mathbb{R}_{+}^{J}\}$ occurs with empirical probability close to $1$, and (ii) in $9465$ times out of $10000$ Monte Carlo replications, their critical value $\check{c}_{n}$ corresponds to the case where none of the moment inequalities have been omitted from its calculation. These findings show the RSW procedure fails to reliably detect the ``positive'' moments in $\mu/\sqrt{n}$ in most of the $10000$ Monte Carlo replications, resulting in it having low empirical power.

\section{Conclusion}\label{Section - Conclusion}
\par This paper has proposed a surgical modification of the generalized moment selection (GMS) procedure put forward by~\cite{andrews2010inference} that improves its performance, called constrained moment selection (CMS). The basic idea of the CMS procedure is to use empirical likelihood to incorporate the information embedded in the moment inequality constraints into the moment selection step of the GMS procedure. Our analyses highlights the importance of using this information to more reliably detect the binding moments, which is the source of the improvement of CMS over GMS tests.

\par There are a number of directions for future research. Although we focus on modifying GMS tests, the intuition of incorporating the information embedded in the identified set transcends this choice and we conjecture that similar finite-sample benefits would arise in a similar modification to the two-step procedure of~\cite{romano2014practical}. There is also an emerging literature that focuses on testing with `many' moments, where the number of inequalities grow exponentially with the sample size (e.g.,~\citealp{chernozhukov2019inference}, and~\citealp{bai2019practical}). Extending the empirical likelihood modification to such testing procedures may improve their performance, but different theoretical tools must be employed to account for the increasing number of constraints. Finally, our paper is related to the semi-infinite programming empirical likelihood procedure proposed by~\cite{Lok-Tabri-inpress} for two-step bootstrap tests of stochastic dominance, where the continuum of unconditional moment inequalities is akin to inference for conditional moment inequalities. Their results are limited to restricted stochastic dominance tests and it would be interesting to extend the insights from this paper to the general conditional moment inequality models of~\cite{andrews2013inference, andrews2017inference}.

\section{Acknowledgements}
We are grateful to Jonathan Roth for providing valuable comments. We are also appreciative of feedback from participants at the Graduate Student Workshop in Econometrics, Harvard University. The computations in this paper were run on the FASRC Cannon cluster supported by the FAS Division of Science Research Computing Group at Harvard University. All errors are our own.
\bibliography{CMS}
\newpage
\appendix
\setcounter{page}{1}
\section{Outline}
This Appendix provides supplementary material to this paper. It is organized as follows.
\begin{itemize}
\item Section~\ref{Section - Test Stat Assump} lists the complete set of assumptions on the test statistic that~\cite{andrews2010inference} use in their work. We use these conditions in the proofs of the main results in the paper.
\item Section~\ref{Section - Proofs of Thms} presents the proofs of the results in the paper: Theorems~1, 2, 3, and 4.
\item Section~\ref{Section - Technical Lemmas for Thm 1} presents technical lemmas used in the proof of Theorem 1.
\item Section~\ref{Section Local Pwr Thms 2 3} presents technical lemmas used in the proofs of Theorems 2 and 3.
\item Section~\ref{RMSdetails} outlines the refined moment selection procedure of~\cite{andrews2012inference}.
\item Section~\ref{RSWdetails} outlines the two-step procedure of \cite{romano2014practical}.
\item Section~\ref{scdetails} details the MNRP corrections.
\end{itemize}
\section{Test Statistic Assumptions}\label{Section - Test Stat Assump}
\begin{assumption}\label{Test1}
\begin{enumerate}
\item Monotonicity: $S(g, \Sigma)$ is nonincreasing in $g$ for all $(g, \Sigma) \in \mathbb{R}^{J}\times \mathbb{R}^{J\times J}$.\label{T11}
\item  Invariance: $S(g, \Sigma) = S(Dg, D\Sigma D)$ for all $g \in \mathbb{R}^{J}$, $\Sigma \in \mathbb{R}^{J\times J}$ and positive definite diagonal matrix of $\Sigma$, $D \in \mathbb{R}^{J\times J}$.
\item Nonnegativity: $S(g, \Omega) \geq 0$ for all $(g, \Omega) \in \mathbb{R}^{J} \times \varPsi_{2}$.
\item Continuity: $S(g, \Omega)$ is a continuous function of $g \in \mathbb{R}^{J}$ and $\Omega \in \varPsi_{2}$.
\end{enumerate}
\end{assumption}
\begin{assumption}\label{T4}
For any $h_{1} \in \mathbb{R}_{+,\infty}^{J}$, $\Omega \in \varPsi_{2}$, $Z^{*} \sim N(0_{J}, I_{J})$, and $x \in \mathbb{R}$, the distribution function of $S(\Omega^{\frac{1}{2}}Z^{*} + h_{1}, \Omega)$ is 1. continuous at $x>0$, 2. strictly increasing in $x>0$ unless $h_{1}= [\infty,...,\infty]^{\top} \in \mathbb{R}_{+,\infty}^{J}$, and 3. does not exceed $1/2$ at $x=0$ when $h_{1} = 0_{J}$.
\end{assumption}
\begin{assumption}\label{T2}
A necessary and sufficient condition for $S(g, \Omega)>0$ is that there exists $j \in \mathcal{J}$ that satisfies $g_{j}<0$, where $g=(g_{1},...,g_{J})^{\top}$ and $\Omega \in \varPsi_{2}$.
\end{assumption}
\begin{assumption}\label{T5}
Let $Z^{*}\sim N(0_{J},I_{J})$, $\alpha \in (0,\frac{1}{2})$, and $c(\Omega, 1-\alpha)$ be the $(1-\alpha)$-quantile of the distribution of $S(\Omega^{\frac{1}{2}}Z^{*}, \Omega)$. We assume
\begin{enumerate}
\item The distribution function of $S(\Omega^{\frac{1}{2}}Z^{*}, \Omega)$ is continuous at $c(\Omega,1-\alpha)$ for all $\Omega \in \varPsi_{2}$.
\item $c(\Omega, 1-\alpha)$ is a uniformly continuous function of $\Omega \in \varPsi_{2}$.\footnote{We apply the definition of uniform continuity provided in \cite{rudin1976principles}. That is, a function $f: X \rightarrow Y$, where $(X,d_{X})$ and $(Y,d_{Y})$ are metric spaces, is uniformly continuous if $\forall \ \varepsilon > 0, \ \exists \ \delta:=\delta(\varepsilon)>0 \ s.t. \ \forall \ x,y \in X, \ d_{X}(x,y)<\delta \ \implies \ d_{Y}(f(x),f(y))< \varepsilon$.}
\end{enumerate}
\end{assumption}
\begin{assumption}\label{T6}
\begin{enumerate}
\item Let $v \in \mathbb{R}_{[+\infty]}^{J}$ and $\Omega \in \varPsi_{2}$ be arbitrary. The distribution function of $S(\Omega^{\frac{1}{2}}Z^{*}+v, \Omega)$ is a) continuous for $x > 0$ and is b) strictly increasing at $x>0$ unless $v = [\infty,...,\infty]^{\top}\in \mathbb{R}_{+,\infty}^{J}$.
\item For all $g_{1},g_{1}^{*} \in \mathbb{R}^{J}_{+,\infty}$ that satisfy $g_{1}^{*}\succ g_{1}$, we assume that $P(S(\Omega^{\frac{1}{2}}Z^{*}+g_{1}, \Omega)\leq x) < P(S(\Omega^{\frac{1}{2}}Z^{*}+g_{1}^{*}, \Omega)\leq x),$ where $x >0$.\footnote{The relation `$b\succ a$' means that every element in $a$ is less than or equal to every element in $b$ and the inequality holds strictly for each least one element.}
\end{enumerate}
\end{assumption}
\begin{assumption}
There exists $\chi>0$ such that for each $a\in \mathbb{R}_{++}$, $S(ag, \Omega)=a^{\chi}S(g,\Omega)$ for all $g \in \mathbb{R}^{J}$ and $\Omega \in \varPsi_{2}$.
\end{assumption}
\begin{assumption}\label{Assump test stat 7}
Let $h_{1,j}: \mathcal{F} \rightarrow \mathbb{R}_{+,\infty}$ given by $h_{1,j}(\theta,F) = \infty$ if $E_{F}(g_{j}(W_{i},\theta))>0$ and $h_{1,j}(\theta, F)=0$ if $E_{F}(g_{j}(W_{i},\theta))=0$, and define $h_{1}(\theta,F)=[h_{1,1}(\theta, F),...,h_{1,J}(\theta,F)]^{\top}$. Moreover, let $\Omega(\theta, F):= \lim_{n\rightarrow\infty}\text{Corr}_{F}(n^{\frac{1}{2}}\hat{g}_{n}(\theta))$. There exists $(\theta, F) \in \mathcal{F}$ such that the distribution of $$S(\Omega^{\frac{1}{2}}(\theta,F)Z^{*}+h_{1}(\theta,F), \Omega(\theta,F))$$ is continuous at its $1-\alpha$ quantile, where $Z^{*} \sim N(0_{J},I_{J})$.
\end{assumption}
\section{Proofs of Theorems}\label{Section - Proofs of Thms}
We introduce notation: $\hat{\varphi}_{n}(\theta):=\varphi\big(\hat{\xi}_{n}(\theta),\hat{\Omega}_{n}(\theta)\big)$, $\acute{\varphi}_{n}(\theta):=\varphi\big(\acute{\xi}_{n}(\theta),\hat{\Omega}_{n}(\theta)\big)$, $\hat{L}_{n}(\theta,Z^{*}):=S\big(\hat{\Omega}_{n}^{\frac{1}{2}}(\theta)Z^{*}+\varphi(\hat{\xi}_{n}(\theta), \hat{\Omega}_{n}(\theta)), \hat{\Omega}_{n}(\theta)\big)$, and $\hat{L}_{n}(\theta,Z^{*}):=S\big(\hat{\Omega}_{n}^{\frac{1}{2}}(\theta)Z^{*}+\varphi(\acute{\xi}_{n}(\theta), \hat{\Omega}_{n}(\theta)), \hat{\Omega}_{n}(\theta)\big)$ for each $\theta \in \Theta$. We are assuming that $\varphi = \varphi^{(1)}$; see the discussion in Section \ref{OtherGMS} for the general case.
\subsection{Theorem 1}\label{Appendix - proof Thm 1}

\begin{proof}
We present an outline of the proof and then the steps in detail.
\paragraph{Outline.} Lemma~\ref{null0} establishes the feasible set in the empirical likelihood optimisation problem~(\ref{eq - EL problem}) is non-empty with probability tending to one uniformly over $\mathcal{F}_{+}$. Consequently, the constrained estimator of the moments exists and is unique with probability tending to one uniformly over $\mathcal{F}_{+}$. With this technical result in mind, the proof has four steps. First, we show that $\{\acute{\varphi}_{n}(\theta) = \hat{\varphi}_{n}(\theta)\}$ occurs with probability approaching $1$ as $n\rightarrow+\infty$ with uniformity over $\mathcal{F}_{+}$. In the second step, we use the first result to show that for any $\alpha \in (0,\frac{1}{2})$ and any $r>0$, $\{|\acute{c}_{n}(\theta,1-\alpha)-\hat{c}_{n}(\theta,1-\alpha)|<r\}$ occurs with probability tending to $1$ as $n\rightarrow+\infty$, uniformly over $\mathcal{F}_{+}$. In the third step, we use step 2 to show that
\begin{align*}
\liminf_{n\rightarrow +\infty}\inf_{(\theta,F) \in \mathcal{F}_{+}}P_{F}\Big(T_{n}(\theta) \leq \acute{c}_{n}(\theta,1-\alpha)\Big) = \liminf_{n\rightarrow +\infty}\inf_{(\theta,F) \in \mathcal{F}_{+}}P_{F}\Big(T_{n}(\theta) \leq \hat{c}_{n}(\theta,1-\alpha)\Big) .
\end{align*}
In the final step, we prove all three statements in the theorem simultaneously by invoking Theorem 1 of \cite{andrews2010inference}.
\paragraph{Step 1.} The complement rule for probability measures implies that it suffices to show
\begin{align*}
\limsup_{n\rightarrow+\infty}\sup_{(\theta,F) \in \mathcal{F}_{+}}P_{F}\Big(\hat{\varphi}_{n}(\theta)\neq\acute{\varphi}_{n}(\theta)\Big)=0,
\end{align*}
which amounts to proving
\begin{align*}
\limsup_{n\rightarrow+\infty}\sup_{(\theta,F) \in \mathcal{F}_{+}}P_{F}\Big(\hat{\varphi}_{n,j}(\theta)\neq\acute{\varphi}_{n,j}(\theta)\Big)=0
\end{align*}
for each $j \in \{1,...,J\}$. Indeed, $\big\{\hat{\varphi}_{n}(\theta)\neq\acute{\varphi}_{n}(\theta)\big\} = \bigcup_{j=1}^{J}\big\{\hat{\varphi}_{n,j}(\theta)\neq\acute{\varphi}_{n,j}(\theta)\big\}$ implies
\begin{align*}
\limsup_{n\rightarrow+\infty}\sup_{(\theta,F) \in \mathcal{F}^{+}}P_{F}\bigg(\hat{\varphi}_{n}(\theta)\neq\acute{\varphi}_{n}(\theta)\bigg) \leq \sum_{j=1}^{J}\limsup_{n\rightarrow+\infty}\sup_{(\theta,F) \in \mathcal{F}^{+}}P_{F}\Big(\hat{\varphi}_{n,j}(\theta)\neq\acute{\varphi}_{n,j}(\theta)\Big)
\end{align*}
by the finite subadditivity of probability measures and basic properties of the supremum.
\paragraph{}
To this end, fix $j \in \{1,...,J\}$ arbitrarily. Recognizing that $\{\hat{\varphi}_{n,j}(\theta)\neq\acute{\varphi}_{n,j}(\theta)\}=\{\hat{\varphi}_{n,j}(\theta)>\acute{\varphi}_{n,j}(\theta)\}\bigcup\{\hat{\varphi}_{n,j}(\theta)<\acute{\varphi}_{n,j}(\theta)\}$, it follows that
\begin{align*}
P_{F}\Big(\hat{\varphi}_{n,j}(\theta)\neq\acute{\varphi}_{n,j}(\theta)\Big) &= P_{F}\Big(\hat{\varphi}_{n,j}(\theta)>\acute{\varphi}_{n,j}(\theta)\Big)+P_{F}\Big(\hat{\varphi}_{n,j}(\theta)<\acute{\varphi}_{n,j}(\theta)\Big) \\
&=P_{F}\Big(\hat{\xi}_{n,j}(\theta)>1, \ \acute{\xi}_{n,j}(\theta)\leq 1 \Big)  +P_{F}\Big( \hat{\xi}_{n,j}(\theta) \leq 1, \ \acute{\xi}_{n,j}(\theta)> 1 \Big) \\
&\leq P_{F}\Big(\hat{g}_{n,j}(\theta)> \acute{g}_{n,j}(\theta)\Big)+P_{F}\Big(\hat{g}_{n,j}(\theta) < \acute{g}_{n,j}(\theta)\Big) \\
&=P_{F}\Big(\hat{g}_{n,j}(\theta) \neq \acute{g}_{n,j}(\theta)\Big)
\end{align*}
where the second equality and the inequality hold by definition of $\varphi^{(1)}$. Lemma \ref{L2CS} then is invoked to establish that
\begin{align*}
\limsup_{n\rightarrow+\infty}\sup_{(\theta,F) \in \mathcal{F}_{+}}P_{F}\Big(\hat{g}_{n,j}(\theta)\neq \acute{g}_{n,j}(\theta)\Big)=0
\end{align*}
and therefore
\begin{align*}
\limsup_{n\rightarrow+\infty}\sup_{(\theta,F) \in \mathcal{F}_{+}}P_{F}\Big(\hat{\varphi}_{n,j}(\theta)\neq\acute{\varphi}_{n,j}(\theta)\Big)\leq \limsup_{n\rightarrow+\infty}\sup_{(\theta,F) \in \mathcal{F}_{+}}P_{F}\Big(\hat{g}_{n,j}(\theta)\neq \acute{g}_{n,j}(\theta)\Big)=0
\end{align*}
which completes the proof of Step 1.
\paragraph{Step 2.} We use step 1 to show that the event $\big\{\hat{c}_{n}(\theta,1-\alpha)\neq \acute{c}_{n}(\theta,1-\alpha)\big\}$ occurs with probability approaching $0$ as $n\rightarrow+\infty$, with uniformity over $\mathcal{F}_{+}$. This follows immediately from step 1 because
\begin{align*}
P_{F}\Big(\hat{c}_{n}(\theta,1-\alpha)\neq\acute{c}_{n}(\theta,1-\alpha)\Big)
\leq P_{F}\Big(\acute{\varphi}_{n}(\theta) \neq \hat{\varphi}_{n}(\theta)\Big)
\end{align*}
where the inequality holds because $\hat{L}_{n}(\theta,Z^{*})$ and $\acute{L}_{n}(\theta,Z^{*})$ only differ through the realization of the moment selection function a.s. $[Z^{*}]$. Step 1 and the squeeze rule then implies that $$\limsup_{n\rightarrow+\infty}\sup_{(\theta,F) \in \mathcal{F}_{+}}P_{F}(\hat{c}_{n}(\theta,1-\alpha) \neq \acute{c}_{n}(\theta,1-\alpha)) =0.$$
\paragraph{Step 3.} The result established in the second step allows us to conclude that $$\big\{T_{n}(\theta)\leq \acute{c}_{n}(\theta,1-\alpha)\big\} = \big\{T_{n}(\theta)\leq \hat{c}_{n}(\theta,1-\alpha)+o_{p}(1)\big\}$$ uniformly over $\mathcal{F}_{+}$. The uniformity implies that
\begin{align*}
\liminf_{n\rightarrow +\infty}\inf_{(\theta,F) \in \mathcal{F}_{+}}P_{F}\Big(T_{n}(\theta) \leq \acute{c}_{n}(\theta,1-\alpha)\Big) = \liminf_{n\rightarrow +\infty}\inf_{(\theta,F) \in \mathcal{F}_{+}}P_{F}\Big(T_{n}(\theta) \leq \hat{c}_{n}(\theta,1-\alpha)\Big).
\end{align*}
\paragraph{Step 4.} The previous step established that the asymptotic confidence sizes of GMS and CMS are equal. Combine this with the fact that $\mathcal{F}_{+}\subseteq \mathcal{F}$ and apply Theorem 1 in \cite{andrews2010inference} to conclude all three statements in the theorem simultaneously.
\end{proof}

\subsection{Theorem 2}\label{Appendix - proof Thm 2}
\begin{proof}
We present an outline of the proof and then the steps in detail.
\paragraph{Outline.} Lemma~\ref{LP0} establishes the feasible set in the empirical likelihood optimisation problem~(\ref{eq - EL problem}) is non-empty with probability tending to one under local alternatives that satisfy Assumption~LA1 and LA2, i.e., local alternatives in the set $\mathcal{H}$. Consequently, the constrained estimator of the moments exists and is unique with probability tending to one, under these local alternatives. With this technical result in mind, the proof has four steps and is similar to the proof of Theorem 1. First, we show $\{\acute{\varphi}_{n}(\theta_{n,*}) = \hat{\varphi}_{n}(\theta_{n,*})\}$ occurs with probability approaching $1$ as $n\rightarrow+\infty$ for any sequence $\{(\theta_{n,*},F_{n}): n \geq 1\}$. Next, we show that $\{\hat{c}_{n}(\theta_{n,*},1-\alpha)\neq \acute{c}_{n}(\theta_{n,*},1-\alpha)\}$ is an event that occurs with probability approaching $0$ as $n\rightarrow +\infty$ along $\{(\theta_{n,*},F_{n}): n \geq 1\}$. In the third step, we use the second step to conclude that $$\lim_{n\rightarrow+\infty}P_{F_{n}}\big(T_{n}(\theta_{n,*}) \leq \hat{c}_{n}(\theta_{n,*},1-\alpha)\big)=\lim_{n\rightarrow+\infty}P_{F_{n}}\big(T_{n}(\theta_{n,*}) \leq \acute{c}_{n}(\theta_{n,*},1-\alpha)\big)$$ for all sequences $\{(\theta_{n,*},F_{n}): n \geq 1\}$. In the fourth step, we invoke Part A of Theorem 2 in \cite{andrews2010inference} to establish the result.
\paragraph{Step 1.} Step 1 follows a similar line of reasoning to the same step in Theorem 1. We pick an arbitrary sequences of $n^{-\frac{1}{2}}$-local alternatives $\{(\theta_{n,*},F_{n}): n \geq 1\}$ and show that
\begin{align*}
\lim_{n\rightarrow+\infty}P_{F_{n}}\big(\hat{\varphi}_{n}(\theta_{n,*}) \neq \acute{\varphi}_{n}(\theta_{n,*})\big) = 0.
\end{align*}
To do this, we recognize that $\big\{\hat{\varphi}_{n}(\theta_{n,*}) \neq \acute{\varphi}_{n}(\theta_{n,*})\big\}=\bigcup_{j=1}^{J}\big\{\hat{\varphi}_{n,j}(\theta_{n,*}) \neq \acute{\varphi}_{n,j}(\theta_{n,*})\big\}$ and therefore that
\begin{align*}
P_{F_{n}}\Big(\hat{\varphi}_{n}(\theta_{n,*}) \neq \acute{\varphi}_{n}(\theta_{n,*})\Big)&\leq\sum_{j=1}^{J}P_{F_{n}}\Big(\hat{\varphi}_{n,j}(\theta_{n,*}) \neq \acute{\varphi}_{n,j}(\theta_{n,*})\Big) \\
&\leq \sum_{j=1}^{J}P_{F_{n}}\Big(\hat{g}_{n,j}(\theta_{n,*}) \neq \acute{g}_{n,j}(\theta_{n,*})\Big)
\end{align*}
using identical reasoning to the corresponding result in the proof of Theorem 1, except replace $\theta$ with $\theta_{n,*}$ and $F$ with $F_{n,*}$. It follows then that
\begin{align*}
\lim_{n\rightarrow+\infty}P_{F_{n}}\Big(\hat{\varphi}_{n}(\theta_{n,*}) \neq \acute{\varphi}_{n}(\theta_{n,*})\Big)\leq \sum_{j=1}^{J}\lim_{n\rightarrow+\infty}P_{F_{n}}\Big(\hat{g}_{n,j}(\theta_{n,*}) \neq \acute{g}_{n,j}(\theta_{n,*})\Big)=0
\end{align*}
where the second equality holds by Lemma \ref{LP4}.
\paragraph{Step 2.} The proof of step 2 is almost identical to step 2 in Theorem 1. We use the exact same reasoning as Step 2 of Theorem 1 to conclude that
\begin{align*}
P_{F_{n}}\Big(\acute{c}_{n}(\theta_{n,*},1-\alpha)\neq \hat{c}_{n}(\theta_{n,*},1-\alpha)\Big) \leq P_{F_{n}}\Big(\acute{\varphi}_{n}(\theta_{n,*}) \neq \hat{\varphi}_{n}(\theta_{n,*})\Big) \quad \forall \ n \geq 1
\end{align*}
and therefore that $\lim_{n\rightarrow+\infty}P_{F_{n}}\Big(\acute{c}_{n}(\theta_{n,*},1-\alpha)\neq \hat{c}_{n}(\theta_{n,*},1-\alpha)\Big)=0$ following step 1.
\paragraph{Step 3.} The result established in the second step allows us to conclude that $$\big\{T_{n}(\theta_{n,*})\leq \acute{c}_{n}(\theta_{n,*},1-\alpha)\big\} = \big\{T_{n}(\theta_{n,*})\leq \hat{c}_{n}(\theta_{n,*},1-\alpha)+o_{p}(1)\big\}$$ along any sequence $\{(\theta_{n,*},F_{n}) : n \geq 1\}$. As such,
\begin{align*}
\lim_{n\rightarrow +\infty}P_{F_{n}}\Big(T_{n}(\theta_{n,*}) \leq \acute{c}_{n}(\theta_{n,*},1-\alpha)\Big) = \lim_{n\rightarrow +\infty}P_{F_{n}}\Big(T_{n}(\theta_{n,*}) \leq \hat{c}_{n}(\theta_{n,*},1-\alpha)\Big)
\end{align*}
for all sequences $\{(\theta_{n,*},F_{n}): n \geq 1\}$.
\paragraph{Step 4.} The previous step established that the $n^{-\frac{1}{2}}$-local power functions of GMS and CMS are equivalent to first order. We can then apply Part A of Theorem 2 in \cite{andrews2010inference} to conclude the theorem.
\end{proof}

\subsection{Theorem 3}\label{Appendix - proof Thm 3 and Prop 1}
For the proof of Theorem 3, we let $A_{n,\alpha}^{*}$ denote the event
\begin{align*}
\bigg\{\acute{\Upsilon}_{n}(\theta_{n,*})\subsetneq \hat{\Upsilon}_{n}(\theta_{n,*})\bigg\} \bigcap \bigg\{\hat{c}_{n}(\theta_{n,*},1-\alpha)>0\bigg\} \bigcap \bigg\{\acute{c}_{n}(\theta_{n,*},1-\alpha)<T_{n}(\theta_{n,*}) \leq \hat{c}_{n}(\theta_{n,*},1-\alpha)\bigg\}.
\end{align*}
We also let $\{W_{i,n}: i \leq n\}$ denote the $n$th row of the triangular array induced by $\{(\theta_{n,*},F_{n}): n \geq 1\}$.
\subsubsection{Proof of Theorem 3}
\begin{proof}
We outline the argument and then prove the result in detail.
\paragraph{Outline.} Lemma~\ref{LP0} establishes the feasible set in the empirical likelihood optimisation problem~(\ref{eq - EL problem}) is non-empty with probability tending to one under local alternatives that satisfy Assumption~LA1 and LA2, i.e., local alternatives in the set $\mathcal{H}$. Consequently, the constrained estimator of the moments exists and is unique with probability tending to one, under these local alternatives. With this technical result in mind, the proof has three steps. First, we show $\big\{\acute{c}_{n}(\theta_{n,*},1-\alpha) \leq \hat{c}_{n}(\theta_{n,*},1-\alpha)\big\}$ occurs with probability tending to $1$ along any sequence $\{(\theta_{n,*},F_{n}): n \geq 1\} \in \mathcal{M}$. This allows us to conclude the first part of the theorem. In the second step, we show that the event $A_{n,\alpha}^{*}$ implies that $\big\{\acute{c}_{n}(\theta_{n,*},1-\alpha) < \hat{c}_{n}(\theta_{n,*},1-\alpha)\big\}$. In the final step, we conclude the strict ordering of the rejection probabilities.
\paragraph{Step 1.} Let $\{(\theta_{n,*}, F_{n}): n \geq 1\}\in \mathcal{M}$. Lemma \ref{GMSorder} states $\bigcap_{j=1}^{J}\big\{\acute{\varphi}_{n,j}(\theta_{n,*}) \geq \hat{\varphi}_{n,j}(\theta_{n,*})\big\}$
with probability approaching 1 along $\{(\theta_{n,*},F_{n}): n \geq 1\}$. It follows from Part 1 of Assumption 1 that $\big\{\acute{L}_{n}(\theta_{n,*},Z^{*}) \leq \hat{L}_{n}(\theta_{n,*},Z^{*}) \ a.s. \ [Z^{*}]\big\}$ with probability approaching $1$ under $\{(\theta_{n,*},F_{n}): n \geq 1\}$. Consequently, $\big\{\acute{c}_{n}(\theta_{n,*}, 1-\alpha) \leq \hat{c}_{n}(\theta_{n,*},1-\alpha)\big\}$ occurs with probability approaching 1 along $\{(\theta_{n,*},F_{n}): n \geq 1\}$. Thus, there exists $N(\theta_{n,*}, F_{n})\geq 1$ such that $P_{F_{n}}(T_{n}(\theta_{n,*}) > \hat{c}_{n}(\theta_{n,*},1-\alpha)) \leq P_{F_{n}}(T_{n}(\theta_{n,*}) > \acute{c}_{n}(\theta_{n,*}, 1-\alpha))$ for all $n \geq N(\theta_{n,*}, F_{n})$.
\paragraph{Step 2.} The event $A_{n,\alpha}^{*}$ implies the event $\big\{\acute{\Upsilon}_{n}(\theta_{n,*}) \subsetneq \hat{\Upsilon}_{n}(\theta_{n,*})\big\}\bigcap \big\{\hat{c}_{n}(\theta_{n,*},1-\alpha)>0\big\}$, which allows us to apply Part 1 of Assumption 2 and Part 2 of Assumption 5 to deduce that $1-\alpha=P\Big(\hat{L}_{n}(\theta_{n,*},Z^{*}) \leq \hat{c}_{n}(\theta_{n,*},1-\alpha)\Big)
<P\Big(\acute{L}_{n}(\theta_{n,*},Z^{*}) \leq \hat{c}_{n}(\theta_{n,*},1-\alpha)\Big) \ a.s. \ \big[\{W_{i,n}: i\leq n \}\big]$. Applying Part 1 of Assumption 2 again, we conclude that $A_{n,\alpha}^{*}\subseteq\big\{\acute{c}_{n}(\theta_{n,*},1-\alpha)<\hat{c}(\theta_{n,*},1-\alpha)\big\}$. This completes step 2.
\paragraph{Step 3.} Since $A_{n,\alpha}^{*}\subseteq \big\{\acute{c}_{n}(\theta_{n,*},1-\alpha)<T_{n}(\theta_{n,*}) \leq \hat{c}_{n}(\theta_{n,*},1-\alpha)\big\}$ by construction, we use Step 2 to deduce $A_{n,\alpha}^{*} \subseteq \big\{\acute{c}_{n}(\theta_{n,*},1-\alpha)<\hat{c}_{n}(1-\alpha,\theta_{n,*})\big\}\bigcap \big\{\acute{c}_{n}(\theta_{n,*},1-\alpha)<T_{n}(\theta_{n,*}) \leq \hat{c}_{n}(\theta_{n,*},1-\alpha)\big\}.$ Consequently, if $P_{F_{n}}(A_{n,\alpha}^{*})>0$ then the proof is complete because
\begin{align*}
& P_{F_{n}}(T_{n}(\theta_{n,*})>\acute{c}_{n}(\theta_{n,*},1-\alpha))-P_{F_{n}}(T_{n}(\theta_{n,*})>\hat{c}_{n}(\theta_{n,*},1-\alpha)) \\
&\quad = P_{F_{n}}\Big(\big\{\acute{c}_{n}(\theta_{n,*},1-\alpha)<\hat{c}_{n}(1-\alpha,\theta_{n,*})\big\}\bigcap \big\{\acute{c}_{n}(\theta_{n,*},1-\alpha)<T_{n}(\theta_{n,*}) \leq \hat{c}_{n}(\theta_{n,*},1-\alpha)\big\}\Big) \\
&\quad \geq P_{F_{n}}(A_{n,\alpha}^{*})
\end{align*}
where the inequality uses monotonicity of probability measures.
\end{proof}
\subsection{Theorem 4}\label{Appendix - proof Thm 4}
In the proof of Theorem 4, we use the notation $\nu_{n}(\theta_{n,*}) := D_{n}^{-\frac{1}{2}}(\theta_{n,*})n^{\frac{1}{2}}\big(\hat{g}_{n}(\theta_{n,*})- E_{F_{n}}g(W_{i},\theta_{n,*})\big)$.
\begin{proof}
Our approach is based on the proof for the corresponding result in \cite{andrews2010inference}. For ease of exposition, we outline the proof  and then provide the details.
\paragraph{Outline. } For $\{w_{n}: n \geq 1\}$ any subsequence of $\{n\}$, it suffices to show that there exists a further subsequence $\{u_{n}: n \geq 1\}$ such that $\lim_{n\rightarrow\infty}P_{F_{u_{n}}}\big(T_{u_{n}}(\theta_{u_{n},*})>\acute{c}_{u_{n}}(\theta_{u_{n},*},1-\alpha)\big)=1$. In Step 1, we define the sub-subsequence. In Step 2, we show that $(u_{n}^{\frac{1}{2}}\upsilon_{u_{n}})^{-\chi}T_{u_{n}}(\theta_{u_{n},*})$ has a positive probability limit, where $\chi>0$ is arbitrary. In Step 3, we show that the probability limit of $(u_{n}^{\frac{1}{2}}\upsilon_{u_{n}})^{-\chi}\acute{c}_{u_{n}}(\theta_{u_{n},*},1-\alpha)$ is zero. In the final step, we use Step 2 and Step 3 to establish that $\lim_{n\rightarrow\infty}P_{F_{u_{n}}}\big(T_{u_{n}}(\theta_{u_{n},*})>\acute{c}_{u_{n}}(\theta_{u_{n},*},1-\alpha)\big)=1$.
\paragraph{Step 1.} Consider any subsequence $\{w_{n}: n \geq 1\}$ of $\{n\}$. We take $\{u_{n}: n \geq 1\}$ so that $g_{u_{n}}^{*}/\upsilon_{u_{n}}\rightarrow e \in [-1,+\infty]^{J}$ as $n\rightarrow+\infty$, where $$g_{u_{n}}^{*}=[E_{F_{u_{n}}}(g_{1}(W_{i},\theta_{u_{n},*}))/\sigma_{F_{u_{n}},1}(\theta_{u_{n},*}),...,E_{F_{u_{n}}}(g_{J}(W_{i},\theta_{u_{n},*}))/\sigma_{F_{u_{n}},J}(\theta_{u_{n},*})]^{\top},$$ and $\upsilon_{u_{n}}=\max_{1 \leq j \leq J}\{-g_{u_{n},j}^{*}\}$. This is the sub-subsequence considered in \cite{andrews2010inference}.
\paragraph{Step 2.} Since we make no modification to the test statistic, we can follow the same argument as (S3.2) in the Supplement to \cite{andrews2010inference} to conclude that $(u_{n}^{\frac{1}{2}}\upsilon_{u_{n}})^{-\chi}T_{u_{n}}(\theta_{u_{n}}^{*})\overset{p}{\rightarrow} S(e,\Omega_{1})>0$, where the inequality holds by Assumption 3. The argument for the convergence in probability is provided below
\begin{align*}
(u_{n}^{\frac{1}{2}}\upsilon_{u_{n}})^{-\chi}T_{u_{n}}(\theta_{u_{n}}^{*})&=(u_{n}^{\frac{1}{2}}\upsilon_{u_{n}})^{-\chi}S\Big(\hat{D}_{u_{n}}^{-\frac{1}{2}}(\theta_{u_{n},*})D^{\frac{1}{2}}(\theta_{u_{n},*})\big(\nu_{u_{n}}(\theta_{u_{n},*})+u_{n}^{\frac{1}{2}}g_{u_{n}}^{*}\big),\hat{\Omega}_{u_{n}}(\theta_{u_{n},*})\Big) \\
&=S\Big(o_{p}(1)+\upsilon_{u_{n}}^{-1}g_{n}^{*},\Omega_{1}+o_{p}(1)\Big) \\
&\overset{p}{\rightarrow} S(e,\Omega_{1})
\end{align*}
where the first equality is algebraic manipulation and Part 2 of Assumption 1, the second equality is Assumption 6 and an application of the WLLN and Lyupanov CLT for triangular arrays of row-wise i.i.d. random variables and Part 2 of Distant Alternatives Assumption 1, and the convergence in probability holds by the construction of the sub-subsequence in Step 1 and Part 4 of Assumption 1. This completes Step 2.
\paragraph{Step 3.} We now establish that $(u_{n}^{\frac{1}{2}}\upsilon_{u_{n}})^{-\chi}\acute{c}_{u_{n}}(\theta_{u_{n},*},1-\alpha)=o_{p}(1)$ along $\{(\theta_{u_{n},*},F_{u_{n}}): n \geq 1\}$. Part 1 and 3 of Assumption 1 and the fact that $\varphi^{(1)} \in \mathbb{R}_{+,\infty}^{J}$ yield
\begin{align*}
0\leq S\big(\hat{\Omega}_{u_{n}}^{\frac{1}{2}}(\theta_{u_{n},*})Z^{*}+\varphi(\acute{\xi}_{u_{n}}(\theta_{u_{n},*}), \hat{\Omega}_{u_{n}}(\theta_{u_{n},*})), \hat{\Omega}_{u_{n}}(\theta_{u_{n},*})\big) \leq S(\hat{\Omega}_{u_{n}}^{\frac{1}{2}}(\theta_{u_{n},*})Z^{*}, \hat{\Omega}_{u_{n}}(\theta_{u_{n},*}))
\end{align*}
$a.s.$ $[Z^{*}]$. Consequently, the CMS critical value satisfies
\begin{align}\label{bigOp1}
0\leq\acute{c}_{u_{n}}(\theta_{u_{n},*},1-\alpha) \leq c(\hat{\Omega}_{u_{n}}(\theta_{u_{n},*}),1-\alpha) \overset{p}{\rightarrow}c(\Omega_{1},1-\alpha)=O_{p}(1)
\end{align}
where $c(\hat{\Omega}_{u_{n}}(\theta_{u_{n},*}),1-\alpha)$ is the $1-\alpha$ quantile of $S(\hat{\Omega}_{u_{n}}(\theta_{u_{n},*})Z^{*}, \hat{\Omega}_{u_{n}}(\theta_{u_{n},*}))$ and the convergence in probability holds by Part 2 of Assumption 4 and $\hat{\Omega}_{u_{n}}\overset{p}{\rightarrow}\Omega_{1}$ along $\{(\theta_{u_{n},*}, F_{u_{n}}): n \geq 1\}$ by the weak law of large numbers for triangular arrays of row-wise i.i.d. random variables and Part 2 of Distant Alternatives Assumption 1. Since $\upsilon_{u_{n}}>0$ for all $n \geq 1$ by construction, equation (\ref{bigOp1}) yields
\begin{align}
0\leq(u_{n}^{\frac{1}{2}}\upsilon_{u_{n}})^{-\chi}\acute{c}_{u_{n}}(\theta_{u_{n},*},1-\alpha) \leq (u_{n}^{\frac{1}{2}}\upsilon_{u_{n}})^{-\chi}c(\hat{\Omega}_{u_{n}}(\theta_{u_{n},*}),1-\alpha) \overset{p}{\rightarrow}0
\end{align}
because Part 1 of Distant Alternatives Assumption 1 states that $u_{n}^{\frac{1}{2}}\upsilon_{u_{n}}\rightarrow\infty$.
\paragraph{Step 4.} Combine Step 2 and Step 3 to conclude that
\begin{align}
&P_{F_{u_{n}}}(T_{u_{n}}(\theta_{u_{n},*})>\acute{c}_{u_{n}}(\theta_{u_{n},*},1-\alpha))\\
&=P_{F_{u_{n}}}((u_{n}^{\frac{1}{2}}\upsilon_{u_{n}})^{-\chi}T_{u_{n}}(\theta_{u_{n},*})>(u_{n}^{\frac{1}{2}}\upsilon_{u_{n}})^{-\chi}\acute{c}_{u_{n}}(\theta_{u_{n},*},1-\alpha))\\
&\rightarrow P(S(e,\Omega_{1})>0)=1
\end{align}
as $n\rightarrow\infty$, where the equality is invokes the scale equivariance of quantiles.
\end{proof}
\section{Technical Lemmas for Confidence Sets}\label{Section - Technical Lemmas for Thm 1}
\subsection{Establishing Uniformity}
The following lemma validates the subsequence approach to establishing uniformity.
\begin{lemma}\label{representation}
Let $\{V_{n}(\theta): n \geq 1\}$ be a sequence of events indexed by $\theta \in \Theta$. The following is true: $\liminf_{n\rightarrow+\infty}P_{F_{w_{n},h}}(V_{w_{n}}(\theta_{w_{n},h}))=1$ for any subsequence $\{(\theta_{w_{n},h},F_{w_{n},h}): n \geq 1\}$ in $\mathcal{F}_{+}$ implies $$\liminf_{n\rightarrow+\infty}\inf_{(\theta,F) \in \mathcal{F}_{+}}P_{F}(V_{n}(\theta))= 1.$$
\end{lemma}
\begin{proof}
We outline the argument and then provide the details.
\paragraph{Outline.} The proof employs the direct method. In the first step, we use the definition of infimum to construct a subsequence $\{(\tilde{\theta}_{w_{n},h}^{*},\tilde{F}_{w_{n},h}^{*}): n \geq 1\}$ in $\mathcal{F}_{+}$ such that for each $n \geq 1$,
\begin{align*}
\inf_{(\theta,F) \in \mathcal{F}_{+}}P_{F}(V_{w_{n}}(\theta))+2^{-w_{n}} > P_{\tilde{F}_{w_{n},h}^{*}}\big(V_{w_{n}}(\tilde{\theta}_{w_{n},h}^{*})\big).
\end{align*}
In the second step, we combine this with the assumption that $\liminf_{n\rightarrow+\infty}P_{F_{w_{n},h}}(V_{w_{n}}(\theta_{w_{n},h}))=1$ for any subsequence $\{(\theta_{w_{n},h},F_{w_{n},h}): n \geq 1\}$ in $\mathcal{F}_{+}$ to conclude the result.
\paragraph{Step 1.} As the smallest subsequential limit, the limit inferior implies the existence of a subsequence $\{w_{n}: n \geq 1\}$ of $\{n\}$ such that
\begin{align*}
\lim_{n\rightarrow+\infty}\inf_{(\theta,F) \in \mathcal{F}_{+}}P_{F}\big(V_{w_{n}}(\theta)\big)=\liminf_{n\rightarrow+\infty}\inf_{(\theta,F) \in \mathcal{F}_{+}}P_{F}\big(V_{n}(\theta)\big).
\end{align*}
Consider the subsequence $\big\{\inf_{(\theta,F) \in \mathcal{F}_{+}}P_{F}(V_{w_{n}}(\theta)): n \geq 1 \big\}$. For each $n \geq 1$ and each $\eta >0$, there exists $(\tilde{\theta}_{w_{n},h,\eta},\tilde{F}_{w_{n},h,\eta}) \in \mathcal{F}_{+}$ such that $\inf_{(\theta,F) \in \mathcal{F}_{+}}P_{F}\big(V_{w_{n}}(\theta)\big)+\eta > P_{F_{w_{n},h,\eta}}\big(V_{w_{n}}(\theta_{w_{n},h,\eta})\big)$, by definition of the infimum. Consequently, there exists a subsequence $\{(\tilde{\theta}_{w_{n},h}^{*},\tilde{F}_{w_{n},h}^{*}): n \geq 1\}$ in $\mathcal{F}_{+}$ that satisfies
\begin{align}\label{infimumfun}
\inf_{(\theta,F) \in \mathcal{F}_{+}}P_{F}\big(V_{w_{n}}(\theta)\big)+2^{-w_{n}} > P_{\tilde{F}_{w_{n},h}^{*}}\big(V_{w_{n}}(\tilde{\theta}_{w_{n},h}^{*})\big)
\end{align}
for each $n \geq 1$. This completes the first step.
\paragraph{Step 2.} If $\liminf_{n\rightarrow+\infty}P_{F_{w_{n},h}}(V_{w_{n}}(\theta_{w_{n},h}))=1$ for any subsequence $\{(\theta_{w_{n},h},F_{w_{n}}): n \geq 1\}$ in $\mathcal{F}_{+}$, then  $\liminf_{n\rightarrow+\infty}P_{\tilde{F}_{w_{n},h}^{*}}\big(V_{w_{n}}(\tilde{\theta}_{w_{n},h}^{*})\big)=1$ by construction. Taking the limit inferior on both sides of (\ref{infimumfun}), we conclude that $\liminf_{n\rightarrow +\infty}\inf_{(\theta,F) \in \mathcal{F}_{+}}P_{F}\big(V_{w_{n}}(\theta)\big)=1$ by the squeeze rule.
\end{proof}
\subsection{Restricted Estimator}
CMS is based on the following empirical likelihood primal problem,
\begin{align}\label{ELprimalCSineq}
\sup_{\mathbf{p}}\Bigg\{\sum_{i=1}^{n}\ln(p_{i}) : \sum_{i=1}^{n}p_{i}g_{j}(W_{i},\theta) \geq 0 \ \forall \ j \in \mathcal{J}, \ \sum_{i=1}^{n}p_{i}=1, \ p_{i} \geq 0 \ \forall \ i\in \mathcal{I}\Bigg\},
\end{align}
where $\mathbf{p} \in \mathbb{R}^{n}$ and $\mathcal{I}:=\{1,...,n\}$. A feasible solution to (\ref{ELprimalCSineq}) is denoted by $\acute{\mathbf{p}} \in \mathbb{R}^{n}$ and is the unique maximiser because the empirical likelihood problem is a strictly convex program (see \citealp{owen2001empirical}).
\paragraph{} We now establish that the feasible set is non-empty with probability tending to one uniformly over $\mathcal{F}_{+}$.
\begin{lemma}\label{null0}
Define the random set
\begin{align*}
\mathcal{C}_{n}(\theta)=\Bigg\{(p_{1},...,p_{n})^{\top} \in \mathbb{R}^{n}: \sum_{i=1}^{n}p_{i}g_{j}(W_{i},\theta) \geq 0 \ \forall \ j \in \mathcal{J}, \ \sum_{i=1}^{n}p_{i}=1, \ p_{i} \geq 0 \ \forall \ i\in \mathcal{I}\Bigg\}
\end{align*}
for all $\theta \in \Theta$. The following is true:
\begin{align*}
\limsup_{n\rightarrow\infty}\sup_{(\theta,F) \in \mathcal{F}_{+}}P_{F}(\mathcal{C}_{n}(\theta) = \emptyset)=0.
\end{align*}
\begin{proof}
We outline the proof and then provide the details.
\paragraph{Outline.}
The proof proceeds by the direct method and, in accordance with Lemma \ref{representation}, we only need to show that $\limsup_{n\rightarrow+\infty}P_{F_{w_{n},h}}(C_{w_{n}}(\theta_{w_{n},h})=\emptyset)=0$ for any subsequence $\{(\theta_{w_{n},h},F_{w_{n},h}): n \geq 1\}$ in $\mathcal{F}_{+}$. In the first step, we establish the result for sequences $\{(\theta_{n,h},F_{n,h}): n \geq 1\}$ in $\mathcal{F}_{+}$ using the union bound and the weak law of large numbers for triangular arrays of row-wise i.i.d. random variables. In Step 2, we generalize the argument to subsequences and complete the proof.
\paragraph{Step 1.}
We start by proving the result along sequences $\{(\theta_{n,h},F_{n,h}): n \geq 1\}$ in $\mathcal{F}_{+}$. Consider an arbitrary sequence $\{(\theta_{n,h},F_{n,h}): n \geq 1\}$ in $\mathcal{F}_{+}$. Recognizing that the standard simplex $\mathcal{S}_{n} = \{\mathbf{p} \in \mathbb{R}^{n}: \sum_{i=1}^{n}p_{i}=1, \ p_{i} \geq 0 \ \forall \ i\in \mathcal{I}\}\neq \emptyset$, it follows that
\begin{align*}
\Big\{\mathcal{C}_{n}(\theta_{n,h}) = \emptyset\Big\} = \Big\{\forall \ \mathbf{p} \in \mathcal{S}_{n}, \ \exists \ j:=j(\mathbf{p}) \in \mathcal{J} \ s.t. \ \sum_{i=1}^{n}p_{i}g_{j}(W_{i},\theta_{n,h})<0\Big\},
\end{align*}
and therefore
\begin{align*}
P_{F_{n}}(C_{n}(\theta_{n,h})=\emptyset) &\leq \sum_{j=1}^{J}P_{F_{n,h}}\bigg(\frac{1}{n}\sum_{i=1}^{n}g_{j}(W_{i},\theta_{n,h})<0\bigg)
\end{align*}
where the first inequality holds by the finite subadditivity of probability measures and because $(\frac{1}{n},...,\frac{1}{n})^{\top} \in \mathcal{S}_{n}$ for each $n \geq 1$. We then apply the weak law of large numbers for triangular arrays of row-wise i.i.d. random variables to conclude that
\begin{align*}
\limsup_{n\rightarrow+\infty}P_{F_{n,h}}(C_{n}(\theta_{n,h})=\emptyset) \leq \sum_{j=1}^{J}\limsup_{n\rightarrow +\infty}P_{F_{n,h}}\bigg(\frac{1}{n}\sum_{i=1}^{n}g_{j}(W_{i},\theta_{n,h})<0\bigg)=0
\end{align*}
where the equality holds because $E_{F_{n}}g_{j}(W_{i},\theta_{n,h}) \geq 0$ for each $n \geq 1$ because $\{(\theta_{n,h},F_{n}): n \geq 1\}$ is a sequence in $\mathcal{F}$. Since $\{(\theta_{n,h},F_{n}): n \geq 1\}$ was arbitrary, we establish the result along sequences.
\paragraph{Step 2.} To establish the result for subsequences $\{w_{n}: n \geq 1\}$ of $\{n\}$, just replaces $n$ with $w_{n}$ in the previous argument.
\end{proof}
\end{lemma}
\paragraph{}
In order to prove technical results, we reformulate the primal problem as one with equality constraints in order to make use of lemmas in \citet{andrews2009validity} (hereafter, AG09). Let $t\in\mathbb{R}^{J}_{+}$ denote a nuisance parameter vector where the $j^{\text{th}}$ element measures the slackness of corresponding moment. The vector $t$ allows us to formulate the empirical likelihood primal problem as a parameterized optimization problem as follows,
\begin{align}\label{ELprimalCSeq}
\mathcal{EL}(t):=\sup_{\mathbf{p}}\Bigg\{\sum_{i=1}^{n}\ln(p_{i}) :\sum_{i=1}^{n}p_{i}g_{i}(t,\theta) =0_{J}, \ \sum_{i=1}^{n}p_{i}=1, \ p_{i} \geq 0 \ \forall \ i\in \mathcal{I}\Bigg\},
\end{align}
where $g_{i}(t,\theta):=g(W_{i},\theta)-t$ and $0_{J}$ denotes the zero vector in $J$-dimensional Euclidean space and the empirical likelihood probabilities $(\acute{p}_{1},...,\acute{p}_{n})'$ are the solution to $\sup_{t \in \mathbb{R}_{+}^{J}}\mathcal{EL}(t)$.
\paragraph{}
A more convenient representation of the probabilities arises through the saddlepoint form of the empirical likelihood problem. The Lagrangian for the constrained optimization problem (\ref{ELprimalCSeq}) is
\begin{align}
\mathcal{L}=\sum_{i=1}^{n}\ln(p_{i})+n\lambda^{\top}\sum_{i=1}^{n}p_{i}g_{i}(t,\theta)+\omega\Bigg(\sum_{i=1}^{n}p_{i}-1\Bigg).
\end{align}
Note that the non-negativity constraints are ignored as $p_{i}=0$ for some $i \in \mathcal{I}$ is never optimal. The first order conditions are
\begin{align}
&\frac{\partial\mathcal{L}}{\partial p_{i}} = \frac{1}{p_{i}}+n\lambda^{\top}g_{i}(t,\theta)+\omega = 0 \ \forall \ i \in \mathcal{I} \label{FOC1}\\
&\frac{\partial\mathcal{L}}{\partial\lambda} = n\sum_{i=1}^{n}p_{i}g_{i}(t,\theta)=0_{J} \\
&\frac{\partial\mathcal{L}}{\partial\omega} = \sum_{i=1}^{n}p_{i}-1=0.
\end{align}
Multiplying $p_{i}$ with the corresponding first order condition in (\ref{FOC1}) and then summing over $\mathcal{I}$ gives $\omega = -n$. Substituting $\omega = -n$ into (\ref{FOC1}), we obtain that
\begin{align}\label{impliedprob}
p_{i}(\lambda,t) = \frac{1}{n\big(1-\lambda^{\top}g_{i}(t,\theta)\big)} \quad \forall \ i=1,...,n.
\end{align}
Substituting $(p_{1}(\lambda,t),...,p_{n}(\lambda,t))'$ into the empirical log-likelihood function implies the saddle point representation of the empirical likelihood problem
\begin{align}\label{ELdualCS}
\inf_{t \in \mathbb{R}^{J}_{+}}\sup_{\lambda \in \acute{\Lambda}_{n}}\frac{1}{n}\sum_{i=1}^{n}\ln\Big(1-\lambda^{\top}g_{i}(t,\theta)\Big),
\end{align}
where $\acute{\Lambda}_{n}:=\{\lambda \in \mathbb{R}^{J} : \lambda^{\top}g_{i}(t,\theta) \in Q\}$ and $Q$ is an open subset of $\mathbb{R}$ that contains $0$. The saddle point problem (\ref{ELdualCS}) is presented in AG09, which implies that the useful lemmas in that paper can be invoked to establish the uniform validity of CMS.
\subsection{Lemmas Relating to the Restricted Estimator}
The next results establish the uniform consistency of the restricted empirical likelihood estimator of the mean and variance over $\mathcal{F}_{+}$. We must define some more notation before proceeding. For any subsequence $\{(\theta_{w_{n},h},F_{w_{n},h}): n\geq 1\}$ in $\mathcal{F}_{+}$, let $(\acute{t}_{n},\acute{\lambda}_{n}) \in \mathbb{R}_{+}^{J} \times \acute{\Lambda}_{n}$ denote the solution to (\ref{ELdualCS}) evaluated at the subsequence (i.e. replace $\theta$ with $\theta_{w_{n,h}}$). The construction of the feasible set implies $\acute{t}_{w_{n}} = \sum_{i=1}^{n}\acute{p}_{i}g(W_{i},\theta_{w_{n},h})$. Lemma \ref{null0} establishes that the estimator exists with probability approaching $1$ uniformly over $\mathcal{F}_{+}$. All subsequent analysis assumes the event $\{\mathcal{C}_{n}(\theta) \neq \emptyset\}$ occurs so that the estimator exists, where the random set $\mathcal{C}_{n}(\theta)$ was defined in Lemma \ref{null0}.
\paragraph{}
Define an empirical process $\{\hat{g}_{n}(t): t \in \mathbb{R}_{+}^{J}\}$ given by $\hat{g}_{n}(t) = n^{-1}\sum_{i=1}^{n}g_{i}(t,\theta)$ for each $t \in \mathbb{R}_{+}^{J}$. Since $\mathcal{F}_{+}$ satisfies Assumption GEL of AG09, we invoke Lemma 6 and the subsequent remark in their paper and state that $\hat{g}_{w_{n}}(\acute{t}_{w_{n}})=O_{p}(w_{n}^{-\frac{1}{2}})$ for any subsequence $\{(\theta_{w_{n},h},F_{w_{n},h}): n \geq 1\}$ in $\mathcal{F}_{+}$. This gives us a uniform rate of convergence result for difference between the constrained and unconstrained estimator of the moments and in the statement $||\cdot||_{\ell^{2}_{J}}$ denotes the Euclidean norm for $\mathbb{R}^{J}$.
\begin{lemma}\label{L2CS}
Let $\acute{g}_{n}(\theta) := \sum_{i=1}^{n}\acute{p}_{i}g(W_{i},\theta)$ and $\hat{g}_{n}(\theta):=n^{-1}\sum_{i=1}^{n}g(W_{i},\theta)$ for each $\theta \in \Theta$. The following is true: $||\acute{g}_{n}(\theta)-\hat{g}_{n}(\theta)||_{\ell^{2}_{J}}=O_{p}(n^{-\frac{1}{2}})$ uniformly over $\mathcal{F}_{+}$.
\begin{proof}
The proof follows by the direct method. Since we want to show that $\acute{g}_{n}(\theta)-\hat{g}_{n}(\theta) = O_{p}(n^{-\frac{1}{2}})$ with uniformity over $\mathcal{F}_{+}$, it suffices to show that $\acute{g}_{w_{n}}(\theta_{w_{n},h})-\hat{g}_{w_{n}}(\theta_{w_{n},h})=O_{p}(w_{n}^{-\frac{1}{2}})$ for all subsequences $\{(\theta_{w_{n},h},F_{w_{n},h}): n \geq 1\}$ (see Lemma \ref{representation}). Observing that $\hat{g}_{w_{n}}(\theta_{w_{n},h})-\acute{g}_{w_{n}}(\theta_{w_{n},h})=\hat{g}_{w_{n}}(\acute{t}_{w_{n}})$ for any subsequence $\{(\theta_{w_{n},h},F_{w_{n},h}): n \geq 1\}$, we apply Lemma 6 of AG09 and conclude that $\hat{g}_{n}(\theta_{w_{n},h})-\acute{g}_{n}(\theta_{w_{n},h})=O_{p}(w_{n}^{-\frac{1}{2}})$, which completes the proof.
\end{proof}
\end{lemma}
For the next lemma, we must introduce some more notation. Let $\Mat_{J\times J}(\mathbb{R})$ denote the vector space of $J\times J$ matrices over $\mathbb{R}$. For each $A \in \Mat_{J\times J}(\mathbb{R})$, let $$||A||_{\ell^{2}_{J\times J}}:= \Bigg(\sum_{i=1}^{J}\sum_{j=1}^{J}|a_{ij}|^{2}\Bigg)^{\frac{1}{2}}.$$ This is the \textit{Frobenius norm}. We let
\begin{align*}
\acute{\Sigma}_{n}(\theta) := \sum_{i=1}^{n}\acute{p}_{i}\big(g(W_{i},\theta)-\acute{g}_{n}(\theta)\big)\big(g(W_{i},\theta)-\acute{g}_{n}(\theta)\big)^{\top}
\end{align*}
denote the constrained estimator of the moment covariance matrix and
\begin{align*}
\hat{\Sigma}_{n}(\theta):=\frac{1}{n}\sum_{i=1}^{n}\big(g(W_{i},\theta)-\acute{g}_{n}(\theta)\big)\big(g(W_{i},\theta)-\acute{g}_{n}(\theta)\big)^{\top}
\end{align*}
denote the unconstrained estimator of the moment covariance matrix for each $\theta \in \Theta$.
\begin{lemma}\label{unifcov}
For each $r>0$,
\begin{align*}
\liminf_{n\rightarrow+\infty}\inf_{(\theta,F) \in \mathcal{F}_{+}}P_{F}\Big(\big | \big | \acute{\Sigma}_{n}(\theta)-\hat{\Sigma}_{n}(\theta) \big | \big |_{\ell^{2}_{J \times J}}<r \Big) = 1
\end{align*}
\begin{proof}
Due to the length of the proof, we provide an outline and then detailed steps.
\paragraph{Outline.}
In accordance with Lemma \ref{representation}, it suffices to show that $\acute{\Sigma}_{w_{n}}(\theta_{w_{n},h}) = \hat{\Sigma}_{w_{n}}(\theta_{w_{n},h})+o_{p}(1)$ for any subsequence $\{(\theta_{w_{n},h},F_{w_{n},h}): n \geq 1\}$ in $\mathcal{F}_{+}$. To do this, we first prove the result along sequences. First, we establish a preliminary result that states that $\max_{1 \leq i \leq n}|\acute{\lambda}_{n}^{\top}g_{i}(\acute{t}_{n},\theta_{n,h})|=o_{p}(1)$ for any arbitrary sequence $\{(\theta_{n,h},F_{n}): n \geq 1\}$ in $\mathcal{F}_{+}$. In Step 2, we show that $\acute{\Sigma}_{n}(\theta_{n,h})=\sum_{i=1}^{n}\acute{p}_{i}\big(g(W_{i},\theta_{n,h})-\hat{g}_{n}(\theta_{n,h})\big)\big(g(W_{i},\theta_{n,h})-\hat{g}_{n}(\theta_{n,h})\big)^{\top}+o_{p}(1)$ holds for an arbitrary sequence $\{(\theta_{n,h},F_{n,h}): n \geq 1\}$ in $\mathcal{F}_{+}$. In Step 3, we use Step 1 to show that $\sum_{i=1}^{n}\acute{p}_{i}\big(g(W_{i},\theta_{n,h})-\hat{g}_{n}(\theta_{n,h})\big)\big(g(W_{i},\theta_{n,h})-\hat{g}_{n}(\theta_{n,h})\big)^{\top}=\hat{\Sigma}_{n}(\theta_{n,h})+o_{p}(1)$ along the sequence $\{(\theta_{n,h},F_{n,h}): n \geq 1\}$. This completes the proof for sequences. In Step 4, we generalize the result to subsequences $\{w_{n}: n \geq 1\}$ of $\{n\}$.
\paragraph{Step 1.} We first establish the preliminary result that $\max_{1 \leq i \leq n}|\acute{\lambda}_{n}^{\top}g_{i}(\acute{t}_{n},\theta_{n,h})|=o_{p}(1)$ along $\{(\theta_{n,h},F_{n,h}): n \geq 1\}$ in $\mathcal{F}_{+}$. The Cauchy-Schwarz inequality yields that
\begin{align}\label{bounduniform}
\max_{1 \leq i \leq n}|\acute{\lambda}_{n}^{\top}g_{i}(\acute{t}_{n},\theta_{n,h})|\leq ||\acute{\lambda}_{n}||_{\ell^{2}_{J}}\max_{1 \leq i \leq n}||g_{i}(\acute{t}_{n},\theta_{n,h})||_{\ell^{2}_{J}}.
\end{align}
Assumption T lets us apply Part (ii) Lemma 3 of AG09 that states $$\max_{1 \leq i \leq n}||g_{i}(\acute{t}_{n},\theta_{n,h})||_{\ell^{2}_{J}}=O_{p}(n^{\frac{1}{2+\delta}})$$ and also apply Lemma 5 in AG09 that states $||\acute{\lambda}_{n}||_{\ell^{2}_{J}}=O_{p}(n^{-\frac{1}{2}})$ along $\{(\theta_{n,h},F_{n,h}): n \geq 1\}$. Combining these with (\ref{bounduniform}), we deduce that
\begin{align*}
\max_{1 \leq i \leq n}|\acute{\lambda}_{n}^{\top}g_{i}(\acute{t}_{n},\theta_{n,h})| \leq O_{p}(n^{-\frac{1}{2}})O_{p}(n^{\frac{1}{2+\delta}}) =O_{p}(n^{-\frac{\delta}{4(1+\delta)}}) = o_{p}(1)
\end{align*}
along $\{(\theta_{n,h},F_{n,h}): n \geq 1\}$. The result established is essential in the third step.
\paragraph{Step 2.} We can decompose $g(W_{i},\theta_{n,h})-\acute{g}_{n}(\theta_{n,h}) = g(W_{i},\theta_{n,h})-\hat{g}_{n}(\theta_{n,h})+\hat{g}_{n}(\theta_{n,h})-\acute{g}_{n}(\theta_{n,h})$ and therefore
\begin{align*}
\acute{\Sigma}_{n}(\theta_{n,h}) &= \sum_{i=1}^{n}\acute{p}_{i}\big(g(W_{i},\theta_{n,h})-\hat{g}_{n}(\theta_{n,h})\big)\big(g(W_{i},\theta_{n,h})-\hat{g}_{n}(\theta_{n,h})\big)^{\top} \\
&+\sum_{i=1}^{n}\acute{p}_{i}\big(g(W_{i},\theta_{n,h})-\hat{g}_{n}(\theta_{n,h})\big)\big(\hat{g}_{n}(\theta_{n,h})-\acute{g}_{n}(\theta_{n,h})\big)^{\top} \\
&+\big(\hat{g}_{n}(\theta_{n,h})-\acute{g}_{n}(\theta_{n,h})\big)\sum_{i=1}^{n}\acute{p}_{i}\big(g(W_{i},\theta_{n,h})-\hat{g}_{n}(\theta_{n,h})\big)^{\top} \\
&+\big(\hat{g}_{n}(\theta_{n,h})-\acute{g}_{n}(\theta_{n,h})\big)\big(\hat{g}_{n}(\theta_{n,h})-\acute{g}_{n}(\theta_{n,h})\big)^{\top} \\
&=\sum_{i=1}^{n}\acute{p}_{i}\big(g(W_{i},\theta_{n,h})-\hat{g}_{n}(\theta_{n,h})\big)\big(g(W_{i},\theta_{n,h})-\hat{g}_{n}(\theta_{n,h})\big)^{\top}+o_{p}(1)
\end{align*}
along $\{(\theta_{n,h},F_{n}): n \geq 1\}$, where the second equality holds by Lemma \ref{L2CS}. Consequently, we need to show that $\sum_{i=1}^{n}\acute{p}_{i}\big(g(W_{i},\theta_{n,h})-\hat{g}_{n}(\theta_{n,h})\big)\big(g(W_{i},\theta_{n,h})-\hat{g}_{n}(\theta_{n,h})\big)^{\top}=\hat{\Sigma}_{n}(\theta_{n,h})+o_{p}(1)$ along $\{(\theta_{n,h},F_{n,h}): n \geq 1\}$. This completes the task for Step 2.
\paragraph{Step 3.} Let $A_{i}(\theta_{n,h}):=\big(g(W_{i},\theta_{n,h})-\hat{g}_{n}(\theta_{n,h})\big)\big(g(W_{i},\theta_{n,h})-\hat{g}_{n}(\theta_{n,h})\big)^{\top}$. The decomposition
\begin{align*}
\sum_{i=1}^{n}\acute{p}_{i}A_{i}(\theta_{n,h}) = \sum_{i=1}^{n}\bigg(\acute{p}_{i}-\frac{1}{n}\bigg)A_{i}(\theta_{n,h})+\hat{\Sigma}_{n}(\theta_{n,h})+o_{p}(1)
\end{align*}
means that we need to show $\sum_{i=1}^{n}(\acute{p}_{i}-\frac{1}{n})A_{i}(\theta_{n,h}) = o_{p}(1)$ along $\{(\theta_{n,h},F_{n}): n \geq 1\}$. By definition of $\acute{p}_{i}$, it follows that
\begin{align*}
\Bigg | \Bigg | \sum_{i=1}^{n}\bigg(\acute{p}_{i}-\frac{1}{n}\bigg)A_{i}(\theta_{n,h})\Bigg | \Bigg |_{\ell^{2}_{J \times J}} &\leq \frac{1}{n}\sum_{i=1}^{n}\bigg | \frac{\acute{\lambda}_{n}^{\top}g_{i}(\acute{t}_{n},\theta_{n,h})}{1-\acute{\lambda}_{n}^{\top}g_{i}(\acute{t}_{n},\theta_{n,h})}\bigg | \big | \big | A_{i}(\theta_{n,h}) \big | \big |_{\ell^{2}_{J\times J}} \\
&\leq \frac{\max_{1 \leq i \leq n}|\acute{\lambda}_{n}^{\top}g_{i}(\acute{t}_{n},\theta_{n,h})|}{1-\max_{1 \leq i \leq n}|\acute{\lambda}_{n}^{\top}g_{i}(\acute{t}_{n},\theta_{n,h})|}\Bigg(\frac{1}{n}\sum_{i=1}^{n}||A_{i}(\theta_{n,h})||_{\ell_{J\times J}^{2}}\Bigg) \\
&=o_{p}(1)O_{p}(1) \\
&=o_{p}(1)
\end{align*}
where the first inequality holds by the triangle inequality, the second holds by the reverse triangle inequality and the definition of maximum, and the first equality holds by Step 1 and the weak law of large numbers for triangular arrays of row-wise i.i.d. random variables. This establishes that $\sum_{i=1}^{n}\big(\acute{p}_{i}-\frac{1}{n}\big)A_{i}(\theta_{n,h})=o_{p}(1)$ along $\{(\theta_{n,h},F_{n}): n \geq 1\}$ and, in combination with the result in Step 2, we conclude that $\acute{\Sigma}_{n}(\theta_{n,h}) = \hat{\Sigma}_{n}(\theta_{n,h})+o_{p}(1)$ along sequences $\{(\theta_{n,h},F_{n,h}): n \geq 1\}$ in $\mathcal{F}_{+}$.
\paragraph{Step 3.} To generalize to subsequences $\{w_{n}: n \geq 1\}$ of $\{n\}$, just replace $n$ with $w_{n}$ and repeat Steps 1, 2 and 3.
\end{proof}
\end{lemma}
\section{Technical Lemmas for Local Power}\label{Section Local Pwr Thms 2 3}
\subsection{A Preliminary Lemma}
This technical result shows that along any sequence of $\{(\theta_{n,*},F_{n}): n \geq 1\}$ that satisfies Assumption LA1, $\max_{1 \leq i \leq n }|g_{j}(W_{i},\theta_{n,*})|=O_{p}(n^{\frac{1}{2+\delta}})=o_{p}(n^{\frac{1}{2}})$ for any $j \in \mathcal{J}$, where $\delta>0$ is defined in Assumption LA1. The practical consequence is that $\max_{1 \leq i \leq n}\max_{1\leq j \leq J}|g_{j}(W_{i},\theta_{n,*})|=o_{p}(n^{\frac{1}{2}})$ and, by the equivalence of norms in Euclidean space, $\max_{1 \leq i \leq n}||g_{j}(W_{i},\theta_{n,*})||_{\ell^{2}_{J}}=o_{p}(n^{\frac{1}{2}})$. The result is used in the proofs of Lemma \ref{LP3}, Lemma \ref{LP4}, and Lemma \ref{LPC1}.
\begin{lemma}\label{owen112}
For any sequence $\{(\theta_{n,*},F_{n}): n \geq 1\}$ of $n^{-\frac{1}{2}}$-local alternatives that satisfies LA1, the following is true: $\max_{1 \leq i \leq n}|g_{j}(W_{i},\theta_{n,*})|=O_{p}(n^{\frac{1}{2+\delta}})$ for each $j \in \mathcal{J}$, where $\delta>0$ is defined in Assumption LA1.
\begin{proof}
We outline the proof and then provide the steps.
\paragraph{Outline.} The proof is similar to that of equation~(2.4) in \cite{guggenberger2005generalized}. In the first step, we choose an appropriate $C>0$. In the second step, we apply the union bound and Markov's inequality to establish the result.
\paragraph{Step 1.} Fix $r>0$, $j \in \mathcal{J}$, and $\{(\theta_{n,*},F_{n}): n \geq 1\}$ arbitrarily. We know that $$K:=\sup_{n \geq 1}E_{F_{n}}|g_{j}(W_{i},\theta_{n,*})|^{2+\delta}<+\infty$$ by Assumption LA1 and can therefore choose $C>0$ so that $K/C<r$. Such a constant $C>0$ exists. For example, one may pick a sufficiently large natural number.
\paragraph{Step 2.} We know that
\begin{align*}
P_{F_{n}}\Big(\max_{1 \leq i \leq n}|g_{j}(W_{i},\theta_{n,*})|\leq (C n)^{\frac{1}{2+\delta}}\Big)\leq \sum_{i=1}^{n}P_{F_{n}}\Big(|g_{j}(W_{i},\theta_{n,*})|^{2+\delta}>n C \Big)\leq \frac{K}{C} <r
\end{align*}
where the first inequality applies the union bound, the second follows from Markov's inequality and taking the supremum of $\{E_{F_{n}}|g_{j}(W_{i},\theta_{n,*})|^{2+\delta}: n \geq 1\}$, and the third holds by the construction of $C$. We have constructed an upper bound (i.e. $K/C$) that is does not depend on $n$, so we can take the supremum to conclude that
\begin{align*}
\sup_{n \geq 1}P_{F_{n}}\Big(\max_{1 \leq i \leq n}|g_{j}(W_{i},\theta_{n,*})|\leq (C n)^{\frac{1}{2+\delta}}\Big)<r
\end{align*}
and complete the proof.
\end{proof}
\end{lemma}
\subsection{Restricted Estimator Under Local Alternatives}
The restricted empirical likelihood problem is
\begin{align}
\sup_{p_{1},...,p_{n}} \Bigg \{\sum_{i=1}^{n}\ln(p_{i}) \bigg | \sum_{i=1}^{n}p_{i}g(W_{i},\theta_{n,*}) \geq 0_{J}, \ \sum_{i=1}^{n}p_{i}=1, \ p_{i} \geq 0 \ \forall \ i=1,...,n\Bigg\}.
\end{align}
The Lagrangian is
\begin{align}\label{LPEL1}
\mathcal{L}(p_{1},...,p_{n}, \mathbf{\lambda}(\theta_{n,*}), \omega(\theta_{n,*})) = \sum_{i=1}^{n}\ln(p_{i})+\omega\Bigg(1-\sum_{i=1}^{n}p_{i}\Bigg) -n\lambda^{\top}\Bigg(\sum_{i=1}^{n}p_{i}g(W_{i}, \theta_{n,*})\Bigg)
\end{align}
and the Karusch-Kuhn-Tucker (KKT) conditions are
\begin{align}
&\frac{\partial \mathcal{L}}{\partial p_{i}} = \frac{1}{p_{i}}-\omega- n \lambda'g(W_{i}, \theta_{n,*})=0, \quad \forall \ i=1,...,n \label{KKT1} \\
&\lambda_{j} \leq 0, \quad \sum_{i=1}^{n}p_{i}g_{j}(W_{i}, \theta_{n,*})\geq 0, \quad \forall \ j \in \mathcal{J} \label{KKT2} \\
&\sum_{i=1}^{n}p_{i}=1, \quad \lambda_{j}\sum_{i=1}^{n}p_{i}g_{j}(W_{i}, \theta_{n,*})=0 \quad \forall \ j \in \mathcal{J}. \label{KKT3}
\end{align}
From the Karusch-Kuhn-Tucker conditions, we have that
\begin{align}\label{LPE2}
\acute{p}_{i} = \frac{1}{n}\bigg (\frac{1}{1+(\acute{\lambda}_{n,b})^{\top}g_{b}(W_{i},\theta_{n,*})}\bigg)
\end{align}
where $g_{b}(W_{i},\theta_{n,*})$ denotes the vector of estimating functions for the moments that are deemed binding by the Karusch-Kuhn-Tucker conditions and $\acute{\lambda}_{n,b}$ is the vector of Lagrange multipliers that corresponds to $g_{b}(W_{i},\theta_{n,*})$. Substituting (\ref{LPE2}) into (\ref{LPEL1}), we obtain the dual representation of the empirical likelihood problem,
\begin{align}\label{LPdual}
\sup_{\lambda \in \mathbb{R}^{J}_{-}}\Bigg\{n\ln(n)+\sum_{i=1}^{n}\ln\big(1+\lambda^{\top}g(W_{i},\theta_{n,*})\big)\Bigg\}.
\end{align}
The existence of Lagrange multipliers holds because of the fact that the constraints are affine functions of the choice variables in the primal problem (\ref{LPEL1}).
\subsection{Technical Results Relating to the Constrained Estimator}
Recall that the set $\mathcal{H}$ is defined as the set of all local alternatives $\{(\theta_{n,*},F_{n}): n \geq 1\}$ that satisfy Assumptions~LA1 and LA2.
\begin{lemma}\label{LP0}
For each $\{(\theta_{n,*},F_{n}): n \geq 1\} \in \mathcal{H}$, define the random set
\begin{align*}
\mathcal{C}_{n}(\theta_{n,*})=\Bigg\{(p_{1},...,p_{n})^{\top} \in \mathbb{R}^{n}: \sum_{i=1}^{n}p_{i}g_{j}(W_{i},\theta_{n,*}) \geq 0 \ \forall \ j \in \mathcal{J}, \ \sum_{i=1}^{n}p_{i}=1, \ p_{i} \geq 0 \ \forall \ i\in \mathcal{I}\Bigg\}.
\end{align*}
Then $\lim_{n\rightarrow+\infty}P_{F_{n}}\big(\mathcal{C}_{n}(\theta_{n,*}) = \emptyset\big) = 0$ for each $\{(\theta_{n,*},F_{n}): n \geq 1\} \in \mathcal{H}$.
\begin{proof}
We start with an outline and the provide the details.
\paragraph{Outline.} The proof proceeds by the direct method. In Step 1, we establish we establish that it suffices to show that $P_{F_{n}}(\hat{g}_{n,j}(\theta_{n,*}) < 0) \rightarrow 0$ as $n\rightarrow +\infty$ for each $j \in \mathcal{J}$. In Step 2, we establish the result using a mean-value expansion and the WLLN for triangular arrays of row-wise i.i.d. random variables.
\paragraph{Step 1.} The proof of the first step follows the a similar argument to that of Lemma \ref{null0}. We know that $\big\{\mathcal{C}_{n}(\theta_{n,*})=\emptyset\big\} = \Big\{\forall \ \mathbf{p} \in \mathcal{S}_{n}, \ \exists \ j:=j(\mathbf{p}) \in \mathcal{J} \ s.t. \ \sum_{i=1}^{n}p_{i}g_{j}(W_{i},\theta_{n,*})<0\big\}$, where $\mathcal{S}_{n}$ is the standard simplex. That is, $\mathcal{S}_{n}:= \{\mathbf{p} \in \mathbb{R}^{n}: \sum_{i=1}^{n}p_{i}=1, \ p_{i} \geq 0 \ \forall \ i \in \mathcal{I}\}$. Since $(n^{-1},...,n^{-1})^{\top} \in \mathcal{S}_{n}$, it follows that
\begin{align*}
\big\{\mathcal{C}_{n}(\theta_{n,*})=\emptyset\big\} \subseteq \bigcup_{j=1}^{J}\Bigg\{\frac{1}{n}\sum_{i=1}^{n}g_{j}(W_{i},\theta_{n,*})<0\Bigg\}
\end{align*}
and therefore
\begin{align*}
\lim_{n\rightarrow+\infty}P_{F_{n}}\Big(\mathcal{C}_{n}(\theta_{n,*})=\emptyset\Big) \leq \sum_{j=1}^{J}\lim_{n\rightarrow+\infty}P_{F_{n}}\Bigg(\frac{1}{n}\sum_{i=1}^{n}g_{j}(W_{i},\theta_{n,*})<0\Bigg).
\end{align*}
Hence it suffices to show that $\lim_{n\rightarrow+\infty}P_{F_{n}}(\hat{g}_{n,j}(\theta_{n,*})< 0 )=0$ for each $j \in \mathcal{J}$.
\paragraph{Step 2.} For each $j \in \mathcal{J}$, we can mean-value expand $E_{F_{n}}g_{j}(W_{i},\theta_{n,*})/\sigma_{F_{n},j}(\theta_{n,*})$ around $\{(\theta_{n},F_{n}): n \geq 1\} \in \mathcal{F}$ and conclude that
\begin{align*}
\frac{E_{F_{n}}g_{j}(W_{i},\theta_{n,*})}{\sigma_{F_{n},j}(\theta_{n,*})} = \frac{E_{F_{n}}g_{j}(W_{i},\theta_{n})}{\sigma_{F_{n},j}(\theta_{n})} + O(n^{-\frac{1}{2}})
\end{align*}
for each $n \geq 1$. So by the weak law of large numbers for triangular arrays of row-wise i.i.d. data, it follows that $\lim_{n\rightarrow+\infty}P_{F_{n}}(\hat{g}_{n,j}(\theta_{n,*})<0)=0$, and therefore $\lim_{n\rightarrow+\infty}P_{F_{n}}(\mathcal{C}_{n}(\theta_{n,*})=\emptyset)=0$.
\end{proof}
\end{lemma}
Lemma \ref{LP0} is an important intermediate technical result because it allows us to conclude that along any sequence $\{(\theta_{n,*},F_{n}): n \geq 1\}\in \mathcal{H}$, the empirical likelihood estimator exists with probability approaching $1$. In all of the subsequent results, it is implicit that the event $\{\mathcal{C}_{n}(\theta_{n,*}) \neq \emptyset\}$ occurs.
\begin{lemma}\label{LP1}
Define $\hat{g}_{n,b}(\theta_{n,*}):=n^{-1}\sum_{i=1}^{n}g_{b}(W_{i},\theta_{n,*})$. The following result holds for any sequence $\{(\theta_{n,*},F_{n}): n \geq 1\}$ of $n^{-\frac{1}{2}}$-local alternatives: $P_{F_{n}}\big((\acute{\lambda}_{n,b})^{\top}\hat{g}_{n,b}(\theta_{n,*})\geq 0\big)=1$ for each $n \geq 1$.
\begin{proof}
We outline the proof and then provide details.
\paragraph{Outline.} The proof employs the direct method. The first step shows that for any sequence $\{(\theta_{n,*},F_{n}): n \geq 1\}$, $\log(1+(\acute{\lambda}_{n,b})^{\top}\hat{g}_{n,b}(\theta_{n,*})) \geq 0$ with probability equal to $1$. The second step concludes the result using basic properties of the logarithmic function.
\paragraph{Step 1.}
Let $\acute{\lambda}_{n}$ denote a feasible solution to the dual problem (\ref{LPdual}) under an arbitrary sequence $\{(\theta_{n,*},F_{n}): n \geq 1\}$. Since the dual variables for the slack inequalities are equal to zero with probability 1 under the Karusch-Kuhn-Tucker conditions, we have that $\acute{\lambda}_{n}^{\top}g(W_{i},\theta_{n,*}) = (\acute{\lambda}_{n,b})^{\top}g_{b}(W_{i},\theta_{n,*})$ and the following holds with probability equal to 1:
\begin{align}
0 \leq \frac{1}{n}\sum_{i=1}^{n}\ln\Big(1+(\acute{\lambda}_{n,b})^{\top}g_{b}(W_{i},\theta_{n,*})\Big) &\leq \ln\bigg( 1+(\acute{\lambda}_{n,b})^{\top}\frac{1}{n}\sum_{i=1}^{n}g_{b}(W_{i},\theta_{n,*})\bigg)
\end{align}
where the first inequality holds as $2\sum_{i=1}^{n}\ln\Big(1+(\acute{\lambda}_{n,b})^{\top}g_{b}(W_{i},\theta_{n,*})\Big)$ is the empirical likelihood ratio statistic for testing the null hypothesis (see \citet{canay2010inference}) and the second holds by Jensen's inequality. This implies that $\log(1+(\acute{\lambda}_{n,b})^{\top}\hat{g}_{n,b}(\theta_{n,*})) \geq 0$ with probability equal to 1 for all $n \geq 1$.
\paragraph{Step 2.} For any $x \in \mathbb{R}$, $\ln(1+x) \geq 0$ if and only $x \geq 0$. Consequently, we use the conclusion of Step 1 to conclude that $(\acute{\lambda}_{n,b})^{\top}\hat{g}_{n,b}(W_{i},\theta_{n,*}) \geq 0$ with probability equal to 1 for each $n \geq 1$.
\end{proof}
\end{lemma}
\begin{lemma}\label{LP2}
Define random index set $\acute{\mathcal{B}}_{\varrho_{n}}:=\{j \in \mathcal{J}: \acute{g}_{n,j}(\theta_{n,*})=0\}$ and deterministic index set $C := \{j \in \mathcal{J}: \lim_{n\rightarrow\infty}E_{F_{n}}\big(g_{j}(W_{i},\theta_{n,*})\big)=0\}$. If Assumptions LA1 and LA2 hold, then the following result is true for any sequence of $n^{-\frac{1}{2}}$-local alternatives $\{(\theta_{n,*},F_{n}): n \geq 1\}$:
\begin{align}
\lim_{n\rightarrow\infty}P_{F_{n}}(\acute{\mathcal{B}}_{\varrho_{n}} \subseteq C) = 1
\end{align}
where $P_{F_{n}}(\cdot)$ is the probability measure induced by repeated sampling from $F_{n}$.
\begin{proof}
The proof has multiple steps so we present an outline and then the steps in detail.
\paragraph{Outline.} We want to show that the event $\big\{\acute{\mathcal{B}}_{\varrho_{n}} \subseteq C\big\}$ occurs with probability approaching $1$ along any sequence $\{(\theta_{n,*},F_{n}): n \geq 1\}$. This involves three steps. In Step 1, we use the complement rule to deduce that this is equivalent to showing that $\big\{\acute{\mathcal{B}}_{\varrho_{n}}\cap C^{c}\neq \emptyset \big\}$ occurs with probability approaching $0$ along $\{(\theta_{n,*},F_{n}): n \geq 1\}$. In Step 2, we characterize the event $\big\{\acute{\mathcal{B}}_{\varrho_{n}}\cap C^{c}\neq \emptyset \big\}$. In Step 3, we argue that that the event $\big\{\acute{\mathcal{B}}_{\varrho_{n}}\cap C^{c}\neq \emptyset \big\}$ occurs with probability approaching $0$ along $\{(\theta_{n,*},F_{n}): n \geq 1\}$.
\paragraph{Step 1.} Let $\{(\theta_{n,*}, F_{n}): n \geq 1\}$ be an arbitrary sequence of $n^{-\frac{1}{2}}$-local alternatives. By the complement rule,
\begin{align}\label{order1}
\lim_{n\rightarrow\infty}P_{F_{n}}(\acute{\mathcal{B}}_{\varrho_{n}} \subseteq C) = 1 -\lim_{n\rightarrow\infty}P_{F_{n}}(\acute{\mathcal{B}}_{\varrho_{n}}\cap C^{c} \neq \emptyset).
\end{align}
So to show $\lim_{n\rightarrow\infty}P_{F_{n}}(\acute{\mathcal{B}}_{\varrho_{n}} \subseteq C)=1$, it suffices to show that $\lim_{n\rightarrow\infty}P_{F_{n}}(\acute{\mathcal{B}}_{\varrho_{n}}\cap C^{c} \neq \emptyset)=0$.
\paragraph{Step 2.} On the event $\{ \acute{\mathcal{B}}_{\varrho_{n}}\cap C^{c} \neq \emptyset\}$, there exists $j \in \mathcal{J}$ such that $\acute{g}_{n,j}(\theta_{n,*}) =0$ and $\lim_{n\rightarrow\infty}E_{F_{n}}(g_{j}(W_{i},\theta_{n,*})) > 0$. The deduction that $\lim_{n\rightarrow+\infty}E_{F_{n}}g_{j}(W_{i},\theta_{n,*})>0$ follows from a mean-value expansion of $E_{F_{n}}g_{j}(W_{i},\theta_{n,*})/\sigma_{F_{n},j}(\theta_{n,*})$ around $\theta_{n}$, which yields
\begin{align*}
\frac{E_{F_{n}}g_{j}(W_{i},\theta_{n,*})}{\sigma_{F_{n},j}(\theta_{n,*})}=\frac{E_{F_{n}}g_{j}(W_{i},\theta_{n})}{\sigma_{F_{n},j}(\theta_{n})}+O(n^{-\frac{1}{2}})
\end{align*}
and therefore $E_{F_{n}}g_{j}(W_{i},\theta_{n,*})$ is asymptotically nonnegative because $\{(\theta_{n},F_{n}): n \geq 1\} \in \mathcal{F}$. The expansion is valid under LA2. So all we need to show that with probability approaching $1$ it is not possible for $\acute{g}_{n,j}(\theta_{n,*})=0$ and $\lim_{n\rightarrow+\infty}E_{F_{n}}g_{j}(W_{i},\theta_{n,*})>0$ to be satisfied jointly for any $j \in \mathcal{J}$.
\paragraph{Step 3.}
Suppose that there exists $j \in \mathcal{J}$ such that $\acute{g}_{n,j}(\theta_{n,*}) = 0$ and $\lim_{n\rightarrow+\infty}E_{F_{n}}g_{j}(W_{i},\theta_{n,*})>0$. We first note that $\acute{g}_{n,j}(\theta_{n,*}) = 0$ implies $\acute{g}_{n,j}(\theta_{n,*}) \geq \hat{g}_{n,j}(\theta_{n,*})$ because
\begin{align*}
0=\acute{g}_{n,j}(\theta_{n,*}) = \frac{1}{n}\sum_{i=1}^{n}\Bigg(\frac{g_{j}(W_{i},\theta_{n,*})}{1+(\acute{\lambda}_{n,b})^{\top}g_{b}(W_{i},\theta_{n,*})}\Bigg) \geq \frac{\hat{g}_{n,j}(\theta_{n,*})}{1+(\acute{\lambda}_{n,b})^{\top}\hat{g}_{b}(\theta_{n,*})}  \geq \hat{g}_{n,j}(\theta_{n,*}).
\end{align*}
by Jensen's inequality and the fact that $\acute{\lambda}_{n,b})^{\top}\hat{g}_{b}(\theta_{n,*}) \geq 0$ (see Lemma \ref{LP1}). Next, by adding and subtracting the expectation of $\hat{g}_{n,j}(\theta_{n,*})$, we have that
\begin{align}\label{order2}
\acute{g}_{n,j}(\theta_{n,*})\geq\hat{g}_{n,j}(\theta_{n,*})-E_{F_{n}}(\hat{g}_{n,j}(\theta_{n,*}))+E_{F_{n}}(\hat{g}_{n,j}(\theta_{n,*}))=o_{p}(1)+E_{F_{n}}(g_{j}(W_{i},\theta_{n,*}))
\end{align}
along $\{(\theta_{n,*},F_{n}: n \geq 1\}$ by the weak law of large numbers for triangular arrays of row-wise i.i.d. random variables and the unbiasedness of $\hat{g}_{n,j}(\theta_{n,*})$ for $E_{F_{n}}(g_{j}(W_{i},\theta_{n,*}))$. If we send $n\rightarrow+\infty$, we deduce that the probability limit of $\acute{g}_{n,j}(\theta_{n,*})$ is strictly positive by the ordering in (\ref{order2}) and the fact that $\lim_{n\rightarrow+\infty}E_{F_{n}}g_{j}(W_{i},\theta_{n,*})>0$. Consequently, $\lim_{n\rightarrow+\infty}P_{F_{n}}(\acute{\mathcal{B}}_{\varrho_{n}}\cap C^{c} \neq \emptyset)=0$ along $\{(\theta_{n,*},F_{n}: n \geq 1\}$. Combining this result with Step 1, we complete the proof.
\end{proof}
\end{lemma}
\begin{lemma}\label{LP3}
Let $B_{n}:=|\acute{\mathcal{B}}_{\varrho_{n}}|$, $\acute{\Lambda}:=\{\acute{\lambda}_{n,b} | (\ref{KKT1}), (\ref{KKT2}), (\ref{KKT3}) \ hold\}\subseteq \mathbb{R}^{B_{n}}$, and $||\cdot ||_{\ell_{B_{n}}^{2}}$ denote the Euclidean norm on $\mathbb{R}^{B_{n}}$. If Assumptions LA1 and LA2 hold, then $\sup_{\acute{\lambda}_{n,b} \in \acute{\Lambda}_{n}}||\acute{\lambda}_{n,b}||_{\ell^{2}_{B_{n}}} =O_{p}(n^{-\frac{1}{2}})$ along any sequence $\{(\theta_{n,*},F_{n}): n \geq 1\}$.
\begin{proof}
The proof proceeds by the direct method. Given the length of the proof, we outline the argument and then provide detailed steps.
\paragraph{Outline.} Our goal is to show that for any $\{(\theta_{n,*},F_{n}): n \geq 1\}$, $\sup_{\acute{\lambda}_{n,b} \in \acute{\Lambda}_{n}}||\acute{\lambda}_{n,b}||_{\ell^{2}_{B_{n}}}=O_{p}(n^{-\frac{1}{2}})$. This involves three steps. In the first step, we do some algebra to relate the Karusch-Kuhn-Tucker conditions to $||\acute{\lambda}_{n,b}||_{\ell^{2}_{B_{n}}}$. In the second step, we derive a bound relating $||\acute{\lambda}_{n,b}||_{\ell^{2}_{B_{n}}}$ and the sample moments of the inequalities that are binding under the Karusch-Kuhn-Tucker conditions. In the third step, we use standard limit theorems for triangular arrays of row-wise i.i.d. random variables and the bound derived in Step 2 to conclude the result.
\paragraph{Step 1.}
Let $\{(\theta_{n,*}, F_{n}): n \geq 1\}$ be arbitrary. The Karusch-Kuhn-Tucker conditions dictate that $\acute{\lambda}_{n}$ satisfies
\begin{align}\label{FOCLP1}
\frac{1}{n}\sum_{i=1}^{n}\frac{g_{b}(W_{i},\theta_{n,*})}{1+(\acute{\lambda}_{n})^{\top}g(W_{i},\theta_{n,*})}=0_{B_{n}}.
\end{align}
Under complementary slackness $\acute{\lambda}_{n}=(\acute{\lambda}_{n,b},0_{J-B_{n}})$, which implies that (\ref{FOCLP1}) is equivalent to
\begin{align}\label{FOC2}
\frac{1}{n}\sum_{i=1}^{n}\frac{g_{b}(W_{i},\theta_{n,*})}{1+(\acute{\lambda}_{n,b})^{\top}g_{b}(W_{i},\theta_{n,*})}=0_{B_{n}}.
\end{align}
Let $\{\beta_{n}: n \geq 1\}$ be a sequence of unit vectors in $\mathbb{R}^{B_{n}}$ that satisfy $\beta_{n}||\acute{\lambda}_{n,b}||_{\ell^{2}_{B_{n}}} = \acute{\lambda}_{n,b}$. We take the inner product between $\beta_{n}$ and (\ref{FOC2}), which gives us
\begin{align}
\beta_{n}^{\top}\Bigg(\frac{1}{n}\sum_{i=1}^{n}\frac{g_{b}(W_{i},\theta_{n,*})}{1+(\acute{\lambda}_{n,b})^{\top}g_{b}(W_{i},\theta_{n,*})}\Bigg)=0
\end{align}
If we define $X_{i}:=(\acute{\lambda}_{n,b})^{\top}g_{b}(W_{i},\theta_{n,*})$ for all $i=1,...,n$ and use the transformation $$\frac{1}{1+X_{i}}= 1 - \frac{X_{i}}{1+X_{i}}$$ for all $i=1,...,n$, then we have that
\begin{align} \label{id1}
\beta_{n}^{\top}\Bigg(\frac{1}{n}\sum_{i=1}^{n}g_{b}(W_{i},\theta_{n,*})\Bigg) &= \beta_{n}^{\top}\Bigg(\frac{1}{n}\sum_{i=1}^{n}\frac{\big(g_{b}(W_{i},\theta_{n,*})\big)(\acute{\lambda}_{n,b})^{\top}g_{b}(W_{i},\theta_{n,*})}{1+(\acute{\lambda}_{n,b})^{\top}g_{b}(W_{i},\theta_{n,*})}\Bigg) \\
&=||\acute{\lambda}_{n,b}||_{\ell^{2}_{B_{n}}}\beta_{n}^{\top}\Bigg(\frac{1}{n}\sum_{i=1}^{n}\frac{g_{b}(W_{i},\theta_{n,*})g_{b}(W_{i},\theta_{n,*})^{\top}}{1+(\acute{\lambda}_{n,b})^{\top}g_{b}(W_{i},\theta_{n,*})}\Bigg)\beta_{n}. \label{id2}
\end{align}
where the last equality holds by the definition of $\{\beta_{n}: n \geq 1\}$.
\paragraph{Step 2.} Let $\hat{\Sigma}_{n,b}(\theta_{n,*})$ denote the sample analogue estimator of the covariance matrix of $g_{b}(W_{i},\theta_{n,*})$. We will relate $\hat{\Sigma}_{n,b}(\theta_{n,*})$ to the RHS of (\ref{id2}). Since $\acute{\mathcal{B}}_{\varrho_{n}}\subseteq C$ with probability approaching 1 (Lemma \ref{LP2}), we have that
\begin{align}
\hat{\Sigma}_{n,b}(\theta_{n,*})=\frac{1}{n}\sum_{i=1}^{n}g_{b}(W_{i},\theta_{n,*})g_{b}(W_{i},\theta_{n,*})^{\top}.
\end{align}
with probability approaching $1$ along $\{(\theta_{n,*},F_{n}): n \geq 1\}$. Since $p_{i}>0$, we have that $1+X_{i}>0$ for all $i=1,...,n$, which implies that with probability tending to $1$ along $\{(\theta_{n,*},F_{n}): n \geq 1\}$:
\begin{align}
||\acute{\lambda}_{n,b}||_{\ell^{2}_{B_{n}}}\beta_{n}^{\top}\hat{\Sigma}_{n,b}(\theta_{n,*})\beta_{n} &\leq ||\acute{\lambda}_{n,b}||_{\ell^{2}_{B_{n}}}\beta^{\top}_{n}\Bigg(\frac{1}{n}\sum_{i=1}^{n}\frac{g_{b}(W_{i},\theta_{n,*})g_{b}(W_{i},\theta_{n,*})^{\top}}{1+(\acute{\lambda}_{n,b})^{\top}g_{b}(W_{i},\theta_{n,*})}\Bigg)\beta_{n}(1+X_{max})
\end{align}
where $X_{max} = \max_{1 \leq i \leq n}|(\acute{\lambda}_{n,b})^{\top}g_{b}(W_{i},\theta_{n,*})|$. Applying the Cauchy-Schwarz inequality, we have that
 \begin{align}
 |(\acute{\lambda}_{n,b})^{\top}g_{b}(W_{i},\theta_{n,*})| \leq ||\acute{\lambda}_{n,b}||_{\ell^{2}_{B_{n}}}||g_{b}(W_{i},\theta_{n,*})||_{\ell^{2}_{B_{n}}},
 \end{align}
implying that
\begin{align} \label{id3}
||\acute{\lambda}_{n,b}||_{\ell^{2}_{B_{n}}}\beta_{n}^{\top}\hat{\Sigma}_{n,b}(\theta_{n,*})\beta_{n} \leq ||\acute{\lambda}_{n,b}||_{\ell^{2}_{B_{n}}}\beta^{\top}_{n}\Bigg(\frac{1}{n}\sum_{i=1}^{n}\frac{g_{b}(W_{i},\theta_{n,*})g_{b}(W_{i},\theta_{n,*})^{\top}}{1+(\acute{\lambda}_{n,b})^{\top}g_{b}(W_{i},\theta_{n,*})}\Bigg)\beta_{n}(1+||\acute{\lambda}_{n,b}||_{\ell^{2}_{B_{n}}}Z_{n}^{*})
\end{align}
where $Z_{n}^{*}:= \max_{1 \leq i \leq n}||g_{b}(W_{i},\theta_{n,*})||_{\ell^{2}_{B_{n}}}$. Apply the equality in (\ref{id2}) to the right hand side of (\ref{id3}) to conclude
\begin{align}
||\acute{\lambda}_{n,b}||_{\ell^{2}_{B_{n}}}\Bigg(\beta_{n}^{\top}\hat{\Sigma}_{n,b}(\theta_{n,*})\beta_{n}-\beta_{n}^{\top}\bigg(\frac{Z_{n}^{*}}{n}\sum_{i=1}^{n}g_{b}(W_{i},\theta_{n,*})\bigg)\Bigg) \leq \beta_{n}^{\top}\Bigg(\frac{1}{n}\sum_{i=1}^{n}g_{b}(W_{i},\theta_{n,*})\Bigg)
\end{align}
\paragraph{Step 3.} By Lemma \ref{owen112}, we have that $Z_{n}^{*}= o_{p}(n^{\frac{1}{2}})$ along $\{(\theta_{n,*}, F_{n}): \ n \geq 1\}$. Since $\acute{\mathcal{B}}_{\varrho_{n}} \subseteq C$ for large $n$, we can apply the Lyapunov CLT to $\frac{1}{n}\sum_{i=1}^{n}g_{b}(W_{i},\theta_{n,*})$ to conclude that $n^{-1}\sum_{i=1}^{n}g_{b}(W_{i},\theta_{n,*}) = O_{p}(n^{-\frac{1}{2}})$ along $\{(\theta_{n,*}, F_{n}): \ n \geq 1\}$. Finally, LA1 implies that
\begin{align}
0<a+o_{p}(1)\leq \beta_{n}^{\top}\hat{\Sigma}_{n,b}(\theta_{n,*})\beta_{n}\leq b+o_{p}(1)
\end{align}
along the sequence $\{(\theta_{n,*}, F_{n}): \ n \geq 1\}$, where $a$ and $b$ are the smallest and largest eigenvalues of the variance matrix of binding moments in the population. These limiting results allow us to conclude that
\begin{align}
||\acute{\lambda}_{n,b}||_{\ell^{2}_{B_{n}}} \leq \frac{O_{p}(n^{-\frac{1}{2}})}{a+o_{p}(1)} \quad \forall \ \acute{\lambda}_{n,b}\in \acute{\Lambda}_{n}
\end{align}
along $\{(\theta_{n,*}, F_{n}): \ n \geq 1\}$. We have shown that the positive random variable $||\acute{\lambda}_{n,b}||_{\ell^{2}_{B_{n}}}$ is bounded above by a random variable that is $O_{p}(n^{-\frac{1}{2}})$ which implies
\begin{align}
\sup_{\acute{\lambda}_{n,b}\in \Lambda_{n}^{*}}||\acute{\lambda}_{n,b}||_{\ell^{2}_{B_{n}}} = O_{p}(n^{-\frac{1}{2}})
\end{align}
along $\{(\theta_{n,*}, F_{n}): n \geq 1\}$.
\end{proof}
\end{lemma}
\begin{lemma}\label{LP4}
Let $\acute{g}_{n}(\theta_{n,*})$ and $\hat{g}_{n}(\theta_{n,*})$ denote the restricted and unrestricted estimators of the moments, respectively, under $n^{-\frac{1}{2}}$-local alternatives. If the sequence $\{(\theta_{n,*},F_{n}): n \geq 1\}$ satisfies Assumptions LA1 and LA2, then $||\hat{g}_{n}(\theta_{n,*})-\acute{g}_{n}(\theta_{n,*})||_{\ell^{2}_{J}}=O_{p}(n^{-\frac{1}{2}})$.
\begin{proof}
Due to the length of the proof, we present an outline and then the steps in detail.
\paragraph{Outline.} Our goal is to show $||\hat{g}_{n}(\theta_{n,*})-\acute{g}_{n}(\theta_{n,*})||_{\ell^{2}_{J}}=O_{p}(n^{-\frac{1}{2}})$ for any sequence $\{(\theta_{n,*},F_{n}): n \geq 1\}$. This involves three steps. In Step 1, we show that proving $||\hat{g}_{n}(\theta_{n,*})-\acute{g}_{n}(\theta_{n,*})||_{\ell^{2}_{J\times J}}=O_{p}(n^{-\frac{1}{2}})$ along a sequence of $n^{-\frac{1}{2}}$-local alternatives amounts to establishing the result coordinate-wise. In Step 2, we use Lemma \ref{LP3} to deduce that showing $|\hat{g}_{n,j}(\theta_{n,*})-\acute{g}_{n,j}(\theta_{n,*})|=O_{p}(n^{-\frac{1}{2}})$ only requires showing $\sum_{i=1}^{n}\acute{p}_{i}||g_{b}(W_{i},\theta_{n,*})g_{j}(W_{i},\theta_{n,*})||_{\ell^{2}_{B_{n}}}=O_{p}(1)$ along $\{(\theta_{n,*},F_{n}): n \geq 1\}$. In Step 3, we show the required result and complete the proof.
\paragraph{Step 1.} Let $\{(\theta_{n,*}, F_{n}): n \geq 1\}$ be arbitrary and let $\{e_{1},...,e_{J}\}$ denote the standard basis for $\mathbb{R}^{J}$. The triangle inequality and the unit length of the basis allows us to conclude that $||\acute{g}_{n}(\theta_{n,*})-\hat{g}_{n}(\theta_{n,*})||_{\ell^{2}_{J}}\leq \sum_{j=1}^{J} | |e_{j}\big(\hat{g}_{n,j}(\theta_{n,*})- \acute{g}_{n,j}(\theta_{n,*})\big) | |_{\ell^{2}_{J}}=\sum_{j=1}^{J} |\hat{g}_{n,j}(\theta_{n,*})- \acute{g}_{n,j}(\theta_{n,*})|$. Consequently, a sufficient condition for $||\acute{g}_{n}(\theta_{n,*})-\hat{g}_{n}(\theta_{n,*})||_{\ell^{2}_{J}}=O_{p}(n^{-\frac{1}{2}})$ along $\{(\theta_{n,*}, F_{n}): n \geq 1\}$ is that, for each $j \in \mathcal{J}$, $|\hat{g}_{n,j}(\theta_{n,*})- \acute{g}_{n,j}(\theta_{n,*})|=O_{p}(n^{-\frac{1}{2}})$ along $\{(\theta_{n,*}, F_{n}): n \geq 1\}$.
\paragraph{Step 2.} Consider $|\hat{g}_{n,j}(\theta_{n,*})-\acute{g}_{n,j}(\theta_{n,*})|$, where $j \in \mathcal{J}$ is fixed arbitrarily. We conclude that
\begin{align*}
|\hat{g}_{n,j}(\theta_{n,*})-\acute{g}_{n,j}(\theta_{n,*})|&=\Bigg | \sum_{i=1}^{n}\Big(\frac{1}{n}-\acute{p}_{i}\Big)g_{j}(W_{i},\theta_{n,*}) \Bigg | \\
& = \Bigg | \frac{1}{n}\sum_{i=1}^{n}\Bigg ( 1- \frac{1}{1+(\acute{\lambda}_{n,b})^{\top}g_{b}(W_{i},\theta_{n,*})}\Bigg)g_{j}(W_{i},\theta_{n,*}) \Bigg | \\
&= \Bigg | \frac{1}{n}\sum_{i=1}^{n}\Bigg(\frac{(\acute{\lambda}_{n,b})^{\top}g_{b}(W_{i},\theta_{n,*})}{1+(\acute{\lambda}_{n,b})^{\top}g_{b}(W_{i},\theta_{n,*})}\Bigg)g_{j}(W_{i},\theta_{n,*})\Bigg | \\
& = \Bigg | \acute{\lambda}_{n,b}^{\top}\sum_{i=1}^{n}\acute{p}_{i}g_{b}(W_{i},\theta_{n,*})g_{j}(W_{i},\theta_{n,*}) \Bigg | \\
& \leq ||\acute{\lambda}_{n,b}||_{\ell^{2}_{B_{n}}}\Bigg | \Bigg | \sum_{i=1}^{n}\acute{p}_{i}g_{b}(W_{i},\theta_{n,*})g_{j}(W_{i},\theta_{n,*})\Bigg | \Bigg |_{\ell^{2}_{B_{n}}}  \\
& \leq ||\acute{\lambda}_{n,b}||_{\ell^{2}_{B_{n}}}\sum_{i=1}^{n} ||\acute{p}_{i}g_{b}(W_{i},\theta_{n,*})g_{j}(W_{i},\theta_{n,*}) ||_{\ell^{2}_{B_{n}}}  \\
& \leq \sup_{\acute{\lambda}_{n,b} \in \acute{\Lambda}_{n}} ||\acute{\lambda}_{n,b}||_{\ell^{2}_{B_{n}}}\sum_{i=1}^{n} \acute{p}_{i}||g_{b}(W_{i},\theta_{n,*})g_{j}(W_{i},\theta_{n,*}) ||_{\ell^{2}_{B_{n}}}
\end{align*}
where the first inequality holds by Cauchy-Schwarz, the second holds by the triangle inequality, and the final holds by the definition of least upper bound. By Lemma \ref{LP3}, it suffices to show that
\begin{align*}
\sum_{i=1}^{n} \acute{p}_{i}||g_{b}(W_{i},\theta_{n,*})g_{j}(W_{i},\theta_{n,*}) ||_{\ell^{2}_{B_{n}}}=O_{p}(1)
\end{align*}
and this is what we do in Step 3.
\paragraph{Step 3.} By definition of $\acute{p}_{i}$ and the fact that $\frac{1}{1+x}=1-\frac{x}{1+x}$, one can show
\begin{align}\label{sufficient}
\sum_{i=1}^{n} \acute{p}_{i}||g_{b}(W_{i},\theta_{n,*})g_{j}(W_{i},\theta_{n,*}) ||_{\ell^{2}_{B_{n}}}=E_{n,1}(\theta_{n,*})-E_{n,2}(\theta_{n,*})
\end{align}
where
\begin{align}\label{LP3id1}
&E_{n,1}(\theta_{n,*}):=\frac{1}{n}\sum_{i=1}^{n}||g_{b}(W_{i},\theta_{n,*})g_{j}(W_{i},\theta_{n,*}) ||_{\ell^{2}_{B_{n}}}
\end{align}
and
\begin{align}\label{identity}
&E_{n,2}(\theta_{n,*}):= \frac{1}{n}\sum_{i=1}^{n}\Bigg(\frac{(\acute{\lambda}_{n,b})^{\top}g_{b}(W_{i},\theta_{n,*})}{1+(\acute{\lambda}_{n,b})^{\top}g_{b}(W_{i},\theta_{n,*})}\Bigg)||g_{b}(W_{i},\theta_{n,*})g_{j}(W_{i},\theta_{n,*}) ||_{\ell^{2}_{B_{n}}}.
\end{align}
First, note that $E_{n,1}(\theta_{n,*})=O_{p}(1)$ under $\{(\theta_{n,*}, F_{n}): n \geq 1\}$ by the weak law of large numbers for triangular arrays of row-wise i.i.d. random variables and LA1. Regarding (\ref{identity}), it is easy to see that $(\acute{\lambda}_{n,b})^{\top}g_{b}(W_{i},\theta_{n,*})=o_{p}(1)$ along $\{(\theta_{n,*}, F_{n}): n \geq 1\}$ because
\begin{align*}
|(\acute{\lambda}_{n,b})^{\top}g_{b}(W_{i},\theta_{n,*})| &\leq \sup_{\acute{\lambda}_{n,b}\in \acute{\Lambda}_{n}}||\acute{\lambda}_{n,b}||_{\ell^{2}_{B_{n}}}\max_{1 \leq i \leq n}||g_{b}(W_{i},\theta_{n,*})||_{\ell^{2}_{B_{n}}}\\
&=O_{p}(n^{-\frac{1}{2}})o_{p}(n^{\frac{1}{2}})\\
&=o_{p}(1)
\end{align*}
by the Cauchy-Schwarz inequality, Lemma \ref{owen112}, and Lemma \ref{LP3}. This implies that
\begin{align}\label{equiv1}
E_{n,2}(\theta_{n,*}) = \underbrace{(\acute{\lambda}_{n,b})^{\top}\Bigg(\frac{1}{n}\sum_{i=1}^{n}g_{b}(W_{i},\theta_{n,*})||g_{b}(W_{i},\theta_{n,*})g_{j}(W_{i},\theta_{n,*}) ||_{\ell^{2}_{B_{n}}}\Bigg)}_{E_{n,3}(\theta_{n,*})}+o_{p}(1)
\end{align}
along $\{(\theta_{n,*}, F_{n}): n \geq 1\}$. The Cauchy-Schwarz inequality, triangle inequality, and definition of least upper bound implies that
\begin{align}\label{ub1}
|E_{n,3}(\theta_{n,*})|&\leq \sup_{\acute{\lambda}_{n,b} \in \acute{\Lambda}_{n}}||\acute{\lambda}||_{\ell^{2}_{B_{n}}}\max_{1\leq i \leq n}||g_{b}(W_{i},\theta_{n,*})||_{\ell^{2}_{B_{n}}}\frac{1}{n}\sum_{i=1}^{n}||g_{b}(W_{i},\theta_{n,*})g_{j}(W_{i},\theta_{n,*}) ||_{\ell^{2}_{B_{n}}} \\
&= O_{p}(n^{-\frac{1}{2}})o_{p}(n^{\frac{1}{2}})O_{p}(1)\\
&=o_{p}(1)
\end{align}
along $\{(\theta_{n,*}, F_{n}): n \geq 1\}$, where the equality holds by Lemma \ref{LP3}, Lemma \ref{owen112}, the weak law of large numbers for triangular arrays of row-wise i.i.d. random variables, and LA1. This result implies that $E_{n,2}(\theta_{n,*})$ is $o_{p}(1)$ and, combined with $E_{n,1}(\theta_{n,*})=O_{p}(1)$, implies that (\ref{sufficient}) is $O_{p}(1)$, which was required to show $|\hat{g}_{n,j}(\theta_{n,*})-\acute{g}_{n,j}(\theta_{n,*})|=O_{p}(n^{-\frac{1}{2}})$.
\end{proof}
\end{lemma}
\subsection{Technical Results for Power Comparison}\label{powervariance}
The first result establishes the consistency of the constrained estimator of the covariance matrix along $n^{-\frac{1}{2}}$-local alternatives. Like Lemma \ref{unifcov}, $||\cdot||_{\ell^{2}_{J\times J}}$ denotes the Frobenius norm.
\begin{lemma}\label{LPC1}
Let $\acute{\Sigma}_{n}(\theta_{n,*}):=\sum_{i=1}^{n}\acute{p}_{i}\big(g(W_{i},\theta_{n,*})-\acute{g}_{n}(\theta_{n,*})\big)\big(g(W_{i},\theta_{n,*})^{\top}-\acute{g}_{n}(\theta_{n,*})\big)^{\top}$ and $\Sigma(\theta_{n,*},F_{n}):=Cov_{F_{n}}\big(g(W_{i},\theta_{n,*})\big)$. If Assumptions LA1 and LA2 hold, then $\forall \ r>0$, $\forall \ \{(\theta_{n,*}, F_{n}): n \geq 1\}$,
\begin{align*}
\lim_{n\rightarrow\infty}P_{F_{n}}\Big(||\acute{\Sigma}_{n}(\theta_{n,*})-\Sigma(\theta_{n,*}, F_{n}) ||_{\ell^{2}_{J\times J}}<r\Big)=1.
\end{align*}
\begin{proof}
The proof proceeds by the direct method. Although the argument is linear, it has a few steps so we outline the proof and then provide the details.
\paragraph{Outline.} We consider an arbitrary sequence $\{(\theta_{n,*},F_{n}): n \geq 1\}$ of $n^{-\frac{1}{2}}$-local alternatives and show that $||\acute{\Sigma}_{n}(\theta_{n,*})-\Sigma(\theta_{n,*}, F_{n}) ||_{\ell^{2}_{J\times J}}=o_{p}(1)$. To do this, there are a FEW steps. In the first step, we deduce that it is sufficient to show $||\sum_{i=1}^{n}(n^{-1}-\acute{p}_{i})g(W_{i},\theta_{n,*})g(W_{i},\theta_{n,*})^{\top}||_{\ell^{2}_{J\times J}}=o_{p}(1)$. In Step 2, we show $||\sum_{i=1}^{n}(n^{-1}-\acute{p}_{i})g(W_{i},\theta_{n,*})g(W_{i},\theta_{n,*})^{\top}||_{\ell^{2}_{J\times J}}\leq o_{p}(1)\sum_{i=1}^{n}\acute{p}_{i}||g(W_{i},\theta_{n,*})g(W_{i},\theta_{n,*})^{\top}||_{\ell^{2}_{J\times J}}$. Subsequently, we show that $\acute{p}_{i}||g(W_{i},\theta_{n,*})g(W_{i},\theta_{n,*})^{\top}||_{\ell^{2}_{J\times J}}=O_{p}(1)$ in Step 3.
\paragraph{Step 1.}
Fix $\{(\theta_{n,*}, F_{n}): n \geq 1\}$ arbitrarily. By the triangle inequality,
\begin{align}
||\acute{\Sigma}_{n}(\theta_{n,*})-\Sigma(\theta_{n,*}, F_{n}) ||_{\ell^{2}_{J\times J}}&\leq ||\acute{\Sigma}_{n}(\theta_{n,*})-\hat{\Sigma}(\theta_{n,*}) ||_{\ell^{2}_{J\times J}}\\
&\quad +||\hat{\Sigma}_{n}(\theta_{n,*})-\Sigma(\theta_{n,*}, F_{n}) ||_{\ell^{2}_{J\times J}} \\
&= ||\acute{\Sigma}_{n}(\theta_{n,*})-\hat{\Sigma}(\theta_{n,*}) ||_{\ell^{2}_{J\times J}}+o_{p}(1)
\end{align}
along $\{(\theta_{n,*}, F_{n}): n \geq 1\}$, where the second equality holds by the weak law of large numbers for triangular arrays of row-wise i.i.d. random variables. Decomposing $||\acute{\Sigma}_{n}(\theta_{n,*})-\hat{\Sigma}(\theta_{n,*}) ||_{\ell^{2}_{J\times J}}$, we obtain
\begin{align}
||\acute{\Sigma}_{n}(\theta_{n,*})-\hat{\Sigma}(\theta_{n,*}) ||_{\ell^{2}_{J\times J}}
& \leq \Bigg | \Bigg | \sum_{i=1}^{n}\Big ( \frac{1}{n}-\acute{p}_{i}\Big) g(W_{i},\theta_{n,*})g(W_{i},\theta_{n,*})^{\top} \Bigg | \Bigg |_{\ell^{2}_{J\times J}}+o_{p}(1) \label{tineq1}
\end{align}
\normalsize
along $\{(\theta_{n,*}, F_{n}): n \geq 1\}$, where the inequality holds by the triangle inequality, Lemma \ref{LP4}, and the continuous mapping theorem. Consequently, the result boils down to being able to show that the first term in (\ref{tineq1}) is $o_{p}(1)$ along $\{(\theta_{n,*}, F_{n}): n \geq 1\}$.
\paragraph{Step 2.} Following a similar derivation to that in Lemma \ref{LP4}, it can be shown that
\begin{align}
 \Bigg | \Bigg | \sum_{i=1}^{n}\Big ( \frac{1}{n}-\acute{p}_{i}\Big) g(W_{i},\theta_{n,*})g(W_{i},\theta_{n,*})^{\top} \Bigg | \Bigg |_{\ell^{2}_{J\times J}} &= \Bigg | \Bigg | \sum_{i=1}^{n}\bigg[\acute{p}_{i}(\acute{\lambda}_{n,b})^{\top}g_{b}(W_{i},\theta_{n,*})\\
 &\quad \quad \quad \quad \times g(W_{i},\theta_{n,*})g(W_{i},\theta_{n,*})^{\top}\bigg]\Bigg | \Bigg | _{\ell^{2}_{J\times J}} \nonumber \\
 & \leq \sup_{\acute{\lambda}_{n,b} \in \acute{\Lambda}_{n}}||\acute{\lambda}_{n,b}||_{\ell^{2}_{B_{n}}}\max_{1\leq i \leq n}||g_{b}(W_{i},\theta_{n,*})||_{\ell^{2}_{B_{n}}}\\
 &\quad \quad \quad \times\sum_{i=1}^{n}\acute{p}_{i}||g(W_{i},\theta_{n,*})g(W_{i},\theta_{n,*})^{\top}||_{\ell^{2}_{J\times J}} \nonumber
\end{align}
\normalsize
where the inequality holds by the triangle inequality, the Cauchy-Schwarz inequality and definition of the least upper bound. Since $$\sup_{\acute{\lambda}_{n,b} \in \acute{\Lambda}_{n}}||\acute{\lambda}_{n,b}||_{\ell^{2}_{B_{n}}}\max_{1\leq i \leq n}||g_{b}(W_{i},\theta_{n,*})||_{\ell^{2}_{J\times J}}=O_{p}(n^{-\frac{1}{2}})o_{p}(n^{\frac{1}{2}})=o_{p}(1)$$ along $\{(\theta_{n,*}, F_{n}): n \geq 1\}$ by Lemma \ref{LP3} and Lemma \ref{owen112}, it suffices to show that $$\sum_{i=1}^{n}\acute{p}_{i}||g(W_{i},\theta_{n,*})g(W_{i},\theta_{n,*})^{\top}||_{\ell^{2}_{J\times J}}=O_{p}(1)$$ along $\{(\theta_{n,*}, F_{n}): n \geq 1\}$. This is the task of Step 3.
\paragraph{Step 3.} Decompose $\sum_{i=1}^{n}\acute{p}_{i}||g(W_{i},\theta_{n,*})g(W_{i},\theta_{n,*})^{\top}||_{\ell^{2}_{J\times J}}$ as follows
\begin{align*}
\sum_{i=1}^{n}\acute{p}_{i}||g(W_{i},\theta_{n,*})g(W_{i},\theta_{n,*})^{\top}||_{\ell^{2}_{J\times J}}&= \frac{1}{n}\sum_{i=1}^{n}\frac{||g(W_{i},\theta_{n,*})g(W_{i},\theta_{n,*})^{\top}||_{\ell^{2}_{J\times J}}}{1+(\acute{\lambda}_{n,b})^{\top}g_{b}(W_{i},\theta_{n,*})} \\
&=\frac{1}{n}\sum_{i=1}^{n}||g(W_{i},\theta_{n,*})g(W_{i},\theta_{n,*})^{\top}||_{\ell^{2}_{J\times J}} \\
&- \frac{1}{n}\sum_{i=1}^{n}\frac{(\acute{\lambda}_{n,b})^{\top}g_{b}(W_{i},\theta_{n,*})||g(W_{i},\theta_{n,*})g(W_{i},\theta_{n,*})^{\top}||_{\ell^{2}_{J\times J}}}{1+(\acute{\lambda}_{n,b})^{\top}g_{b}(W_{i},\theta_{n,*})} \\
& \leq \frac{1}{n}\sum_{i=1}^{n}||g(W_{i},\theta_{n,*})g(W_{i},\theta_{n,*})^{\top}||_{\ell^{2}_{J\times J}} \\
&-\frac{(\acute{\lambda}_{n,b})^{\top}\frac{1}{n}\sum_{i=1}^{n}g_{b}(W_{i},\theta_{n,*})||g(W_{i},\theta_{n,*})g(W_{i},\theta_{n,*})^{\top}||_{\ell^{2}_{J\times J}}}{1+(\acute{\lambda}_{n,b})^{\top}\hat{g}_{n,b}(\theta_{n,*})}.
\end{align*}
\normalsize
where the inequality holds by Jensen's inequality. By the weak law of large numbers for row-wise i.i.d. random variables and LA1, $\frac{1}{n}\sum_{i=1}^{n}||g(W_{i},\theta_{n,*})g(W_{i},\theta_{n,*})^{\top}||_{\ell^{2}_{J\times J}}=O_{p}(1)$ along $\{(\theta_{n,*}, F_{n}): n \geq 1\}$. Now, let $$E_{n,4}(\theta_{n,*}):=\frac{(\acute{\lambda}_{n,b})^{\top}\frac{1}{n}\sum_{i=1}^{n}g_{b}(W_{i},\theta_{n,*})||g(W_{i},\theta_{n,*})g(W_{i},\theta_{n,*})^{\top}||_{\ell^{2}_{J\times J}}}{1+(\acute{\lambda}_{n,b})^{\top}\hat{g}_{n,b}(\theta_{n,*})}$$
and notice that
\begin{align}
|E_{n,4}(\theta_{n,*})| \leq &\Bigg(\frac{\sup_{\acute{\lambda}_{n,b} \in \acute{\Lambda}_{n}}||\acute{\lambda}_{n,b}||_{\ell^{2}_{B_{n}}}\max_{1\leq i \leq n}||g_{b}(W_{i},\theta_{n,*})||_{\ell^{2}_{J\times J}}}{1+(\acute{\lambda}_{n,b})^{\top}\hat{g}_{n,b}(\theta_{n,*})}\Bigg) \label{ubvcov} \\
&\times \Bigg(\frac{1}{n}\sum_{i=1}^{n}||g(W_{i},\theta_{n,*})g(W_{i},\theta_{n,*})^{\top}||_{\ell^{2}_{J\times J}}\Bigg).\nonumber
\end{align}
The numerator of (\ref{ubvcov}) is $o_{p}(1)$ along $\{(\theta_{n,*}, F_{n}): n \geq 1\}$ by Lemma \ref{LP3}, Lemma \ref{owen112}, and the weak law of large numbers for triangular arrays of row-wise i.i.d. random variables. The denominator is $1+o_{p}(1)$ along $\{(\theta_{n,*}, F_{n}): n\geq 1\}$ because
\begin{align}
|(\acute{\lambda}_{n,b})^{\top}\hat{g}_{n,b}(\theta_{n,*})| \leq \sup_{\acute{\lambda}_{n,b} \in \acute{\Lambda}_{n}}||\acute{\lambda}_{n,b}||_{\ell^{2}_{B_{n}}}||\hat{g}_{n,b}(\theta_{n,*})||_{\ell^{2}_{B_{n}}}=O_{p}(n^{-\frac{1}{2}})O_{p}(n^{-\frac{1}{2}})=o_{p}(1)
\end{align}
where the first inequality is the Cauchy-Schwarz inequality and the definition of least upper bound, the first equality holds by Lemma \ref{LP3} and a Liaponuv CLT for triangular arrays of row-wise i.i.d. random variables. Note we do not need to recenter as $\acute{\mathcal{B}_{\varrho_{n}}} \subseteq C$ w.p.a. 1 as $n\rightarrow\infty$. Thus, $\sum_{i=1}^{n}\acute{p}_{i}||g(W_{i},\theta_{n,*})g(W_{i},\theta_{n,*})^{\top}||_{\ell^{2}_{J\times J}}=O_{p}(1)+o_{p}(1)=O_{p}(1)$ along $\{(\theta_{n,*}, F_{n}): n \geq 1\}$.
\end{proof}
\end{lemma}
The next lemma establishes an ordering of the restricted and unrestricted estimator of the moments that occurs with probability approaching 1 when the moments are nonnegatively correlated.
\begin{lemma}\label{LPC2}
Let $\mathcal{M}$ be as in~(\ref{eq - local alternatives family}). For any $\{(\theta_{n,*}, F_{n}): n \geq 1\} \in \mathcal{M}$ and $j \in \mathcal{J}$, $$\lim_{n\rightarrow+\infty}P_{F_{n}}\big(\hat{g}_{n,j}(\theta_{n,*}) \leq \acute{g}_{n,j}(\theta_{n,*})\big)=1.$$
\begin{proof}
We outline the steps to the proof and then provide details.
\paragraph{Outline.} The first step shows that $\hat{g}_{n,j}(\theta_{n,*}) - \acute{g}_{n,j}(\theta_{n,*})=\acute{\lambda}_{n,b}^{\top}\sum_{i=1}^{n}\acute{p}_{i}g_{b}(W_{i},\theta_{n,*})g_{j}(W_{i},\theta_{n,*})$ for any $j \in \mathcal{J}$. The second step uses the sign restrictions on the elements in $\{\Omega(\theta_{n,*},F_{n}): n \geq 1\}$ and on $\acute{\lambda}_{n,b}$ to conclude the result.
\paragraph{Step 1.}
Fix $j \in \mathcal{J}$ and $\{(\theta_{n,*}, F_{n}): n \geq 1\} \in \mathcal{M}$ arbitrarily. Using a derivation similar to that presented in Lemma \ref{LP4}, we have that
\begin{align}
\hat{g}_{n,j}(\theta_{n,*}) - \acute{g}_{n,j}(\theta_{n,*}) &= \sum_{i=1}^{n}\bigg(\frac{1}{n}-\acute{p}_{i}\bigg)g_{j}(W_{i},\theta_{n,*}) \\
&= \acute{\lambda}_{n,b}^{\top}\sum_{i=1}^{n}\acute{p}_{i}g_{b}(W_{i},\theta_{n,*})g_{j}(W_{i},\theta_{n,*}) \label{decomp1}
\end{align}
\paragraph{Step 2.} Let $\Xi(\theta_{n,*},F_{n})$ denote the $B_{n}\times 1$ vector of covariances between $g_{j}(W_{i},\theta_{n,*})$ and the elements of $g_{b}(W_{i},\theta_{n,*})$. From (\ref{decomp1}), we can write
\begin{align}
\hat{g}_{n,j}(\theta_{n,*}) - \acute{g}_{n,j}(\theta_{n,*}) = \acute{\lambda}_{n,b}^{\top}\Bigg(\sum_{i=1}^{n}\acute{p}_{i}g_{b}(W_{i},\theta_{n,*})g_{j}(W_{i},\theta_{n,*})-\Xi(\theta_{n,*},F_{n})+\Xi(\theta_{n,*}, F_{n})\Bigg)
\end{align}
From Lemma \ref{LPC1}, we have that $\sum_{i=1}^{n}\acute{p}_{i}g_{b}(W_{i},\theta_{n,*})g_{j}(W_{i},\theta_{n,*})-\Xi(\theta_{n,*},F_{n})=o_{p}(1)$ along $\{(\theta_{n,*}, F_{n}): n \geq 1\}$. Hence,
\begin{align}
\hat{g}_{n,j}(\theta_{n,*}) - \acute{g}_{n,j}(\theta_{n,*}) = \acute{\lambda}_{n,b}^{\top}\Bigg(o_{p}(1)+\Xi(\theta_{n,*}, F_{n})\Bigg).
\end{align}
Since the the Karusch-Kuhn-Tucker conditions dictate that $\acute{\lambda}_{n,b,k}\leq 0$ for each $k \in \{1,...,B_{n}\}$ and $\Xi(\theta_{n,*}, F_{n})$ is a vector of nonnegative terms, we have that $\hat{g}_{n,j}(\theta_{n,*}) - \acute{g}_{n,j}(\theta_{n,*}) \leq 0$ with probability approaching 1 as $n\rightarrow\infty$ along $\{(\theta_{n,*},F_{n}): n \geq 1\}$.
\end{proof}
\end{lemma}
The next result provides an ordering of the elementwise moment selection functions that occurs with probability 1 under nonnegative correlation. Let $a,b \in \mathbb{R}_{[\pm \infty]}^{J}$, the relation $a\succsim b$ means that $a_{j} \geq b_{j}$ for each $j \in \mathcal{J}$.
\begin{lemma}\label{GMSorder}
Let $\mathcal{M}$ be as in~(\ref{eq - local alternatives family}). Then $\forall \ \{(\theta_{n,*},F_{n}): n \geq 1\} \in \mathcal{M}$,
\begin{align*}
\lim_{n\rightarrow\infty}P_{F_{n}}\Big(\varphi^{(1)}(\acute{\xi}_{n}(\theta_{n,*}), \hat{\Omega}_{n}(\theta_{n,*})) \succsim \varphi^{(1)}(\hat{\xi}_{n}(\theta_{n,*}), \hat{\Omega}_{n}(\theta_{n,*}))\Big)=1.
\end{align*}
\begin{proof}
The proof uses the direct method. We present an outline and then the steps in detail.
\paragraph{Outline.} The proof involves two short steps. Step 1 shows that an ordering of the restricted and unrestricted estimators implies an ordering of the moment selection functions. Step 2 invokes Lemma \ref{LPC2} to establish the result.
\paragraph{Step 1.} Since $\acute{\xi}_{n,j}(\theta_{n,*})$ and $\hat{\xi}_{n,j}(\theta_{n,*})$ are just $\acute{g}_{n,j}(\theta_{n,*})$ and $\hat{g}_{n,j}(\theta_{n,*})$, respectively, scaled by common positive factor $\hat{\sigma}_{n,j}^{-1}(\theta_{n,*})\kappa_{n}^{-1}n^{\frac{1}{2}}$, it follows that
\begin{align}
\Big\{\acute{g}_{n}(\theta_{n,*}) \succsim \hat{g}_{n}(\theta_{n,*})\Big\} &\subseteq \Big\{\acute{\xi}_{n}(\theta_{n,*}) \succsim \hat{\xi}_{n}(\theta_{n,*})\Big\}\\
&\subseteq \Big\{\varphi^{(1)}(\acute{\xi}_{n,j}(\theta_{n,*}),\hat{\Omega}_{n}(\theta_{n,*})) \succsim \varphi^{(1)}(\hat{\xi}_{n,j}(\theta_{n,*}),\hat{\Omega}_{n}(\theta_{n,*}))\Big\}
\end{align}
where the second set inclusion holds because $ \varphi^{(1)}(\xi,\Omega)$ is nondecreasing in $\xi$.
\paragraph{Step 2.} Step 1 and the monotonicity of probability measures yield
\begin{align*}
P_{F_{n}}\Big(\acute{g}_{n}(\theta_{n,*}) \succsim \hat{g}_{n}(\theta_{n,*})\Big) \leq P_{F_{n}}\Big(\varphi^{(1)}(\acute{\xi}_{n,j}(\theta_{n,*}),\hat{\Omega}_{n}(\theta_{n,*})) \succsim \varphi^{(1)}(\hat{\xi}_{n,j}(\theta_{n,*}),\hat{\Omega}_{n}(\theta_{n,*}))\Big)
\end{align*}
for each $n \geq 1$. So for any $\{(\theta_{n,*},F_{n}): n \geq 1\} \in \mathcal{M}$, we invoke Lemma \ref{LPC2} to conclude that
\begin{align*}
\lim_{n\rightarrow+\infty}P_{F_{n}}\Big(\varphi^{(1)}(\acute{\xi}_{n,j}(\theta_{n,*}),\hat{\Omega}_{n}(\theta_{n,*})) \succsim \varphi^{(1)}(\hat{\xi}_{n,j}(\theta_{n,*}),\hat{\Omega}_{n}(\theta_{n,*}))\Big)=1.
\end{align*}
\end{proof}
\end{lemma}
\section{Further Theoretical Discussion}
\subsection{Other GMS Functions}\label{OtherGMS}
\subsubsection{GMS Assumptions}
We restate the GMS Assumptions in \cite{andrews2010inference} to aid discussion in the next subsection. We restrict $\Omega \in \varPsi_{2}$ in the statements to accord with the assumptions imposed on $\mathcal{F}$.
\begin{msassumption}
For each $j \in \mathcal{J}$, 1. $\varphi_{j}(\xi,\Omega)=0$ is continuous for all $(\xi,\Omega) \in \mathbb{R}_{[+\infty]}^{J}\times \varPsi_{2}$ with $\xi_{j}=0$ and  2. $\varphi_{j}(\xi,\Omega)=0$ for all $(\xi,\Omega) \in \mathbb{R}_{[+\infty]}^{J}\times \varPsi_{2}$ with $\xi_{j}=0$.
\end{msassumption}
\begin{msassumption}
$\kappa_{n}\rightarrow+\infty$ as $n\rightarrow+\infty$
\end{msassumption}
\begin{msassumption}
For each $j \in \mathcal{J}$, $\varphi_{j}(\xi,\Omega) \rightarrow +\infty$ as $(\xi,\Omega)\rightarrow (\xi_{*},\Omega_{*})$ for any $(\xi_{*},\Omega_{*}) \in \mathbb{R}_{[+\infty]}^{J}\times \overline{\varPsi}_{2}$ with $\xi_{*,j}=+\infty$.
\end{msassumption}
\begin{msassumption}
$\kappa_{n}^{-1}n^{\frac{1}{2}}\rightarrow+\infty$ as $n\rightarrow+\infty$
\end{msassumption}
\begin{msassumption6}
For each $j \in \mathcal{J}$, $\varphi_{j}(\xi,\Omega) \geq 0$ for all $(\xi,\Omega) \in \mathbb{R}^{J}_{[+\infty]}\times \varPsi_{2}$.
\end{msassumption6}
\begin{msassumption7}
For each $j \in \mathcal{J}$, $\varphi_{j}(\xi,\Omega) \geq \min\{0,\xi_{j}\}$ for all $(\xi,\Omega) \in \mathbb{R}_{[+\infty]}\times \varPsi_{2}$.
\end{msassumption7}
We do not list Assumption GMS 5 because it is required to compare moment selection and subsampling critical values, a topic we do not discuss formally in our paper. GMS2 and GMS4 combine to form Assumption K in the paper.
\subsubsection{Alternative Choices of $\varphi$}
The main theoretical results in the paper assumed that $\varphi=\varphi^{(1)}$, but there are many other choices for $\varphi$. These include $\varphi_{j}^{(2)}(\xi, \Omega) = \psi(\xi_{j}), \ \varphi^{(3)}_{j}(\xi, \Omega)=\max(0, \xi_{j}), \ \varphi_{j}^{(4)}(\xi, \Omega)=\xi_{j}$, where $\psi(\cdot)$ is nondecreasing and satisfies $\psi(x)=0$ if $x \leq a_{L}$, $\psi(x) \in [0, \infty]$ if $x \in (a_{L}, a_{U})$, and $\psi(x) = \infty$ if $x \geq a_{U}$ \citep{andrews2010inference}. Another choice is the modified MSC choice defined as
\begin{align*}
\varphi^{(5)}_{j} =
\begin{cases}
0 \quad &\text{if $c_{j}(\xi,\Omega) = 1$} \\
\infty \quad &\text{if $c_{j}(\xi,\Omega)=0$}
\end{cases}
\end{align*}
where $c:=(c_{1}(\xi,\Omega),...,c_{J}(\xi,\Omega))'$ solves the integer program $\min_{c \in \{0,1\}^{J}}\{S(-c^{\top}\xi,\Omega)-\zeta(|c|)\}$ for some increasing function $\zeta(\cdot)$.\footnote{Note if $c_{j}=0$ and $\xi_{j}=+\infty$, the convention is adopted that $c_{j}\xi_{j}=0$.} Modified MSC uses the information embedded in the off-diagonals of the correlation matrix $\Omega$ in a computationally expensive way, whereas $\varphi^{(k)}$, $k \in \{1,2,3,4\}$, does not \citep{andrews2010inference}.
\paragraph{}
Our decision to focus on $\varphi=\varphi^{(1)}$ is essentially without loss of generality because the results can be generalized to any $\varphi$ that satisfies the assumptions of \cite{andrews2010inference}. To see this, we first recall that Lemma \ref{L2CS} implies that for any $r>0$,
\begin{align*}
\liminf_{n\rightarrow+\infty}\inf_{(\theta,F) \in \mathcal{F}_{+}}P_{F}\Big(\big | \big | (\acute{\xi}_{n}(\theta),\hat{\Omega}_{n}(\theta))-(\hat{\xi}_{n}(\theta),\hat{\Omega}_{n}(\theta)) \big | \big |_{\ell^{2}_{J \times \varPsi_{2}}}<r\Big)=1
\end{align*}
and Lemma \ref{LP4} implies that for any $r>0$, $\{(\theta_{n,*},F_{n}): n \geq 1\} \in \mathcal{H}$ that satisfies Assumptions LA1 and LA2,
\begin{align*}
\lim_{n\rightarrow+\infty}P_{F_{n}}\Big(\big | \big | (\acute{\xi}_{n}(\theta_{n,*}),\hat{\Omega}_{n}(\theta_{n,*}))-(\hat{\xi}_{n}(\theta_{n,*}),\hat{\Omega}_{n}(\theta_{n,*})) \big | \big |_{\ell^{2}_{J \times \varPsi_{2}}}<r\Big)=1,
\end{align*}
where $||\cdot ||\ell^{2}_{J \times \varPsi_{2}} := (||\cdot ||_{\ell^{2}_{J}}+||\cdot ||_{\ell^{2}_{J\times J}})^{\frac{1}{2}}$ in both statements. These convergence results are enough to extend our asymptotic size and limiting local power results to any $\varphi$ that satisfies Assumptions GMS1--4 with appropriate modifications to notation. Indeed, we can replicate the arguments in the proofs of Theorem 1 and Theorem 2 in \cite{andrews2010inference} (with modifications to notation). The same comment applies to Theorem 4 in their paper because the use of a constrained estimator does not challenge the validity of GMS7.
\paragraph{}
The ordering of the local power functions also extends. The weak ordering of the local power functions extends to $\varphi^{(k)}$, $k \in \{1,2,3,4\}$, because the result only requires that $\varphi_{j}(\xi,\Omega)$ be nondecreasing in $\xi$. However, the strict ordering does not apply under $\varphi^{(3)}$ because we require $\varphi_{j}(\xi,\Omega) \geq 0$ for each $j \in \mathcal{J}$ in order to invoke Part 2 of Assumption 5, effectively restricting attention to those that satisfy GMS5. We do not view this to be a serious limitation, especially given that $\varphi=\varphi^{(1)}$ is the recommended by~\cite{andrews2012inference}. A final technical point is that to generalize Theorem \ref{result3}, we must replace the event $\{\acute{\Upsilon}_{n}(\theta_{n,*}) \subsetneq \hat{\Upsilon}_{n}(\theta_{n,*})\}$ in the statement of Theorem \ref{result3} with a more general event $\big\{\varphi(\acute{\xi}_{n}(\theta_{n,*}),\hat{\Omega}_{n}(\theta_{n,*}))\succ\varphi(\hat{\xi}_{n}(\theta_{n,*}),\hat{\Omega}_{n}(\theta_{n,*})))\big\}$ because the first uses the specific form of $\varphi^{(1)}$.\footnote{For any vectors $a,b \in \mathbb{R}_{[+\infty]}^{J}$, the relation $a\succ b$ means that $a_{j} \geq b_{j}$ for each $j$ with at least one strict inequality.}
\subsection{Elaboration on Remark 1}
In Remark 1, we state that one can `fully constrain' the CMS procedure. This involves use of the empirical likelihood estimator of the covariance matrix $\acute{\Sigma}_{n}(\theta)$ and correlation matrix $\acute{\Omega}_{n}(\theta) = \acute{D}_{n}^{-\frac{1}{2}}(\theta)\acute{\Sigma}_{n}(\theta)\acute{D}_{n}^{-\frac{1}{2}}(\theta)$. Lemma \ref{L2CS} and \ref{unifcov} imply that for any $r>0$,
\begin{align*}
\liminf_{n\rightarrow+\infty}\inf_{(\theta,F) \in \mathcal{F}_{+}}P_{F}\Big(\big | \big | (\acute{\xi}_{n}^{FC}(\theta),\acute{\Omega}_{n}(\theta))-(\hat{\xi}_{n}(\theta),\hat{\Omega}_{n}(\theta)) \big | \big |_{\ell^{2}_{J \times \varPsi_{2}}}<r\Big)=1.
\end{align*}
So simple modifications of the arguments in the proof of Theorem 1 establish validity of fully-constrained confidence sets.
\paragraph{}
Similarly, Lemma \ref{LP4} and \ref{LPC1} allow us to conclude that for any $r>0$ and any $\{(\theta_{n,*},F_{n}): n \geq 1\} \in \mathcal{H}$ that satisfies Assumption LA1 and LA2,
\begin{align*}
\lim_{n\rightarrow+\infty}P_{F_{n}}\Big(\big | \big | (\acute{\xi}_{n}^{FC}(\theta_{n,*}),\acute{\Omega}_{n}(\theta_{n,*}))-(\hat{\xi}_{n}(\theta_{n,*}),\hat{\Omega}_{n}(\theta_{n,*})) \big | \big |_{\ell^{2}_{J \times \varPsi_{2}}}<r\Big)=1
\end{align*}
implying that adjustments to the proof of Theorem 2 extend the limiting local power results to the fully constrained case.  It is difficult to establish a general ordering between $\acute{\xi}_{n}^{FC}(\cdot)$ and $\hat{\xi}(\cdot)$ so it is unclear whether the finite-sample $n^{-\frac{1}{2}}$-local power comparisons hold in the fully constrained case. The consistency against distant alternatives also extends because GMS7 and the use of the constrained estimator imply that $\varphi_{j}(\acute{\xi}_{n}^{FC}(\theta_{n,*}),\acute{\Omega}_{n}(\theta_{n,*})\geq 0$ for each $j \in \mathcal{J}$ so we can similarly bound the fully constrained critical value from above by the plug-in asymptotic critical value.
\section{Further Simulation Details}
\subsection{Outline of RMS}\label{RMSdetails}
\cite{andrews2012inference} present a modification of the GMS procedure. From an implementation standpoint, the approach is basically the same as GMS except that:
\begin{enumerate}
\item They replace $\kappa_{n}$ with a data-driven tuning parameter $\hat{\kappa}:=\kappa(\hat{\delta}_{n}(\cdot))$, where $\hat{\delta}_{n}(\cdot)$ is the minimum off-diagonal element of $\hat{\Omega}_{n}(\cdot)$.
\item They add a size-correction factor $\hat{\eta}:=\eta_{1}(\hat{\delta}_{n}(\cdot))+\eta_{2}(J)$ to the GMS critical value that results from using the tuning parameter $\hat{\kappa}$.
\end{enumerate}
The need to size-correct reflects the fact that $\hat{\kappa}$ is a finite constant plus $o_{p}(1)$ rather than a divergent sequence and the method of data-driven tuning parameters is referred to as $\kappa$-auto \citep{andrews2012inference}.
\subsection{Outline of the Two-Step Procedure}\label{RSWdetails}
We outline the two-step procedure of \cite{romano2014practical} to aid understanding of the simulation results. The procedure needs some modification because we test $H_{0}: \mu \in \mathbb{R}_{+}^{J}$ rather than $H_{0}: \mu \in \mathbb{R}_{-}^{J}$. To this end, let $\mathcal{F} = \{F= N(\mu,\Sigma): (\mu,\Sigma) \in \mathbb{R}^{J}\times \varPsi_{2}\}$, $\mathcal{F}_{0} = \{F \in \mathcal{F}: \mu \in \mathbb{R}^{J}_{+}\}$, and assume that the correlation matrix $\Sigma$ is known. The following steps describe a level $\alpha$ test for $H_{0}: F_{0} \in \mathcal{F}_{0}$ vs. $H_{1}: F_{0} \in \mathcal{F} \setminus \mathcal{F}_{0}$ using a random sample $\{W_{i}:i =1,...,n\}\overset{iid}{\sim} F_{0}$:
\begin{enumerate}
\item Compute the test statistic $T_{n}= S(\sqrt{n}\hat{g}_{n},\hat{\Sigma}_{n})$.
\item Generate bootstrap samples $\{W_{i,b}^{*}: i =1,...,n\}$, $b=1,...,B$, by sampling with replacement from the data $\{W_{i}:i =1,...,n\}$.
\item Compute a lower confidence rectangle $M_{n}(\beta) = \{\mu \in \mathbb{R}^{J}: \min_{1 \leq j \leq J}[\hat{\sigma}_{n,j}^{-1}\sqrt{n}(\mu_{j}-\hat{g}_{n,j})]\geq K^{-1}_{n}(\beta)\}$, where $K^{-1}_{n}(\beta)$ is the $\beta$-quantile of $\{\min_{1 \leq j \leq J}[(\hat{\sigma}_{n,j,b}^{*})^{-1}\sqrt{n}(\hat{g}_{n,j}-\hat{g}_{n,j,b}^{*})]: b=1,...,B\}$, $\hat{g}_{n,j,b}^{*} = n^{-1}\sum_{i=1}^{n}W_{i,j,b}^{*}$, and $\hat{\sigma}_{n,b,j}^{*} = n^{-1}\sum_{i=1}^{n}(W_{i,j,b}^{*}-\hat{g}_{n,j,b}^{*})^{2}$ for $b=1,...,B$ and $j=1,...,J$. This determines which components of $\mu$ are `positive'.
\item Compute bootstrap test statistics $\{T_{n,b}^{*}:b=1,...,B\}$, where $T_{n,b}^{*}=S\big((\hat{D}_{n,b}^{*})^{-\frac{1}{2}}\sqrt{n}(\hat{g}_{n,b}^{*}-\hat{g}_{n})+(\hat{D}_{n,b}^{*})^{-\frac{1}{2}}\sqrt{n}\lambda^{*}, \hat{\Sigma}_{n,b}^{*}\big)$, and $\lambda \in \mathbb{R}_{+}^{J}$ with $\lambda_{j}^{*} = \max\{0, n^{-\frac{1}{2}}\hat{\sigma}_{n,j} K^{-1}_{n}(\beta)+\hat{g}_{n,j}\}$ for $j=1,...,J$.
\item Compute the critical value $c_{n}^{RSW}(1-\alpha+\beta)$, which is defined as the $1-\alpha+\beta$ quantile of $\{T_{n,b}^{*}: b=1,...,B\}$.
\item Reject $H_{0}$ at significance level $\alpha$ if $T_{n}>c_{n}^{RSW}(1-\alpha+\beta)$ and $M_{n}(\beta)\nsubseteq \mathbb{R}^{J}_{+}$.
\end{enumerate}
Following the choice of \cite{romano2014practical}, we set $\beta = \alpha/10$ for all simulations.
\subsection{MNRP Corrections}\label{scdetails}
We outline the MNRP corrections used for finite-sample $n^{-\frac{1}{2}}$ local power results, an essential ingredient for a fair comparison of the procedures under the alternative. For a given pair $(J,\Omega)$, let $p^{RSW}_{n,R}\equiv p^{RSW}_{n,R}(J,\Omega)$ denote the maximum null rejection probability for the two-step procedure of \cite{romano2014practical} based on $R$ Monte Carlo simulations and sample size $n$. For $t \in \{GMS,CMS,RMS\}$, the random variable $\delta_{n,R}^{t}\equiv\delta_{n,R}^{t}(J,\Omega)$ is the $(1-p^{RSW}_{n,R})$-empirical quantile based on the simulated process $\{T_{n,r}^{*,t}-c_{n,r}^{*,t}: r=1,...,R\}$, where $(T_{n,r}^{*,t},c_{n,r}^{*,t})$ correspond to the mean vector $\mu^{*,t}$ that maximizes null rejection probability for test $t$. We add $\delta_{n,R}^{t}$ to the corresponding critical value in the power results to ensure that all procedures have the same MNRP. Indeed, by construction
\begin{align*}
\hat{P}_{R,t}^{*}(T_{n,r}^{*,t}-c_{n,r}^{*,t} \leq \delta_{n,R}^{t})= p_{n,R}^{RSW} \quad \forall \ t \in \{GMS,CMS,RMS\}
\end{align*}
where $\hat{P}_{R,t}^{*}(\cdot)$ denotes the simulation distribution of $\{T_{n,r}^{*,t}-c_{n,r}^{*,t}: r=1,...,R\}$.
\end{document}